\RequirePackage{snapshot}
\documentclass{article}
\usepackage[margin=1.4in]{geometry}

\usepackage{amsmath}
\usepackage{amsopn}
\usepackage{amsthm}
\usepackage{amssymb}
\usepackage{float}
\usepackage{graphics}
\usepackage{pdfpages}
\usepackage{psfrag}
\usepackage{graphicx,url,colortbl}
\usepackage{xspace}
\usepackage{algorithmic,algorithm}
\usepackage{authblk}

\newtheorem{definition}{Definition}[section]
\newtheorem{example}{Example}[section]
\newtheorem{theorem}{Theorem}[section]
\newtheorem{lemma}{Lemma}[section]
\newtheorem{proposition}{Proposition}[section]
\newtheorem{corollary}{Corollary}[section]

\newtheorem{remark}{Remark}

\newcommand{\norme}[1]{\left\Vert #1\right\Vert}

\def\bbR{\mathbb{R}}

\def\bbB{\mathbb{B}}

\newcommand{\op}{\mathrm{op}}
\newcommand{\BL}{\mathrm{BL}}
\newcommand{\LL}{\mathrm{L}}
\newcommand{\proj}{\mathrm{p}}
\newcommand{\diam}{\mathrm{diam}}
\newcommand{\diamnoise}{\mathrm{D}}
\newcommand{\Lip}{\mathrm{Lip}}
\renewcommand{\d}{\mathrm{d}}
\newcommand{\NN}{\mathrm{N}}

\newcommand{\n}{\mathbf{n}}

\newcommand{\Pow}{\mathrm{Pow}}
\newcommand{\Wass}{\mathrm{W}}

\newcommand{\B}{\mathrm{B}}

\usepackage{soul}

\begin{document}

\title{Robust Geometry Estimation using the Generalized Voronoi Covariance Measure\footnote{This research has been supported in part by the ANR grants DigitalSnow ANR-11-BS02-009, KIDICO ANR-2010-BLAN-0205 and TopData ANR-13-BS01-0008.}
\footnote{An extended abstract of this work has been presented at EuroCG 2014.}}

\author[1,2]{Louis Cuel}
\author[2]{Jacques-Olivier Lachaud }
\author[1,3]{Quentin M\'erigot}
\author[2]{Boris Thibert}
\affil[1]{Laboratoire Jean Kuntzman, Universit\'e Grenoble-Alpes, France}
\affil[2]{Laboratoire de Math\'ematiques (LAMA), Universit\'e de Savoie, France}
\affil[3]{CNRS}
\maketitle
\begin{abstract}
The Voronoi Covariance Measure of a compact set $K$ of $\bbR^d$ is a tensor-valued measure that encodes geometrical information on $K$ and which is known to be resilient to Hausdorff noise but sensitive to outliers. In this article, we generalize  this notion to any distance-like function $\delta$ and define the $\delta$-VCM. Combining the VCM with the distance to a measure and also with the witnessed-$k$-distance, we get a provably good tool for normal estimation that is resilient to Hausdorff noise and to outliers. 
We present  experiments showing the robustness of our approach for normal and curvature estimation and sharp feature detection.
\end{abstract}


\section{Introduction}
The estimation of normal vectors and other differential quantities,
such as principal curvature directions and direction of sharp
features, has many application in computer graphics, geometry
processing or computational geometry. The output of most acquisition
devices is a raw point cloud, often corrupted with noise. For
applications such as surface reconstruction, it is therefore very
useful to be able to perform these estimations directly on such data.
Many methods have been proposed for estimating normal and curvature
directions directly from a point cloud, such as principal component
analysis \cite{Hoppe}, local fitting of polynomials quantities
\cite{Pouget}, integral estimates \cite{pottmann_integral_2009},
moving least squares \cite{MLS-outliers}, statistical normal
clustering \cite{sgp-2012}, to name but a few. Our work is in 
the continuity of Voronoi-based methods for normal estimations, which have
been introduced in \cite{poles} and refined by many authors
\cite{dey2006normal,voronoi-based,vcm}. These methods are robust to
Hausdorff noise but not outliers. Our contribution is to generalize
the Voronoi-based approach, and combine it to the notion of distance
to a measure \cite{dist-measure,witnessed} so as to gain resilience to
outliers.

\paragraph{Voronoi-based estimation}
Classical principal component analysis methods try to estimate normal vectors
by fitting a tangent plane.  In contrast, Voronoi-based methods try to
fit the normal cones to the underlying shape, either geometrically
\cite{poles,dey2006normal} or  using covariance matrices
of Voronoi cells \cite{voronoi-based,vcm}.
In \cite{voronoi-based}, the authors estimate the normal at a data
point in two steps. They start by considering the covariance matrix of
the union of Voronoi cells of nearby points, with respect to the
center of mass of this union. Then, the normal estimation is given by
the the eigenvector corresponding to the largest eigenvalue of this
matrix. In \cite{vcm}, the authors improved this method by changing the domain of
integration and the averaging process. The authors showed that it is
possible to associate to any compact set $K$ a tensor-valued measure,
which they called the \textit{Voronoi Covariance Measure} of $K$
(VCM). Then, they proved that this notion is Hausdorff-stable, in the
sense that if two compact sets are close in the Hausdorff sense, their
VCM are also close to each other. The VCM of a smooth surface encodes
the normal vector field to this surface; this stability result
therefore ensures that this information is kept in the VCM of a point
cloud which is Hausdorff-close to the surface.

\paragraph{Distance to a measure} All the aforementioned Voronoi-based
methods for normal estimation rely directly or indirectly on the
notion of distance function. Recall that the distance function to a
compact subset $K$ of $\bbR^d$ is the function on $\bbR^d$ defined by
the formula $\d_K(x) = \min_{p\in K} \norme{p-x}$. The stability
properties of geometric inference methods based on the distance
function is derived from the stability property of the distance
function, namely that if $K$ and $L$ are Hausdorff-close, then $\d_K$
and $\d_L$ are uniformly close.  Unfortunately, geometric data is
usually corrupted with outliers, and the hypothesis of Hausdorff noise
is not realistic in practice. To make things worse, even the addition
of a single outlier to a point cloud can perturb the distance function
to that point cloud drastically. To counter this difficulty a robust
variant of the notion of distance function to a compact set, called
the \emph{distance to the measure}, was proposed in
\cite{dist-measure}. This new definition of distance is able to cope
with the presence of some outliers. Moreover, the distance to a
measure is \emph{distance-like}: this means that it possesses the
regularity properties of distance functions to compact sets which
makes them amenable to geometric inference.

\paragraph{Contributions} 
\begin{itemize}
\item We extend the definition of Voronoi covariance measure of a 
compact set. More precisely, we associate to any \emph{distance-like} function $\delta$, 
a family of tensor-valued measures called the $\delta$-Voronoi covariance 
measure ($\delta$-VCM). 

\item We show the stability of the $\delta$-VCM. Our main general theorem (Theorem \ref{thm:main-stability}) 
asserts  that if a distance-like function $\delta$ well approximates the distance function to a compact set $K$, 
then the $\delta$-VCM is close to the VCM of $K$. When applied to a point cloud $P$ approximating a surface 
$S$ of $\bbR^3$, this implies that one can recover the normal vectors of $S$ very accurately (Proposition \ref{prop:stability-distance-to-a-measure}). 
This estimation is Hausdorff stable and  robust to outliers.

\item The distance to a measure of a point cloud being not computable in practice, we replace it by the witnessed $k$-distance \cite{witnessed}. We show that the associated VCM still well approximates the VCM of the underlying surface (Proposition \ref{prop:stability-witness-distance}), which opens the door to practical computations. 

\item We show on various examples that the $\delta$-VCM provides a
  robust normal estimator resilient to Hausdorff noise and to
  outliers. For the experiments, we introduce another distance-like function, the \textit{median-$k$-distance}. Although we do not have any guarantee for the VCM based on the median-$k$-distance, it gives very robust estimations in practice.
   We also use the $\delta$-VCM to estimate curvatures and sharp
  features. Our estimators improve the results based on the VCM
  \cite{vcm} or on the Jet Fitting \cite{Pouget}, even when there are
  no outliers. They are also compared favorably to the
  normal classifier of Boulch {\em et al.} \cite{sgp-2012}.

 
\end{itemize}

\paragraph{Notation} In the following we denote by $\|.\|$ the
Euclidean norm of $\bbR^d$. We will call \emph{tensor} a square
matrix. The tensor product $\mathbf{v} \otimes \mathbf{w}$ of two
vectors $\mathbf{v}, \mathbf{w}$ is the $d \times d$ matrix whose
entries are defined by $(\mathbf{v} \otimes \mathbf{w})_{ij} =
v_iw_j$. The $d$-dimensional ball of center $x$ and radius $r$ is
denoted by $\bbB(x,r)$.

\section{$\delta$-Voronoi Covariance Measure}
As remarked in \cite{dist-measure}, many inference results rely only
on two properties of the distance function to a compact set. Any
functions satisfying these properties is called \emph{distance-like};
in particular, the usual distance function to a compact set is
distance-like.
The goal of this section is to extend the definition of Voronoi
Covariance Measure of a compact set, introduced in \cite{vcm} for the
purpose of geometric inference. We associate to any distance-like
function $\delta$ a tensor-valued measure called the
\emph{$\delta$-Voronoi covariance measure} or $\delta$-VCM.  Our main
theoretical result is Theorem \ref{thm:main-stability}, which asserts
in a quantitative way that if a distance-like function $\delta$ is
uniformly close to the distance function to a compact set, then the
$\delta$-VCM is close to the VCM of this compact set. Informally, this
theorem shows that one can recover geometric information about a
compact set using an approximation of its distance function by a
distance-like function.







\subsection{$\delta$-Voronoi Covariance Measure}\label{subsection:defVCM}

In this paragraph, we introduce the definitions necessary for
the precise statement of the main theorem. We start by the definition
of distance-like function, following \cite{dist-measure}. Note that
we used the remark following Proposition~3.1 in \cite{dist-measure} to
drop the $1$-Lipschitz assumption, which follows from the two other
assumptions.

\begin{definition}[Distance-like function]\label{def:dist-like}
  A function $\delta:\bbR^d \to \bbR^+$ is called \emph{distance-like} if
\begin{itemize}
\item $\delta$ is proper, i.e. $\lim_{\|x\| \to \infty} \delta(x)= \infty$.
\item $\delta^2$ is $1$-semiconcave, that is $\delta^2(.) - \|.\|^2 $
  is concave.
\end{itemize}
\end{definition}

The typical examples of distance-like functions that we will consider
are power distances.  Given a point cloud $P$ and a family of
non-negative weights $(\omega_p)_{p\in P}$, we call \emph{power
  distance} to $P$ the distance-like function $\delta_P$ defined by
\begin{equation}
\label{eq:def-power-distance}
\delta_{P}(x) := \left(\min_{p\in P} \left( \| x - p \|^2 + \omega_p \right)\right)^{1/2}.
\end{equation}
Note that when the weights are all zero, the power distance is nothing
but the distance function to the point cloud $P$.

\begin{definition}[$\delta$-VCM]
  The {\em $\delta$-Voronoi Covariance Measure} is a tensor-valued
  measure. Given any non-negative \emph{probe function} $\chi$, i.e. an
  integrable function on $\bbR^d$, we associate a positive
  semi-definite matrix defined by
\begin{equation}
\mathcal{V}_{\delta,R}(\chi) := \int_{\delta^R}\mathbf{n}_{\delta}(x) \otimes \mathbf{n}_{\delta}(x) \mathbf{.} 
\chi \left( x - \mathbf{n}_{\delta}(x)  \right)\d x, 
\label{eq:vcm}
\end{equation}
where $\delta^R :=\delta^{-1}((-\infty,R])$ and where
$\mathbf{n}_{\delta}(x):= \frac{1}{2}\nabla\delta^2(x)$. Note that
this vector field $\mathbf{n}_{\delta}$ is defined at almost every point in $\bbR^d$ by the
$1$-concavity property of distance-like functions.
\end{definition}
The tensor $\mathcal{V}_{\delta,R}(\chi)$ is a convolution of a tensor with the function $\chi$ that localizes the calculation on the support of $\chi$. Intuitively, if $\chi$ is the indicatrix of a ball $\B\subset \bbR^d$, then the VCM $\mathcal{V}_{\delta,R}(\chi)$  is the integral of $\mathbf{n}_{\delta} \otimes \mathbf{n}_{\delta}$ over the set of points of $\delta^R$ that are ``projected'' into the ball $\B$. 



\begin{example}[$\delta$-VCM of a distance function]
When considering the distance function to a compact set $K$, the set
$\d_K^R$ coincides with the offset of $K$ with radius $R$, thus
explaining our choice of notation. Moreover, the vector field
$\mathbf{n}_{d_K}$ has an explicit expression in term of the
projection function $\proj_K$ on $K$, that is the application that
maps a point to its closest neighbor in $K$: 
\begin{equation}
  \mathbf{n}_{d_K}(x) =
  x-\proj_K(x).
\end{equation}
Comparing Eq.~\eqref{eq:vcm} to Eq.~(4.2) in \cite{vcm} and using
this remark, one sees that the $\d_K$-VCM defined here
matches the original definition of the VCM of $K$:
$$
\mathcal{V}_{d_K,R}(\chi) = \int_{K^R}(x-p_K(x)) \otimes (x-p_K(x)) \mathbf{.} \chi(p_K(x))\d x.
$$

Consider a smooth compact surface $S$ of $\bbR^3$, with exterior unit
normal $\n$. If $R$ is chosen small enough the following
  expansion holds as $r\to 0$, where $\norme{.}_\op$ is the operator norm \cite{vcm}:
\begin{equation}\label{eq:vcm-normal}
  \norme{\mathcal{V}_{d_S,R}(\mathbf{1}_{\B(p,r)}) -
    \frac{2\pi}{3} R^3r^2\left[
      \n(p) \otimes \n(p)
    \right]}_\op = O(r^3).
\end{equation}

This equation shows that one can recover local information about
differential quantities from the VCM of a surface. Note that curvature
information is also present in the matrix.
\end{example}

\begin{example}[$\delta$-VCM of a power distance]
  Let us give a closed form expression for the VCM of a power
  distance, which we will use in the computations.  Each power
  distance defines a decomposition of the space into a family of
  convex polyhedra, called power cells, on which the function
  $\delta_P^2$ is quadratic. The power cell of the point $p$ in $P$ is
  defined by
\begin{equation}
\Pow_{P}(p)=\{x\in \bbR^d; \forall q\in P, \|x-p\|^2 + \omega_p \leq \|x-q\|^2 + \omega_{q} \}.
\end{equation}
When the weight vector $\omega$ vanishes, we recover the notion of
Voronoi cell. The following Lemma generalizes Eq.~(2.1) in
\cite{vcm}, and shows that computing the VCM of a power distance
amounts to computing the covariance matrix of the intersection of each
power cell with a ball (see also Algorithm~\ref{algo:iterative}).
\end{example}

\begin{lemma}\label{lemma:power-distance-vcm}
  Let $(P,\omega)$ be a weighted point cloud. Given a probe function $\chi$,
\begin{equation}
\mathcal{V}_{\delta_{P},R}(\chi) = \sum_{p\in P} \chi(p)\ M_p,
\end{equation}
where  $M_p$ is the covariance matrix of $C_p:=\Pow_P(p) \cap \bbB(p, (R^2 - \omega_p)^{1/2})$,
\begin{equation}
M_p:=\int_{C_p} (x - p) \otimes (x-p) \d x.
\end{equation}
\end{lemma}
\begin{proof}
Since for every point $x$ in the interior of the power cell
$\Pow_P(p)$, $\mathbf{n}_{\delta_{P}}(x) = x-p$, hence
$\chi(x-\mathbf{n}_{\delta_{P}}(x))=\chi(p)$ is constant and one can
decompose the integral over power cells
\begin{equation*}
\mathcal{V}_{\delta_{P},R}(\chi) =
\sum_{p\in P} \chi(p) \int_{\Pow_P(p) \cap \delta_{P}^{R}} (x - p) \otimes (x-p) \d x.
\end{equation*}
Furthermore a computation gives
$\Pow_P(p) \cap \delta_{P}^{R} =
 \Pow_P(p) \cap 
\mathbb{B}\left(p, (R^2 - \omega_p)^{1/2}\right)$,
if we consider that $R^2 - \omega_p <0$ defines the empty ball.
\end{proof}

\begin{figure}[t]  
\centering
\includegraphics[width=.66\linewidth]{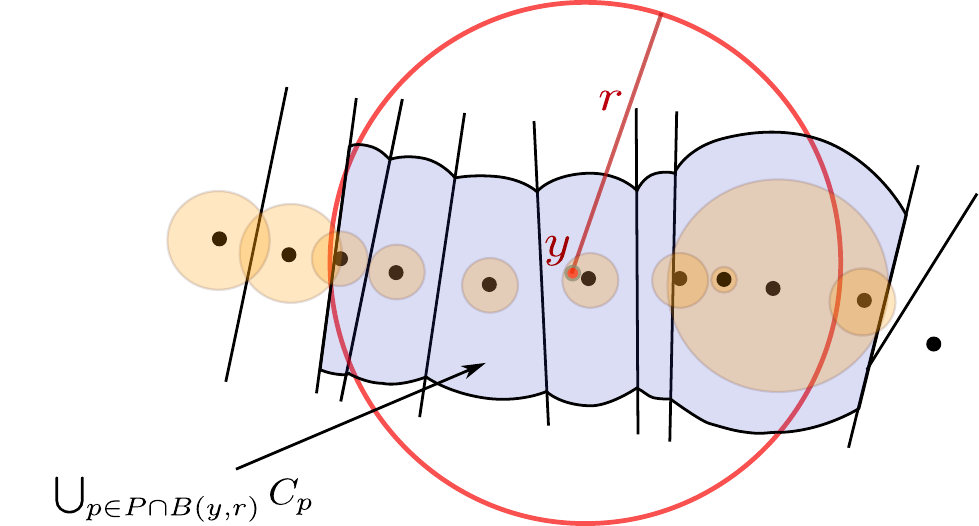}\hspace{0.9cm}
\caption{Integration domain of $\delta_P$-VCM. 
The weighted point cloud $P$ is represented by a union of circles whose radius are the weights. 
We suppose $\chi$ to be the indicatrix of the ball $\bbB(y,r)$ of center $y$ and radius $r$. The number of power cells is less than the cardinal of the point cloud since one cell is empty.}
\label{figure:power2}
\end{figure}

\subsection{Stability of the $\delta$-VCM}\label{subsection:stability}
We are now able to state our main theorem, which is a
generalization of the stability theorem for VCM proven in
\cite{vcm}. Our theorem asserts that if the distance function to a
compact set $K$ is well approximated by a distance-like function
$\delta$, then the VCM of $K$ is also well approximated by the
$\delta$-VCM. The hypothesis of this theorem is satisfied for instance
under the sampling model considered in \cite{witnessed}.

The uniform norm of a function $f$ on $\bbR^d$ is denoted
$\norme{f}_\infty = \sup_{\bbR^d} |f|$. Given a Lipschitz function $f$
with Lipschitz constant $\Lip(f)$, we introduce its bounded-Lipschitz
norm $\norme{f}_\BL = \norme{f}_\infty + \Lip(f)$.
\begin{theorem} \label{thm:main-stability} Let $K$ be a compact set and $\delta$ a distance-like function.  For any bounded Lipschitz function $\chi : \bbR^d \longrightarrow \bbR^+$, one has
\begin{equation*}
\| \mathcal{V}_{\delta,R}(\chi) - \mathcal{V}_{d_K,R}(\chi)  \|_{\op} \leq C_1\ \norme{\chi}_\BL \norme{\delta-d_{K}}_{ \infty}^{\frac{1}{2}},
\end{equation*}
where $C_1$ is a constant that only depends on $R$, $d$ and $\diam(K)$.
\end{theorem}
In practice choosing a probe function $\chi$ supported in a small ball
allows one to recover local information from the $\delta$-VCM.

\begin{remark}
A notable feature of this theorem is that the constant in the upper
bound depends only on the diameter of $K$ and not on its local geometry or
on its regularity.
\end{remark}

We now recall a simplified version of Davis-Kahan Theorem \cite{matrix-perturbation}. 
\begin{theorem}[Davis-Kahan]
Let $M$ and $M'$ be two symmetric matrices, $\lambda$ an eigenvalue of M and $\delta_{\lambda}(M)$ be the minimum distance between $\lambda$ and any other eigenvalue of $M$. Then for every eigenvector $v$ of $M$ with eigenvalue $\lambda$, there exists an eigenpair $(\lambda',v')$  of $M'$ such that
$$
| \lambda - \lambda '| \leq \sqrt{2}\|M-M'\|_{\op}
\quad \mbox{and}\quad
\|v - v'\| \leq \frac{\sqrt{2}\|M-M'\|_{\op}}{\delta_{\lambda}(M)},
$$
provided that $\|M-M'\|_{\op} \leq \sqrt{2}\ \delta_{\lambda}(M)$.
\end{theorem}

\begin{remark}\label{remark:davis-kahan}
When the compact set $K$ is a smooth surface $S$ and $\chi_p^r$ is the
indicatrix of a ball centered at a point $p$ of $S$, Equation
(\ref{eq:vcm-normal}) implies that the eigenvector associated to the
highest eigenvalue of $\mathcal{V}_{d_K,R}(\chi_p^r)$ gives an estimation
of the normal vector to $S$ at $p$. Indeed, $\lambda = 1$ is the only non-vanishing eigenvalue 
of the matrix $M:=\n(p) \otimes \n(p)$ and $\delta_{\lambda(M)}=1$. Hence, by 
Davis-Kahan theorem, and under reasonable assumptions, an eigenvector associated to the highest 
eigenvalue of $\mathcal{V}_{\delta,R}(\chi_p^r)$ thus gives an estimation of the normal direction of $S$ at the point $p$.
%
\end{remark}

\subsection{Stability of gradients of distance-like functions}\label{subsection:gradients}
We mention in this subsection an intermediary result, Corollary
\ref{coro:gradient}, that guarantees the $\LL^1$-stability of
gradients of distance-like functions. It is a consequence of Theorem
3.5 of \cite{boundary} which gives the stability of gradients of
convex functions.  For any function $f:\bbR^d\to\bbR^k$ and any set $E$
of $\bbR^d$, we put
$$\|f \|_{1,E} = \int_E \|f(x)\| \d x \qquad \hbox{ and } \qquad
\|f\|_{\infty,E}=\sup_{x\in E}\|f(x)\|.$$
If $E$ is rectifiable, we denote by $\mathcal{H}^n(E)$ its $n$-dimensional Hausdorff measure. 
If $f$ is differentiable almost everywhere, one puts $\nabla f(E):=\{\nabla f(x),\ x\in E \cap \Omega_f\}$ where $\Omega_f$ is the set of points where $f$ is differentiable. The diameter of a set $X\subset \bbR^d$ is given by $\diam(X):=\sup_{x,y\in X}\|x-y\|$. We first recall the following theorem.

\begin{theorem}[Theorem 3.5 of \cite{boundary}]\label{thm:3.5}
Let E be an open subset of $\bbR^d$ with rectifiable boundary, and $f$, $g$ be two convex functions from $E$ to $\mathbb{R}$ such that $\diam(\nabla f(E) \cup \nabla g(E)) \leq D$. Then  
$$
\| \nabla f - \nabla g \|_{1,E} \leq C_2\ \left( \mathcal{H}^d(E) + (D + \| f-g \|^{\frac{1}{2}}_{\infty , E}) \mathcal{H}^{d-1}(\partial E) \right) \| f- g \|^{\frac{1}{2}}_{\infty , E},
$$
\label{thm:convex}
where $C_2$ is a constant that only depends on the dimension $d$.
\end{theorem}

\begin{corollary}\label{coro:gradient}
Let $E$ be a set of $\bbR^d$ with rectifiable boundary and $\varepsilon >0$. For any distance-like functions $\delta$ and $\delta'$ such that $\| \delta - \delta' \|_{\infty,E}\leq \varepsilon$, one has
$$
\| \mathbf{n}_{\delta'} - \mathbf{n}_{\delta} \|_{1,E} \leq C_3\  [\mathcal{H}^d(E) + (2\diam(E) + 4R+ \sqrt{2R \varepsilon})\ \mathcal{H}^{d-1}(\partial E)]\ \sqrt{2R \varepsilon},
$$
where $R=\max(\|\delta\|_{\infty,E}, \|\delta'\|_{\infty,E})$ and
$C_3$ is a constant depending only on $d$.
\end{corollary}
We introduce $\psi_{\delta}(x)=\|x\|^2 - \delta^2(x)$ and $\psi_{\delta'}(x)=\|x\|^2 - \delta'^2(x)$. For almost every $x$, one has 
$$
\mathbf{n}_{\delta'}(x) - \mathbf{n}_{\delta}(x) = \frac{1}{2} (\nabla \delta'^2(x) -\nabla \delta^2(x) )
= \frac{1}{2} (\nabla \psi_{\delta}(x) - \nabla \psi_{\delta'}(x)).
$$

Using the convexity of $\psi_{\delta}$ and $\psi_{\delta'}$ and Theorem \ref{thm:3.5}, Corollary \ref{coro:gradient} follows from the following lemma.
\begin{lemma}\label{lemma:gradient}Under the assumptions of Corollary \ref{coro:gradient}, one has
\begin{itemize}
\item[(i)] $\forall x \in E, \; \vert \psi_{\delta}(x) - \psi_{\delta'}(x) \vert \leq 2R \varepsilon$
\item[(ii)] $\diam(\nabla \psi_{\delta}(E) \cup \nabla\psi_{\delta'}(E)) \leq 2\diam(E) + 4R.$
\end{itemize}
\end{lemma}

\begin{proof}
For every point $x$ in $E$, one has
\begin{eqnarray*}
| \psi_{\delta}(x) - \psi_{\delta'}(x) | &=& | \delta^2(x) - \delta'^2(x) | \\
&=& | \delta(x) - \delta'(x) | \times | \delta(x) + \delta'(x) | \\
&\leq& 2 R \varepsilon.
\end{eqnarray*}
Let now $X$ and $X'$ be points in $\nabla \psi_{\delta}(E)$ and
$\nabla \psi_{\delta'}(E)$ respectively. There exist $x,x'$ in $E$
such that $X= 2 x - 2\delta(x)\nabla \delta (x)$ and $X'= 2 x' -
2\delta'(x') \nabla \delta' (x')$. Then
\begin{eqnarray*}
\| X-X' \vert \vert & \leq & \| 2x - 2x' \| + \|  2\delta'(x')\nabla \delta' (x') - 2\delta(x)\nabla \delta (x)\| \\
&\leq& 2\ \diam(E) + 2R \| \nabla \delta \|_{\infty,E} + 2R \| \nabla \delta' \|_{\infty,E} \\
&\leq& 2\ \diam(E) + 4R.
\end{eqnarray*}
In the last inequality, we used the fact that a distance-like function
is $1$-Lipschitz (see Proposition~3.1 in \cite{dist-measure}).  The
result still holds if $X,X'$ both belong to $\nabla \psi_{\delta}(E)$
or to $\nabla \psi_{\delta'}(E)$.
\end{proof}

\subsection{Proof of Theorem \ref{thm:main-stability}}
This proof follows the proof of the stability theorem in
\cite{vcm}. The idea is to compare the two integrals on the common set
$E = K^{R-\varepsilon}$ with $\varepsilon = \norme{\delta-d_{K}}_{
  \infty,E}$ and to show that remaining parts are negligible. By
Proposition 4.2 of \cite{boundary}, the set $\partial E$ is
rectifiable. For every $R \in \bbR^+$, we denote by
$\mathcal{N}(\partial K,R)$ the covering number of $\partial K$ with a
radius parameter $R$, namely the minimal number of balls of radius $R$
needed to cover $\partial K$.  Let us first suppose that $\varepsilon
< \frac{R}{2}$. We have
\begin{align*}
{\mathcal{V}}_{d_K,R}(\chi) &= 
\int_{E} \mathbf{n}_{d_K}(x) \otimes \mathbf{n}_{d_K}(x)\chi( x -\mathbf{n}_{d_K}(x) ) \d x \\
&+ \int_{K^R \backslash E} \mathbf{n}_{d_K}(x) \otimes \mathbf{n}_{d_K}(x) \chi(x - \mathbf{n}_{d_K}(x) ) \d x.
\end{align*}
 For every $x \in K^R$, one has $\| \mathbf{n}_{d_K}(x) \| = \| x - p_K(x)\| \leq R$. The fact that $\|\chi\|_\infty \leq \|\chi\|_\BL$ implies
\begin{equation}\label{Eq:bound1}
\left\| {\mathcal{V}}_{d_K,R}(\chi) - 
\int_{E} \mathbf{n}_{d_K}(x) \otimes \mathbf{n}_{d_K}(x)\chi(x-\mathbf{n}_{d_K}(x)) \d x \right\|_{op} 
\leq R^2\cdot \|\chi\|_\BL \cdot \mathcal{H}^{d}(K^R \backslash K^{R-\varepsilon}). 
\end{equation} 
We proceed similarly for the $\delta$-VCM.
By definition, we have $\|\mathbf{n}_{\delta}(x)\| = \|\delta(x) \nabla \delta(x)\| \leq
|\delta(x)| \leq R$, since $\delta$ is 1-Lipschitz. Moreover, one has $K^{ R - \varepsilon} \subset \delta^R \subset K^{ R + \varepsilon}$. Hence, the volume of $\delta^R \setminus K^{ R - \varepsilon}$ is less than the volume of $ K^{ R + \varepsilon} \backslash K^{ R - \varepsilon}$ and
\begin{equation}\label{Eq:bound2}
\left\| {\mathcal{V}}_{\delta,R}(\chi) - 
\int_{E} \mathbf{n}_{\delta}(x) \otimes \mathbf{n}_{\delta}(x)\chi(x-\mathbf{n}_{\delta}(x)) \d x \right\|_{op} 
\leq  R^2\cdot \|\chi\|_\BL \cdot \mathcal{H}^{d}(K^{R+\varepsilon} \backslash K^{R-\varepsilon}). 
\end{equation}
We now bound the volume of $K^{R+\varepsilon} \backslash
K^{R-\varepsilon}$ by using Proposition 4.2 of \cite{boundary}. We set
$\omega_{n}(t)$ to be the volume of the $n$-dimensional ball of radius
$t$.
\begin{align}\label{Eq:bound-volume}
\mathcal{H}^{d}(K^{R+\varepsilon} \backslash K^{R-\varepsilon})
& = \int_{R-\varepsilon}^{R+\varepsilon} \mathcal{H}^{d-1}(\partial K^t)\ \d t \notag\\
& \leq \int_{R-\varepsilon}^{R+\varepsilon}\mathcal{N}(\partial K,t) \omega_{d-1}(2t)\ \d t \notag \\
& \leq \mathcal{N}(\partial K,R-\varepsilon) \omega_{d-1}(2(R+\varepsilon)) 2\varepsilon \notag \\
& \leq  2\ \mathcal{N}(\partial K,\frac{R}{2}) \omega_{d-1}(3R)\ \varepsilon.
\end{align}

We now need to bound the operator norm of the integral $ \int_E[A_{d_K}(x)-A_{\delta}(x)]dx$, where
$$A_{d_K}(x) := \mathbf{n}_{d_K}(x) \otimes \mathbf{n}_{d_K}(x) \chi(x-\mathbf{n}_{d_K}(x) ),$$
and $A_{\delta}$ is defined similarly. The difference of these two terms is decomposed as:
\begin{align}
A_{d_K}(x) - A_{\delta}(x) &= \chi(x-\mathbf{n}_{d_K}(x) )(\mathbf{n}_{d_K}(x) \otimes \mathbf{n}_{d_K}(x)-\mathbf{n}_{\delta}(x) \otimes \mathbf{n}_{\delta}(x)) \nonumber \notag\\
&+\ [\chi(x-\mathbf{n}_{d_K}(x) )-\chi(x-\mathbf{n}_{\delta}(x) )]\mathbf{n}_{\delta}(x) \otimes \mathbf{n}_{\delta}(x) \notag
\end{align}
From the facts that
\begin{displaymath}
\begin{aligned}
&\quad \|\mathbf{n}_{d_K}(x) \otimes \mathbf{n}_{d_K}(x)-\mathbf{n}_{\delta}(x) \otimes \mathbf{n}_{\delta}(x)\| \\
&\leq \| \mathbf{n}_{d_K}(x) \otimes ( \mathbf{n}_{d_K}(x) - \mathbf{n}_{\delta}(x)) \| \nonumber 
+  \|( \mathbf{n}_{\delta}(x) - \mathbf{n}_{d_K}(x)) \otimes \mathbf{n}_{\delta}(x) \|\notag\\
&\leq 2 R\ \|\mathbf{n}_{d_K}(x) - \mathbf{n}_{\delta}(x) \|,\\
\end{aligned}
\end{displaymath}
and also
\begin{displaymath}
\| [\chi(x-\mathbf{n}_{d_K}(x) )-\chi(x - \mathbf{n}_{\delta}(x))]\mathbf{n}_{\delta}(x) \otimes \mathbf{n}_{\delta}(x) \|
\leq \Lip(\chi)\ \| \mathbf{n}_{d_K}(x) - \mathbf{n}_{\delta}(x)\|\ R^2
\end{displaymath}
one has
\begin{displaymath}
\|  A_{d_K}(x)-A_{\delta}(x) \|_{\op} \leq \norme{\chi}_\BL (R^2 + 2 R)\| \mathbf{n}_{d_K}(x)-\mathbf{n}_{\delta}(x)\|,
\end{displaymath}
%
%
%
which by integration leads to
\begin{displaymath}
\left\| \int_E [A_{d_K}(x)-A_{\delta}(x)]dx \right\|_{\op} \leq 
\norme{\chi}_\BL (R^2 + 2 R)\| \mathbf{n}_{d_K}-\mathbf{n}_{\delta}\|_{1,E}.
\end{displaymath}
Corollary \ref{coro:gradient} implies that the norm $\|\mathbf{n}_{d_K} - \mathbf{n}_{\delta} \|_{1,E}$ is bounded by
\begin{displaymath}
C_2 \times [\mathcal{H}^d(E) + (2\diam(E) + 4R+ \sqrt{2R \varepsilon})  \times \mathcal{H}^{d-1}(\partial E)] \times \sqrt{2R \varepsilon}.
\end{displaymath}
Since $\diam(E) \leq \diam(K) +2R$ and $\varepsilon \leq \frac{R}{2}$, the two last equations imply
\begin{align}\label{Eq:bound3}
&\quad \left\| \int_E [A_{d_K}(x)-A_{\delta}(x)]dx \right\|_{\op}  \\
&\leq C_2 \norme{\chi}_\BL (R^2 + 2 R)  [\mathcal{H}^d(K^R) + (2\diam(K) + 9R)   \mathcal{H}^{d-1}(\partial K^{R - \varepsilon})]\ \sqrt{2R}\  \sqrt{\varepsilon}.\notag
\end{align}
Again by Proposition 4.2 of \cite{boundary}, we bound $\mathcal{H}^{d-1}(\partial K^{R-\varepsilon})$ by $\mathcal{N}(\partial K,\frac{R}{2}) \cdot \omega_{d-1}(3R)$.  
%
%
%
Furthermore, since the diameter of $K^R$ is less than $\diam(K)+2R$, its volume $\mathcal{H}^d(K^R)$ is bounded by a constant involving $\diam(K)$, $R$ and $d$. Remark that $\mathcal{N}(\partial K, \frac{R}{2})$ is also bounded by a constant depending on the same quantities. 
%
Hence, the constants involved in Equations (\ref{Eq:bound3}) and (\ref{Eq:bound-volume}) also depend solely on $\diam(K)$, $R$ and $d$. 
We conclude the proof of Theorem \ref{thm:main-stability} by combining (\ref{Eq:bound3}) with (\ref{Eq:bound1}) and (\ref{Eq:bound2}). 

The case where $\varepsilon \geq \frac{R}{2}$ is trivial (and not interesting in practice) since both $\mathcal{V}_{\delta,R}(\chi)$ and $\mathcal{V}_{d_K,R}(\chi)$ can be upper bounded by a quantity depending on $\diam(K)$, $R$ and $d$, which can be put into the constant $C_1$ of the Theorem.

\section{VCM using the distance to a measure}\label{section:distance-to-measure}
The {\em distance to a measure} is a distance-like function that is
known to be resilient to outliers. It is therefore natural to consider
the $\delta$-VCM in the particular case where $\delta$ is a distance
to a measure.  First we recall the definition and some stability
results of the distance to a measure. Then we study the $\delta$-VCM
when $\delta$ is a distance to a measure.

\subsection{Distance to a measure}\label{subsection:distance-to-a-measure}
The distance to a measure has been introduced in \cite{dist-measure}
and is defined for any probability measure $\mu$ on $\bbR^d$. We
denote in the following $\mathrm{supp}(\mu)$ the support of $\mu$.
\begin{figure}
\begin{center}
\includegraphics[width=5.5cm]{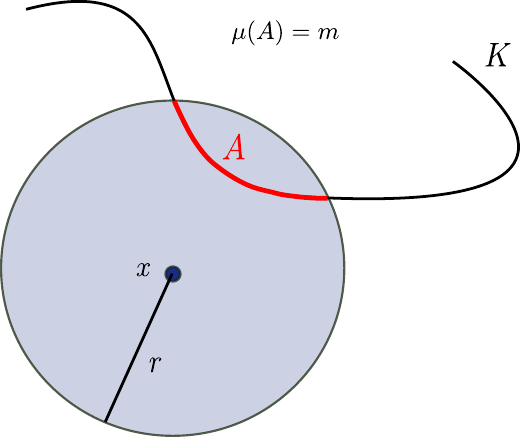}
\end{center}
\caption{Distance to a measure. We suppose here that the measure $\mu$
  is supported on a compact set $K$.  The quantity $\delta_{\mu,m}(x)$
  is the minimal radius $r$ such that $\bbB(x,r)\cap K$ has a mass
  $m$.  } \label{figure:dist-to-measure}
\end{figure}
\begin{definition}
  Let $\mu$ be a probability measure on $\bbR^d$ and $m_0$ a
  regularization parameter in $(0,1)$. The distance function to the
  measure $\mu$ is defined for every point $x$ in $\bbR^d$ by
\begin{equation}
\d_{\mu,m_0} (x) :=\left(\frac{1}{m_0}\int_0^{m_0} \delta^2_{\mu,m}(x) \d m\right)^{1/2},
\label{eq:dm}
\end{equation}
where $\delta_{\mu,m} (x) = \inf \{ r\geq 0, \; \mu(\bbB(x,r)) \geq m \}$.
\end{definition}
The formula defining the function $\delta_{\mu,m}$ mimics a similar
formula for the distance function to a set $K$: $\d_K(x) = \inf \{
r\geq 0, \; K \cap \bbB(x,r) \neq \emptyset \}$, as shown on Figure
\ref{figure:dist-to-measure}.  However it turns out that the function
$\delta_{\mu,m}$ is not distance-like \cite{dist-measure}, while the
regularization defined by Eq.~\eqref{eq:dm} is distance-like.
Furthermore, the distance to the measure has been shown to be
resilient to outliers, and more precisely, Theorem~3.5 of
\cite{dist-measure} states that
$$
\norme{d_{\mu,m_0}-d_{\mu ',m_0}}_\infty \leq \frac{1}{\sqrt{m_0}} \Wass_2( \mu , \mu '),
$$
where $\Wass_2$ is the $2$-Wasserstein distance between measures.  For
more details on Wasserstein distances, which are also known under the
name Earthmover distances in image retrieval and computer science, one
can refer to \cite{Villani}. To give an intuition, when $\mu_K$ and
$\mu_{K'}$ are uniform probability measures on two point clouds $K$
and $K'$ with the same number of points, the distance
$\Wass_2(\mu_K,\mu_{K'})$ is the square root of the cost of a least
square assignment between $K$ and $K'$.

\begin{figure}
\begin{center}
\begin{tabular}{cccc}
\includegraphics[width=3.5cm]{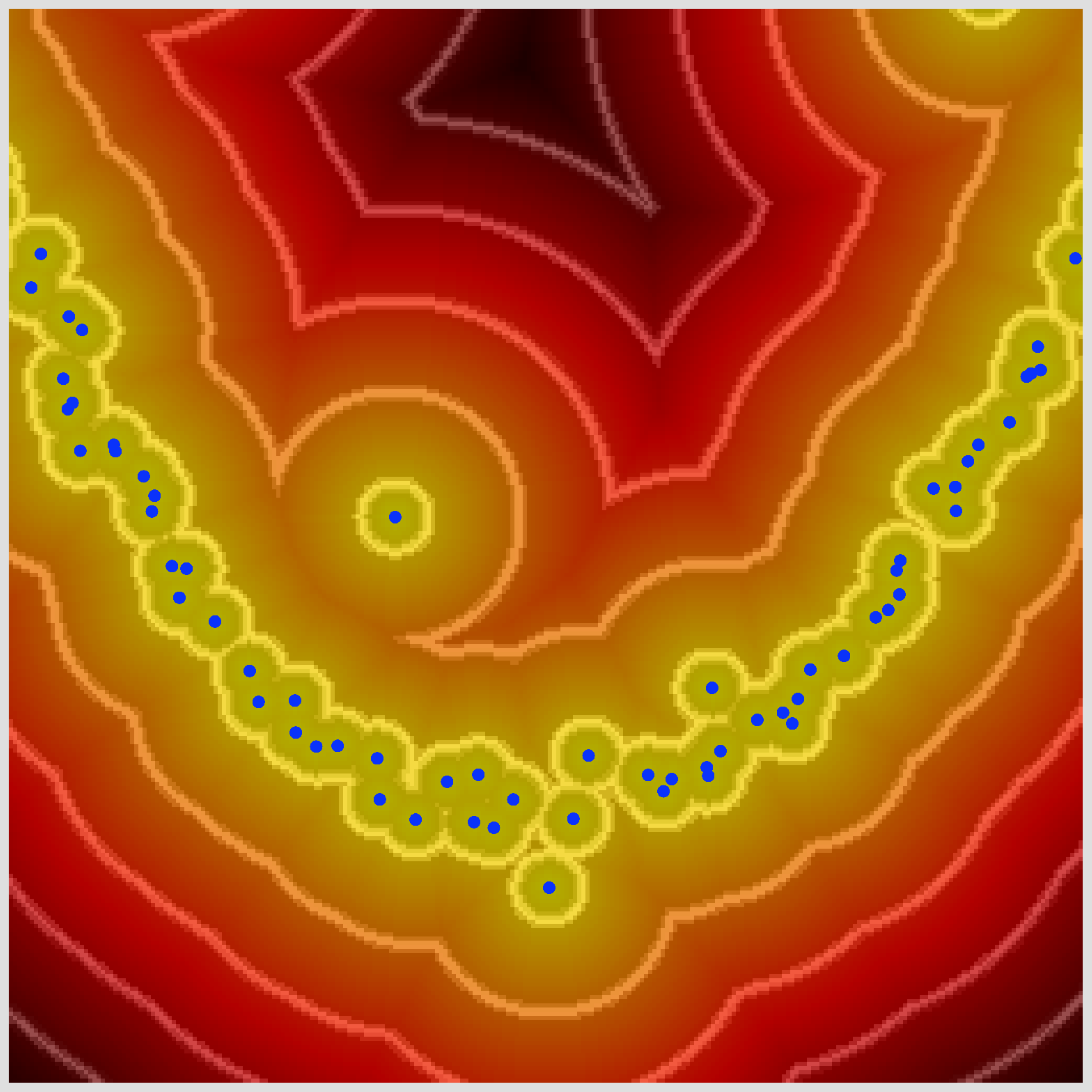}&
\includegraphics[width=3.5cm]{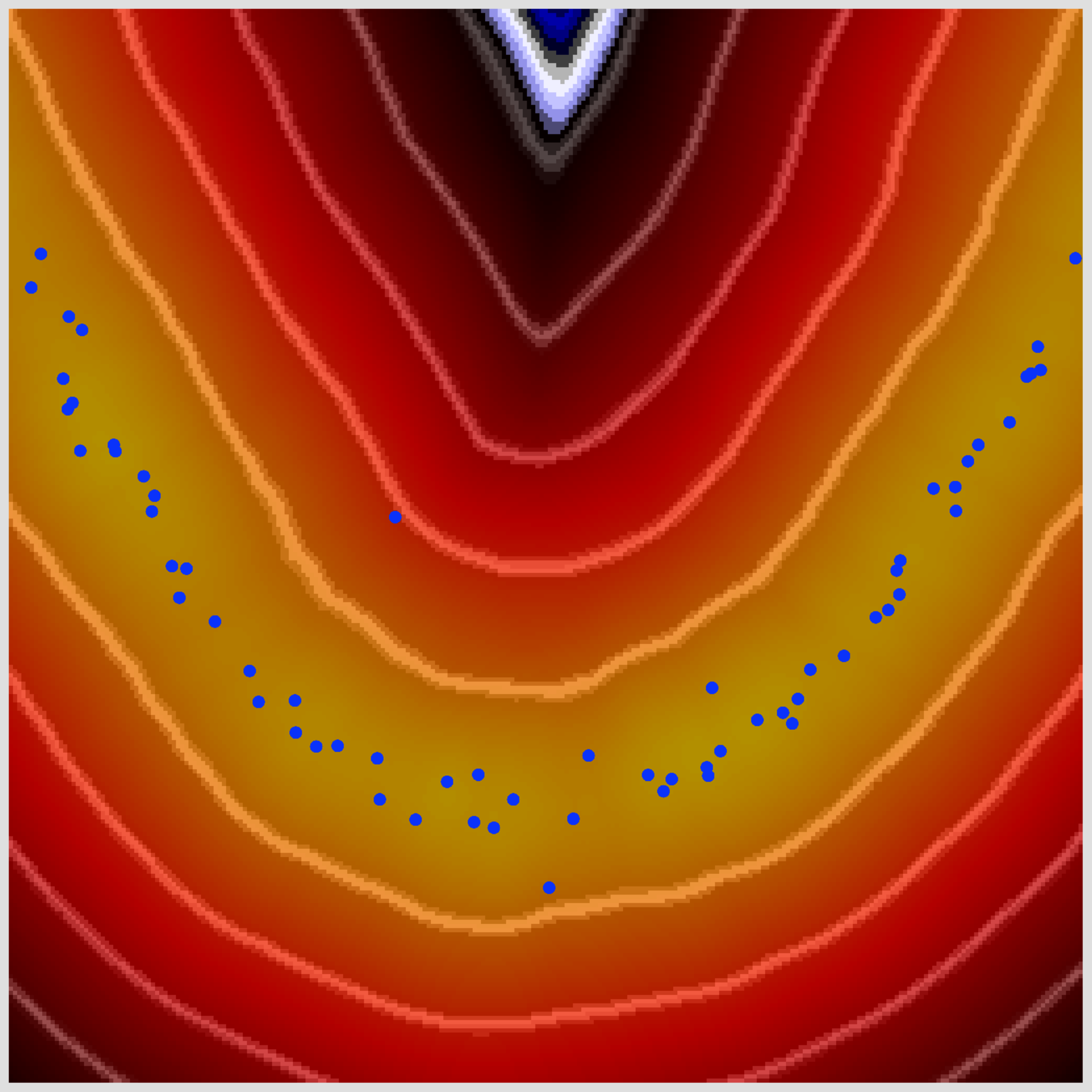}&
\includegraphics[width=3.5cm]{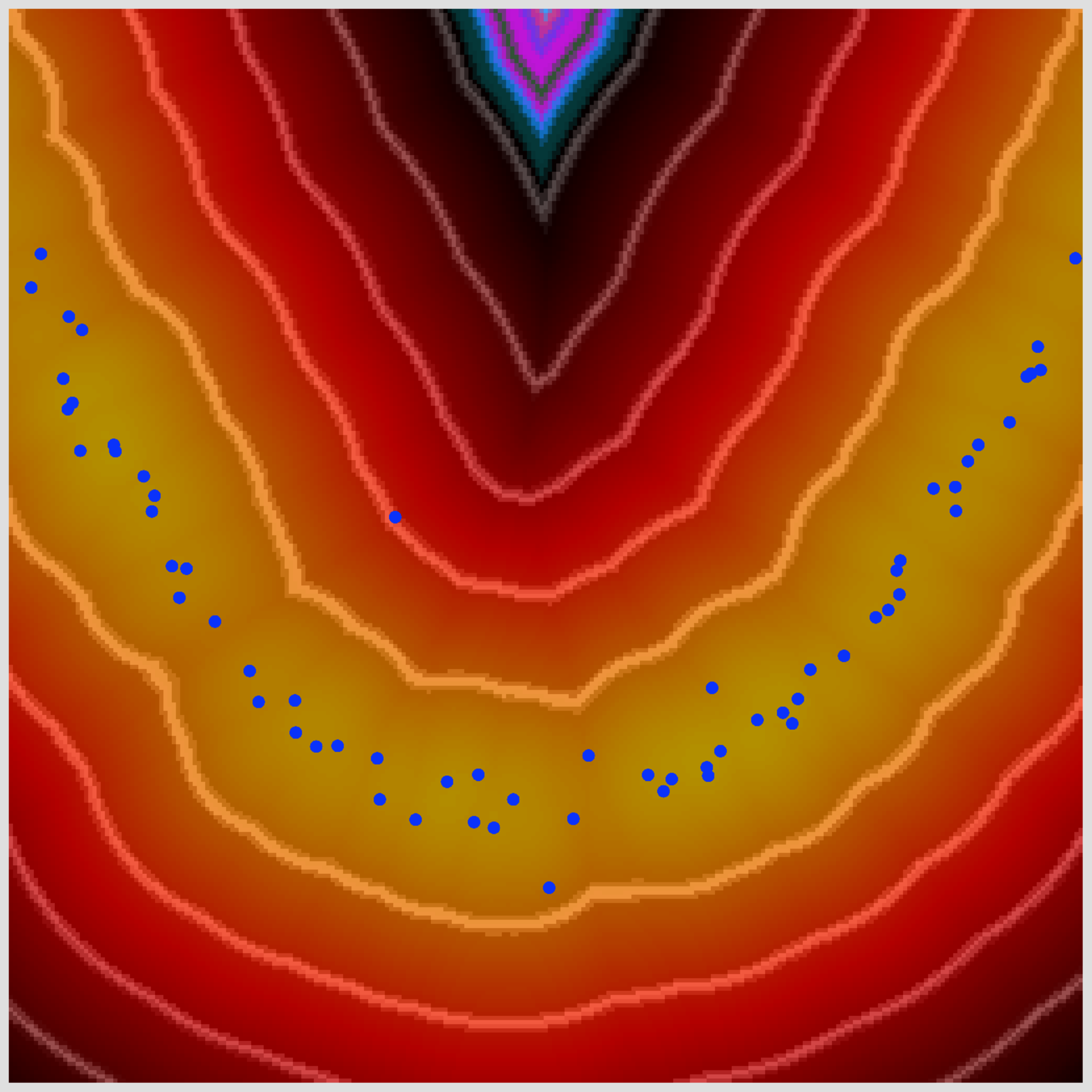}&
\includegraphics[width=3.5cm]{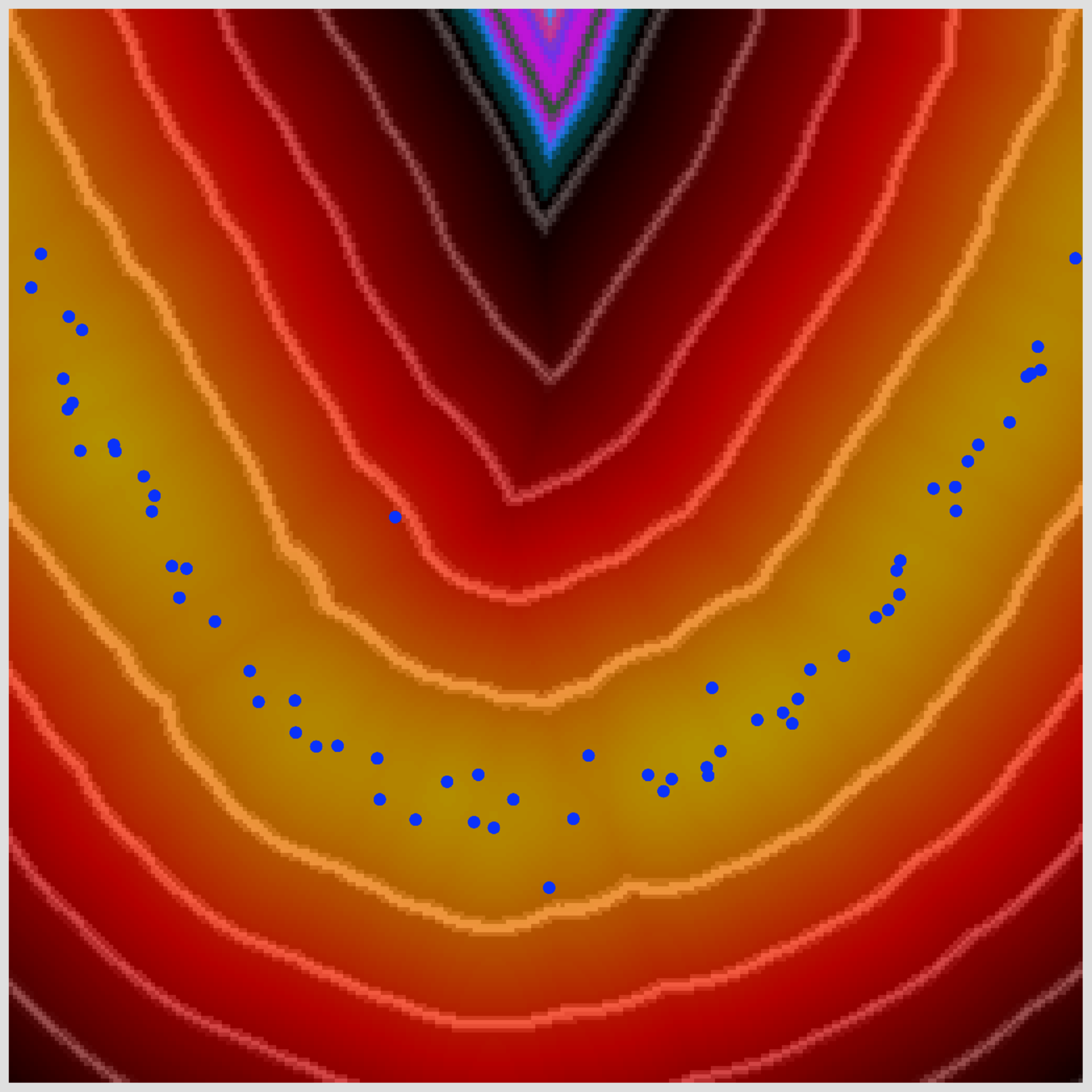}\\
\end{tabular}
\end{center}
\caption{Level sets of distance functions of a noisy point cloud. From left to right: usual distance function to the point cloud; $k$-distance; witnessed-$k$-distance; median-witnessed-$k$-distance}\label{figure:dists}
\vspace{-0.4cm}
\end{figure}

\paragraph{Point cloud case} Let $P \subset \bbR^d$ be a finite point
set with $n$ points, $k\in (0,n)$ a real number and $m_0=
k/n$. Following \cite{witnessed}, we call \textit{$k$-distance to $P$}
and denote $d_{P,k}$ the distance to the measure for the uniform
measure on $P$ for the parameter $m_0$. In the particular case where
$k$ is an integer, a simple calculation \cite{dist-measure} shows that
for every point $x$ in $\bbR^d$,
$$
d^2_{P,k}(x) = \frac{1}{k}\sum_{p_i \in \NN_{P,k}(x)} \norme{x-p_i}^2,  
$$
where $\NN_{P,k}(x)$ are the $k$ nearest neighbors of $x$ in
$P$. Furthermore, the $k$-distance $d_{P,k}$ is a power distance
\cite[Proposition~3.1]{witnessed}.  More precisely, if we denote
$\mathrm{Bary}_{P,k}$ the set of isobarycenters of $k$ distinct points
in $P$, one has
$$
\forall  x\in \bbR^d\quad d^2_{P,k}(x) = \min_{b \in \mathrm{Bary}_{P,k}} \left(\| x - b \|^2 + \omega_{b} \right),
$$
where the weight $\omega_{b} = \frac{1}{k} \sum_{p_i \in
  \NN_{P,k}(b)} \|b - p_i\|^2$. Figure~\ref{figure:dists} illustrates the stability of the
$k$-distance of a point cloud with respect to outliers. Remark that
the level sets of the $k$-distance (second picture), are much smoother and faithful
than the level sets of the distance function (first picture).

\subsection{Stability}\label{subsection:stability-distance-to-a-measure}
We state here a stability result obtained by combining our result with a stability result for  the distance to a measure established in \cite{witnessed}. Although this result can hold for measures supported in $\bbR^d$, we choose to state it for uniform measures on surfaces of $\bbR^3$, so as to give a better intuition.

Let $S$ be a two dimensional surface and $\mu_S$ denote the uniform
probability measure on $S$. An application of G\"unther-Bishop theorem
which can be found explicitely on page 749 of \cite{dist-measure} is
that there exists a constant $\alpha_S > 0$ such that
\begin{equation}
 \forall p\in S,~\forall r
\leq \diam(S),\qquad\mu_S(\bbB(p,r)) \geq \alpha_S r^2.
\end{equation}
\begin{proposition}\label{prop:stability-distance-to-a-measure}
  Let $P$ be a point cloud of cardinal $n$, $S$ a surface of $\bbR^3$
  and let $\mu_P$ and $\mu_S$ be the uniform probability measures on
  $P$ and $S$ respectively. If we define $k = W_2(\mu_S,\mu_P)
  \sqrt{\alpha_{S}} n$, then we have
\begin{equation*}
\| \mathcal{V}_{d_{P,k},R}(\chi) - \mathcal{V}_{d_S,R}(\chi)  \|_{\op} \leq C_4\ \alpha_S^{-1/8}\ \norme{\chi}_\BL\ W_2(\mu_S,\mu_P)^{\frac{1}{4}},
\end{equation*}
where the constant $C_4$  depends on $\diam(S)$ and $R$. 
\end{proposition}

\begin{proof}
Let $k\in (0,n)$ and $m_0=k/n$. By Theorem 1 of \cite{witnessed} (Section 4.2), one has
$$
\| d_{P,k} -d_K \|_{\infty} \leq m_0^{-\frac{1}{2}}W_2(\mu_S,\mu_P)  + \alpha_{S}^{-\frac{1}{2}} m_0^{\frac{1}{2}}.
$$
The right hand-side is minimal when $m_0$ satisfies $ m_0^{-\frac{1}{2}}W_2(\mu_S,\mu_P)  = \alpha_{S}^{-\frac{1}{2}} m_0^{\frac{1}{2}} $, namely $m_0=W_2(\mu_S,\mu_P) \sqrt{\alpha_S}$. If $k = W_2(\mu_S,\mu_P) \sqrt{\alpha_{S}} n$, the bound in the previous equation becomes 
$W_2(\mu_S,\mu_P)^{\frac{1}{2}} \alpha_S^{-1/4}$. We conclude by Theorem \ref{thm:main-stability}.
\end{proof}
Remark this result can be easily extended to any dimension $d$ and also to the case where $S$ is a compact set with a dimension at most $l$ (\textit{i.e.} there exists a constant $\alpha_S$ such that $\mu_S(\bbB(p,r)) \geq \alpha_S r^l$).

\section{Tractable variants for the $d_{P,k}$-VCM}
We have seen in the previous section that the $d_{P,k}$-VCM is
resilient to outliers. Unfortunately, it is not practically
computable, because the computation of the $k$-distance is not. We
therefore study in this section variants of the $d_{P,k}$-VCM, by
using variants of the $k$-distance. In Section
\ref{subsection:witnessed}, we propose an efficient relaxation of the
$d_{P,k}$-VCM for which we have a stability result. In Section
\ref{subsection:other-distances}, we investigate another $\delta$-VCM
interesting for its resilience to outliers.

\subsection{Witnessed-$k$-distance and $d^{\mathrm{w}}_{P,k}$-VCM}\label{subsection:witnessed}
Let $P\subset \bbR^d$ be a point cloud, $k$ an integer. We have seen
in the previous section that the $k$-distance is a power distance on
the set $\mathrm{Bary}_{P,k}$ of $k$ points of $P$. Lemma
\ref{lemma:power-distance-vcm} then allows us in theory to calculate
the $d_{P,k}$-VCM. Unfortunately, it is not computable in
practice. The set $\mathrm{Bary}_{P,k}$ is indeed huge since its
cardinal is of the order of $\binom{n}{k}$, if $n$ is the cardinal of
$P$. To overcome this problem, we use the witnessed $k$-distance that
was introduced in \cite{witnessed}.

\paragraph{Relaxation with the witnessed $k$-distance}
A point is said to be a $k$ witness of $P$ if it is the barycenter of $k$ points $p_1,\hdots, p_k$ of $P$ such that  $p_2,\hdots,p_{k}$ are the $(k-1)$-nearest neighboors of $p_1$ in $P\setminus\{p_1\}$. We denote by $\mathrm{Bary}_{P,k}^\mathrm{w}$ this set of points.
The {\em witnessed $k$-distance} is then defined as the following power distance 
\begin{equation*}
d^{\mathrm{w}}_{P,k}(x) := \left(\min_{b \in \mathrm{Bary}_{P,k}^\mathrm{w}} \| x - b
  \|^2 + \omega_{b}\right)^{1/2},
\end{equation*}
where $\omega_{b} = \frac{1}{k} \sum_{p_i \in \NN_{P,k}(b)} \|b - p_i\|^2$. We then use Lemma \ref{lemma:power-distance-vcm} for the computation of its VCM. Note that the number of cells in the power diagram is bounded by the number of points in $P$, and can therefore be computed efficiently using e.g. CGAL \cite{cgal}. We also have a stability result.

\begin{proposition}\label{prop:stability-witness-distance}
Let $\mu_S$ be the uniform measure on a surface $S\subset \bbR^3$ and  $\mu_P$ denotes the uniform probability measure on a point cloud $P$. 
Let $k =  W_2(\mu_P,\mu_S) \sqrt{\alpha_{S}} n / 4$, where $n$ is the number of points of $P$. Then for any bounded Lipschitz function $\chi:\bbR^d\to \bbR^+$ 
\begin{equation*}
\| \mathcal{V}_{d^{\mathrm{w}}_{P,k},R}(\chi) - \mathcal{V}_{S,R}(\chi)  \|_{\op} 
\leq C_5\ \alpha_S^{-1/8}\  \norme{\chi}_\BL\ W_2(\mu_P,\mu_S)^{\frac{1}{4}},
\end{equation*}
where the constant $C_5$  depends on $\diam(S)$ and $R$. 
\end{proposition}
The proof of this proposition is similar to the one of Proposition \ref{prop:stability-distance-to-a-measure}, and is using Theorem 4 of \cite{witnessed}.

%
%

\subsection{Median-$k$-distance and $d^{\mathrm{m}}_{P,k}$-VCM}\label{subsection:other-distances}
We introduce here the \textit{median-$k$-distance} which is a
distance-like function derived from the witnessed-$k$-distance. We do
not have yet stability results for this function, but numerical
experiments in the next section show that the corresponding
$\delta$-VCM is very stable in practice (see also the visual intuition
given by Figure~\ref{figure:dist-to-measure}).
Let $P$ be a point cloud and
$k>1$ be an integer. A point is said to be a {\em median $k$-witness}
of $P$ if it is the geometric median of $k$ points $p_1,\hdots, p_k$
of $P$ such that $p_2,\hdots,p_{k}$ are the $(k-1)$-nearest neightbors
of $p_1$ in $P\setminus\{p_1\}$. We denote by
$\mathrm{Bary}_{P,k}^{\mathrm{m}}$ this set of points.

 \begin{definition}
The \textit{median-$k$-distance} is the power distance defined by
\begin{equation*}
d^{\mathrm{m}}_{P,k}(x) := \left(\min_{b \in \mathrm{Bary}_{P,k}^{\mathrm{m}}} \| x - b
  \|^2 + \omega_{b}\right)^{1/2},
\end{equation*}
where $\omega_{b} = \frac{1}{k} \sum_{p_i \in \NN_{P,k}(b)} \|b - p_i\|^2$.
\end{definition}
The key idea is to replace the barycenter of $k$ points involved in
the witnessed-$k$-distance by the geometric median. This can be seen
as replacing the $L^2$-norm by the $L^1$-norm. Indeed, it is
well-known that the barycenter of $p_1,\hdots,p_k$ is the point $b$
that minimizes $\sum_{i=1}^d \| b-p_i \|^2$. Similarly, a geometric
median is a point $b$ that minimizes $\sum_{i=1}^d \| b-p_i \|$. Note
that the geometric median is unique when the $k$ points are not
colinear.

As for the witnessed-$k$-distance, one can compute the
$d^{\mathrm{m}}_{P,k}$-VCM by using Lemma
\ref{lemma:power-distance-vcm}.


\section{Computation and Experiments}

\subsection{Computation of the VCM for a power distance}
We describe here our algorithm to compute an approximation of the VCM
of a power distance. Let thus $(P,\omega)$ be a weighted point cloud
that defines a distance-like function $\delta_{P}$.  For the purpose
of normal and curvature estimation, we have to compute for each
point $q$ in $P$ the covariance matrix
$\mathcal{V}_{\delta_{P},R}(\chi_q^r)$, where $\chi_q^r$ is the
indicatrix of the ball $\bbB(q,r)$.
Using Lemma~\ref{lemma:power-distance-vcm}, we have
\begin{displaymath}
  \mathcal{V}_{\delta_{P},R}(\chi_q^r) = \sum_{p \in P \cap \bbB(q,r)} M_p,
\end{displaymath}
where the matrix $M_p$ is defined by
\begin{align*}
C_p &= \Pow_P(p) \cap  \mathbb{B}\left(p, (R^2 - \omega_p)^{1/2}\right)\\
M_p &= \int_{C_p} (x - p) \otimes (x-p) \d x
\end{align*}
The main difficulty is to compute these covariance matrices, and in
practice we approach it by replacing the ball in the definition of
$C_p$ by a convex polyhedron. The input of our algorithm (summarized
in Algorithm~\ref{algo:iterative}) is a weighted point cloud, a radius
$R$ and an approximation of the unit ball by a convex polyhedron
$B$. In all examples, the polyhedron $B$ was a dodecahedron
circumscribed to the unit ball. To compute the approximate covariance
matrices $(M_p^B)_{p\in P}$, we first build the power diagram of
the weighted point cloud. Then, we proceed in two steps for each
point:
\subsubsection{Intersection}
\label{subsec:intersection} We compute an approximation of the
cell $C_p$ by
\begin{equation}
C_p^B := \Pow_P(p) \cap (p + (R^2 - \omega_p)^{1/2} B)
\end{equation}
To perform this computation, we gather the half-spaces defining both
polyhedra and we compute their intersection. By duality, this is
equivalent to the computation of a convex hull.

\subsubsection{Integration}
\label{subsec:integration} The boundary of the polyhedron $C_p^B$ is
decomposed as a union of triangles, and we consider the tetrahedron
$\Delta_p^1,\hdots,\Delta_p^{k_p}$ joining these triangles to the
centroid of $C_p$.  We compute
\begin{equation}
M^B_p = \int_{C^B_p} (x - p) \otimes (x-p) \d x
\end{equation}
by summing the signed contribution of each tetrahedron, which can be
evaluated exactly using the same formulas as in \cite{voronoi-based}.
We implemented this algorithm using CGAL \cite{cgal} for the power
diagram calculation and the intersection step.  We report in
Table~\ref{tab:time} some running times for the
witnessed-$k$-distance and the median-$k$-distance.

\begin{algorithm}
  \caption{Computation of  $\mathcal{V}_{\delta_{P},R}$.}
  \label{algo:iterative} 
  \begin{algorithmic} \label{algorithm:iterative}
       \REQUIRE{$P \subseteq \bbR^d$ point cloud,
                $R>0$, $B=$ approximation of $\bbB(0,1)$}
	\STATE{Computation of the power diagram $(\Pow_P(p))_{p\in P}$}
	\FORALL{$p \in P$} 
                \STATE{$C_p^B \gets \Pow_P(p) \cap (p + (R^2 - \omega_p)^{1/2} B)$}  \hfill\COMMENT{\S\ref{subsec:intersection}}
		\STATE{$\Delta_p^1,\hdots,\Delta_p^{k_p}\gets$ decomposition of $C_p^B$ into tetrahedra}   \hfill\COMMENT{\S\ref{subsec:integration}}
		\STATE{$M^B_p$ $\leftarrow$ $\sum_{i=1}^{k_p} \int_{\Delta_p^i} (x - p) \otimes (x-p) \d x$} 	  \hfill\COMMENT{\S\ref{subsec:integration}}
	\ENDFOR
	\RETURN $(M_p^B)_{p\in P}$.
  \end{algorithmic}
\end{algorithm}

%
%

\begin{table}
\begin{center}
   \begin{tabular}{|c|c|c|c|}
  \hline
  shapes &  number of points & $d^{\mathrm{w}}_{P,k}$-VCM & $d^{\mathrm{m}}_{P,k}$-VCM \\
  \hline
  ellipsoid &  10K & 3.31 s & 3.44 s\\
  \hline
 hand & 36K & 12.31s & 13.86s\\
\hline
 bimba & 74K & 25.98 s & 29.08 s\\
  \hline
 ceasar & 387K & 175.16 s & 201.20 s\\
  \hline
   \end{tabular}
\end{center}
\caption{Computation times of $(M_p^B)_{p\in P}$ for the $d^{\mathrm{w}}_{P,k}$-VCM and the $d^{\mathrm{m}}_{P,k}$-VCM (with a $4 \times 2.5 Ghz$ CPU)}
\label{tab:time}
\end{table}



\subsection{$d^{\mathrm{w}}_{P,k}$-VCM evaluation}
We show in this section that the $d^{\mathrm{w}}_{P,k}$-VCM provides
in practice a robust estimator for normal vectors and curvatures, as well as for feature
detection. It is also resilient to both Hausdorff noise and outliers.
Note that in our experiments, the point cloud $P$ is the set of
vertices of a mesh. The mesh itself is only used for
visualisation, our algorithm being only based on the point cloud
$P$. The diameter of the point cloud in all our examples is
$D=2$. We say that the
Hausdorff noise is $\varepsilon$ if every point $p$ of $P$ can be
randomly moved at a distance less than $\varepsilon$. We will speak
about outliers if in addition a certain amount of points can be moved
much further.

\subsubsection{Normal estimation} As suggested by Remark \ref{remark:davis-kahan}, we define a normal at each point $p$ of a point cloud $P$ as the eigenvector associated to the largest eigenvalue of $\mathcal{V}_{\delta_{P},R}(\chi_p^r)$.


\begin{figure}[!h]  
\begin{tabular}{cc}
\includegraphics[width=.39\textwidth]{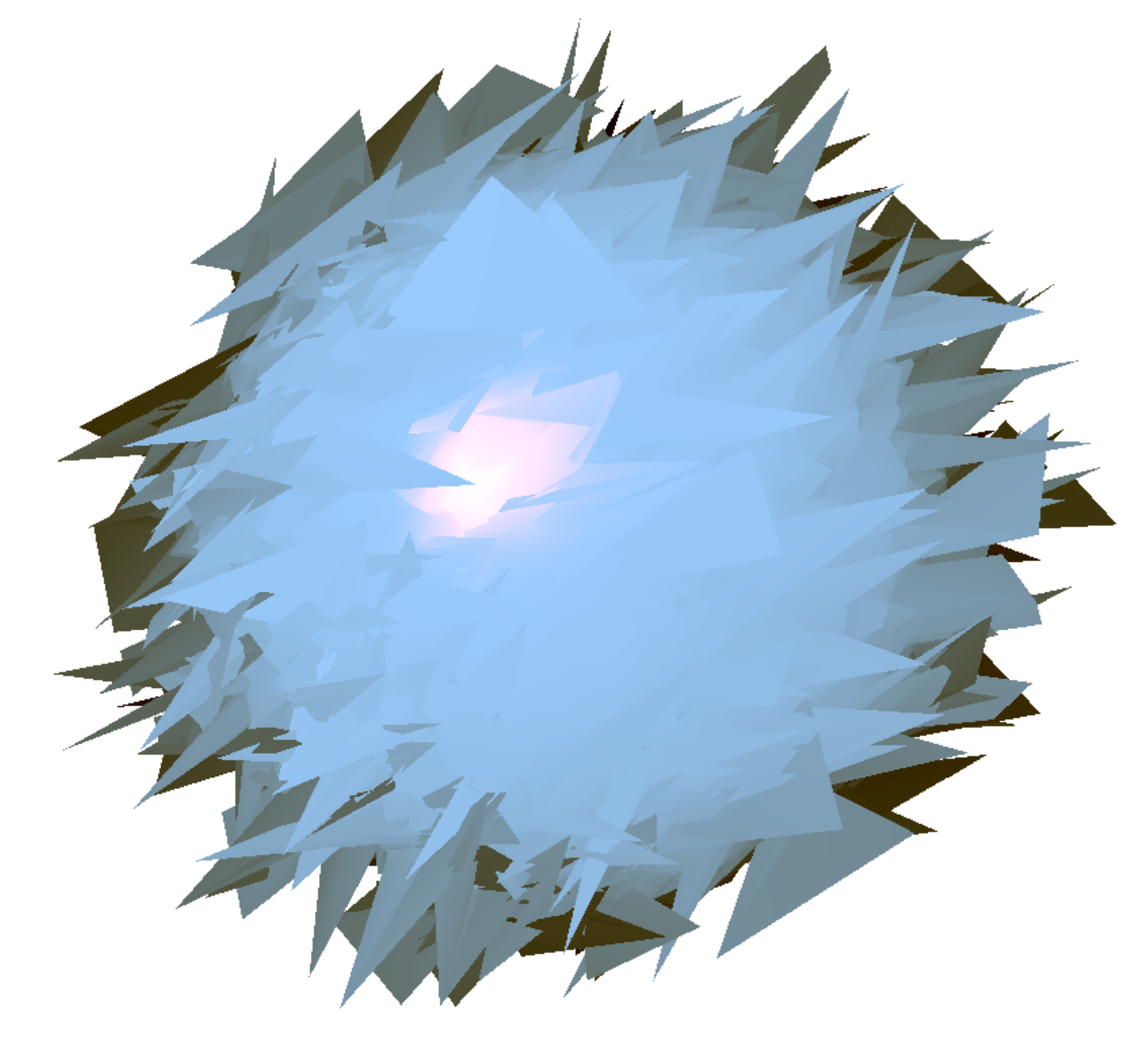} &
\includegraphics[width=.6\textwidth]{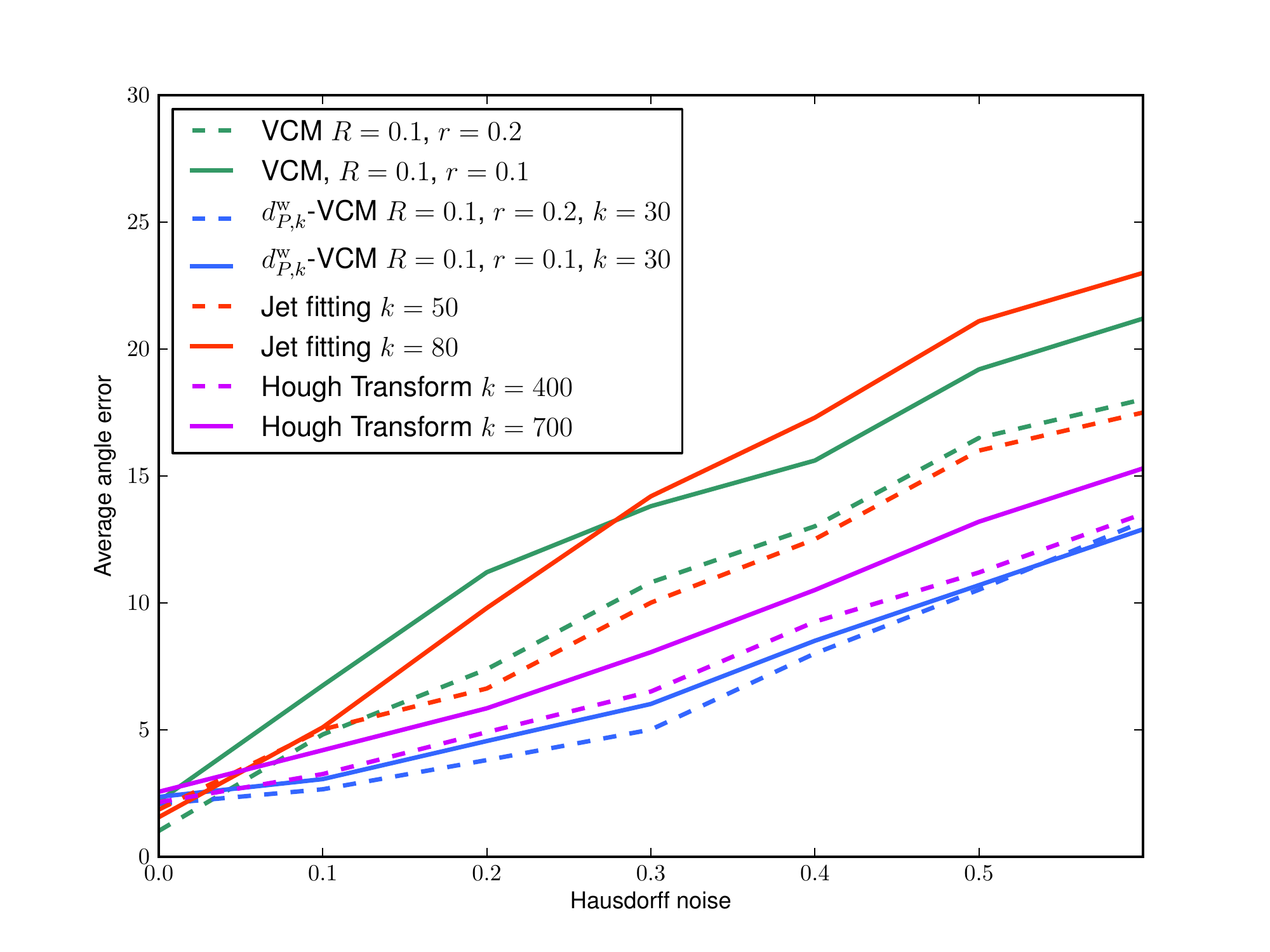}\\
\end{tabular}
\caption{Ellipsoid normal estimation. Left: the rendering illustrates the quality of the normal estimation (parameters: $k=30$, $R=0.2\diamnoise$, $r=0.2\diamnoise$, Hausdorff noise $=0.4\diamnoise$), where $\diamnoise$ is the diameter of the original shape. Right: Comparison of the average normal deviation for Jet Fitting, VCM,  $d^{\mathrm{w}}_{P,k}$-VCM and the Randomized Hough
  Transform method \cite{sgp-2012} for different values of Hausdorff noise.}
\label{ell}
\end{figure}

\paragraph{Comparison with other methods}
We compare the accuracy of our method to the VCM method~\cite{vcm}, to the
\textit{Jet Fitting} method~\cite{Pouget} and to the Randomized Hough
  Transform method\cite{sgp-2012}.  Figure \ref{ell} displays the
average angular deviation between the estimated normal and the ground
truth in the case of a noisy ellipsoid with outliers. Our method
gives better result than the others.

\paragraph{Sensitivity to parameters}
To measure the dependence of our method on the different parameters,
we computed the average deviation of the normal on a noisy ellipsoid
for different choices of parameters $k$ and $R$. As can be seen in
Figures~\ref{parameter1} and~\ref{parameter2}, the results are
essentially stable. Figure~\ref{parameter1} shows that we should
choose $k$ larger when the noise is large. We observe that the value
$k=30$ gives good results in both presence and absence of
noise. Therefore all our experiments (except of course the ones of
Figure~\ref{parameter1}) are done with the value $k=30$. Note that the
classical VCM method coincides exactly with the
$d^{\mathrm{w}}_{P,k}$-VCM method when $k=1$.

\begin{figure}[!h]  
\begin{tabular}{cc}
\includegraphics[width=0.49\textwidth]{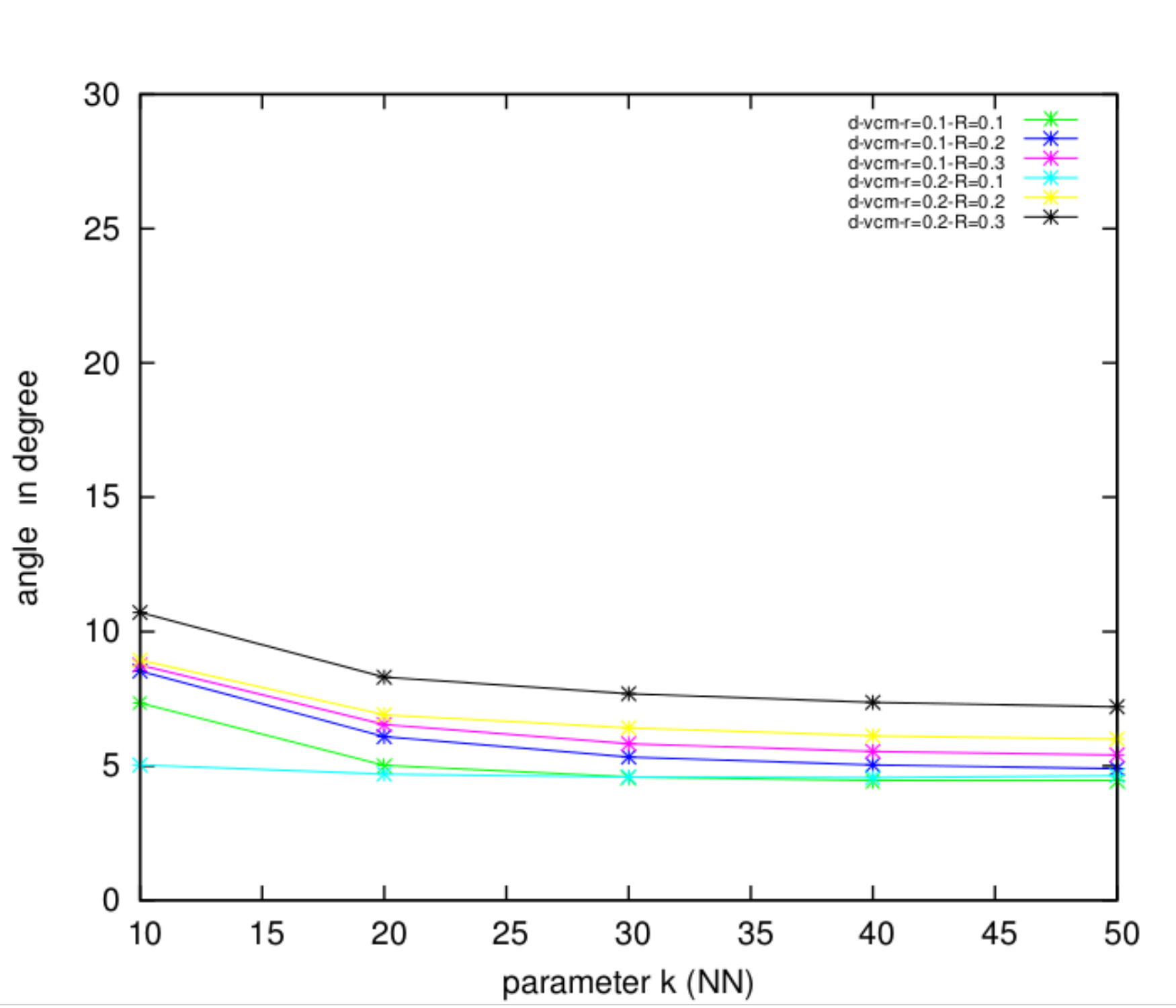}&
\includegraphics[width=0.49\textwidth]{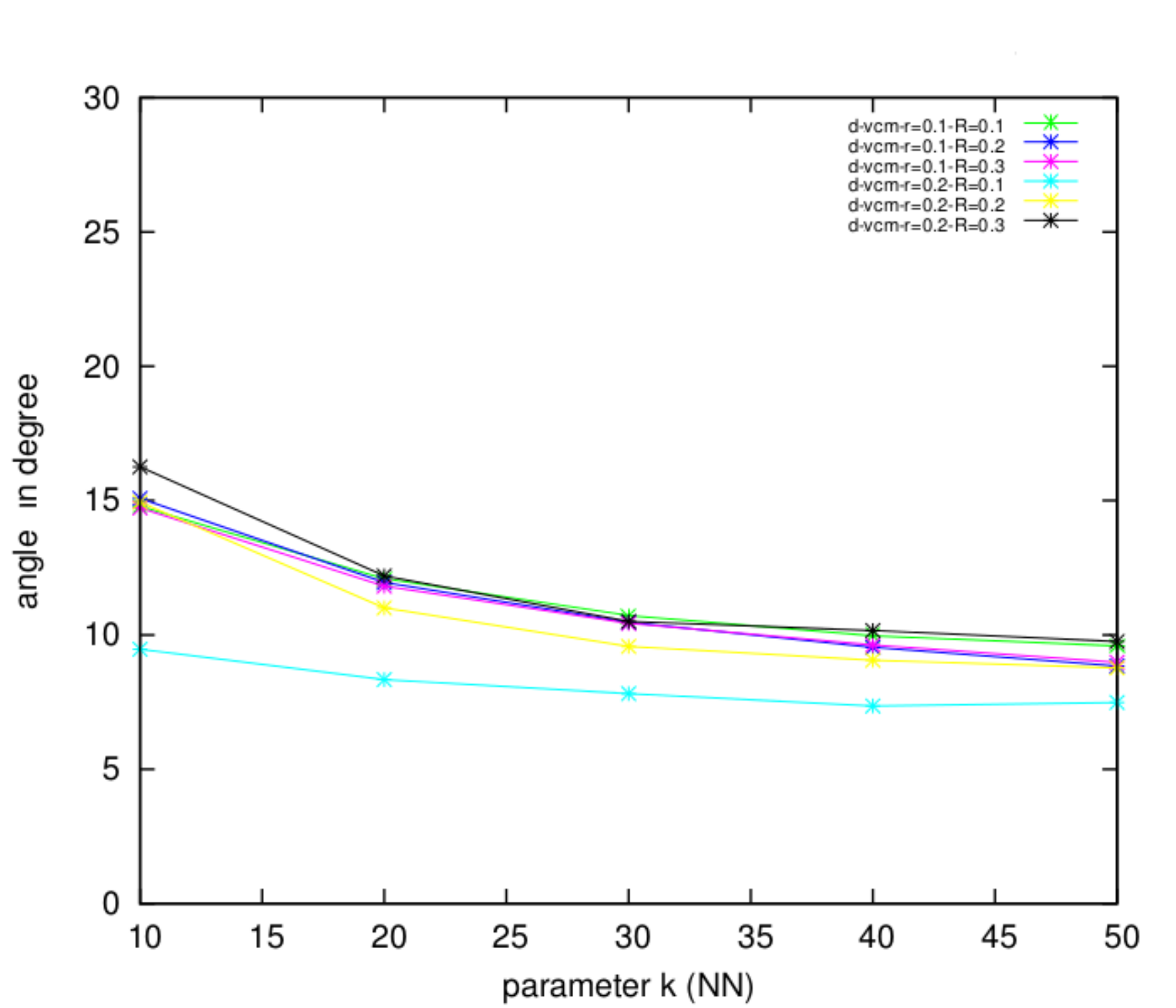}\\
\end{tabular}
\caption{Dependence on the parameter $k$: we see the influence of the
  parameter $k$ (in abscissa) on the average error on the normal estimation
  using $d^{\mathrm{w}}_{P,k}$-VCM. The experiment is done for
  different values of $r$ and $R$ and for two different noisy
  ellipsoids. Left: Hausdorff noise $=0.2\diamnoise$. Right: Hausdorff noise
  $=0.4\diamnoise$. Here $\diamnoise$ is the diameter of the original shape.}
\label{parameter1}
\end{figure}

\begin{figure}[!h]  
\begin{tabular}{cc}
\includegraphics[width=0.49\textwidth]{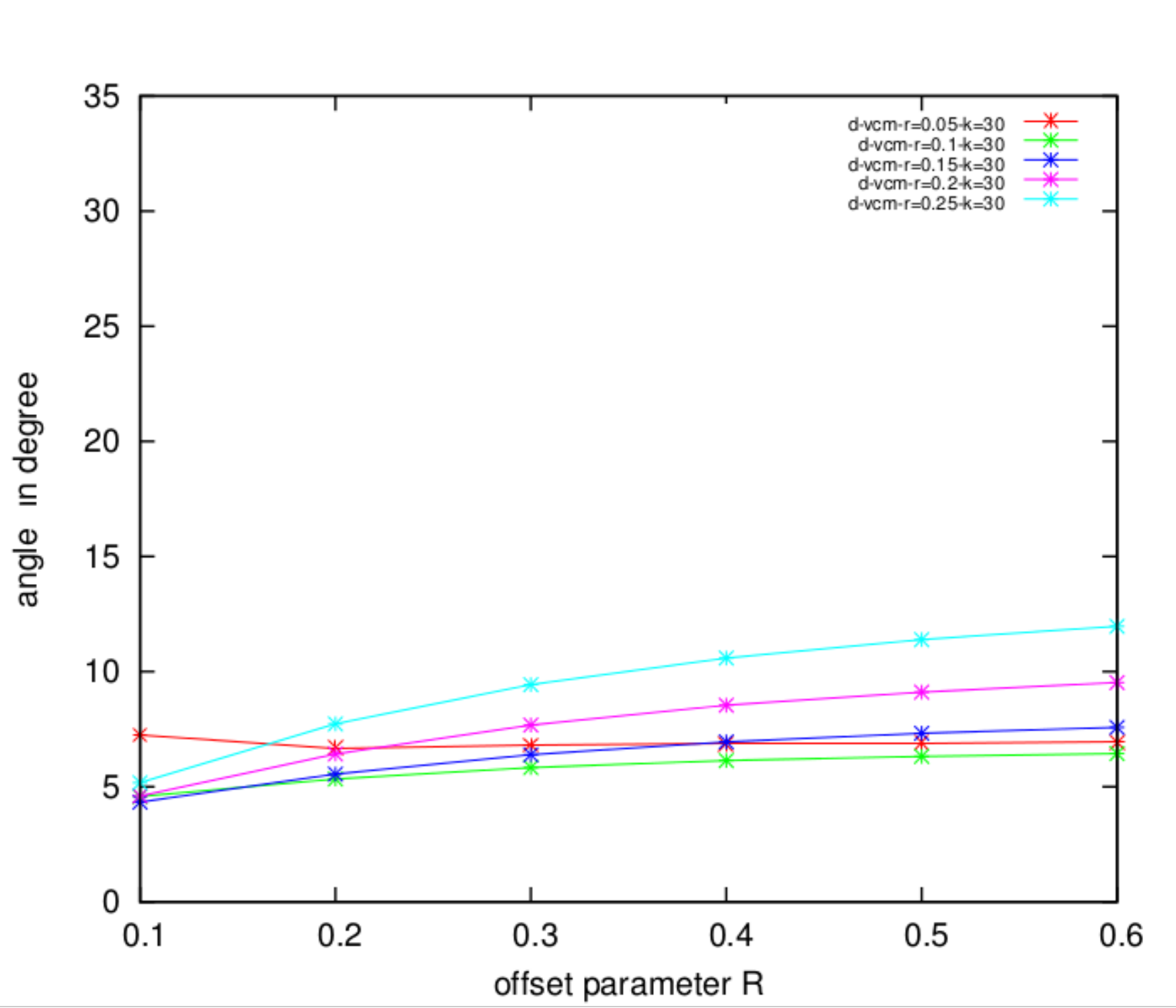}&
\includegraphics[width=0.49\textwidth]{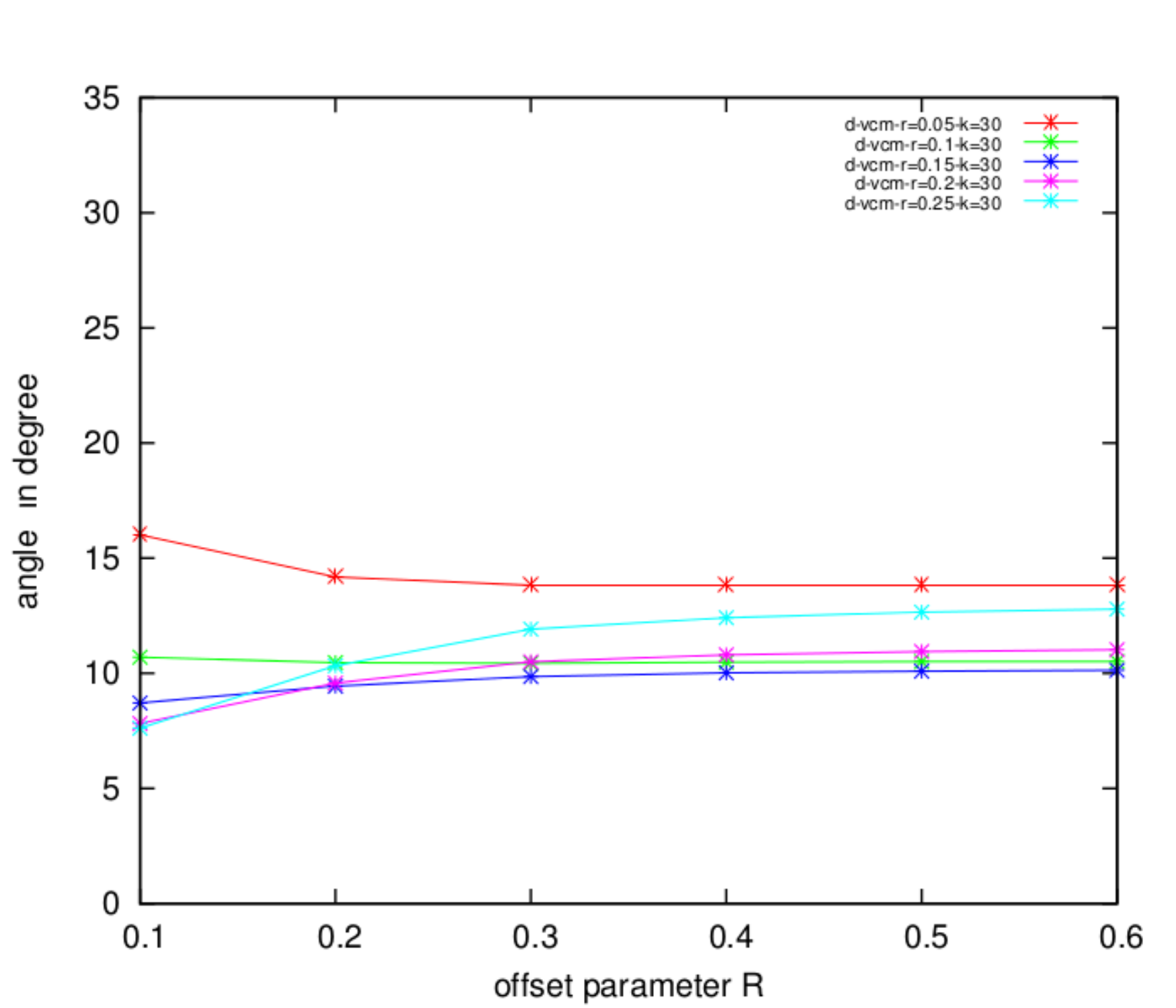}\\
\end{tabular}
\caption{Dependence on the parameter $R$: we see the influence of the
  parameter $R$ (in abscissa) on the average error on the normal estimation
  using $d^{\mathrm{w}}_{P,k}$-VCM. The experiment is done for
  different values of $r$ and for two different noisy
  ellipsoids. Left: Hausdorff noise $=0.2\diamnoise$. Right: Hausdorff noise
  $=0.4\diamnoise$. Here  $\diamnoise$ is the diameter of the original shape.}
\label{parameter2}
\end{figure}

\paragraph{Visualisation of normal estimation}

\begin{figure}[!h]  
\centering
\includegraphics[width=.19\linewidth]{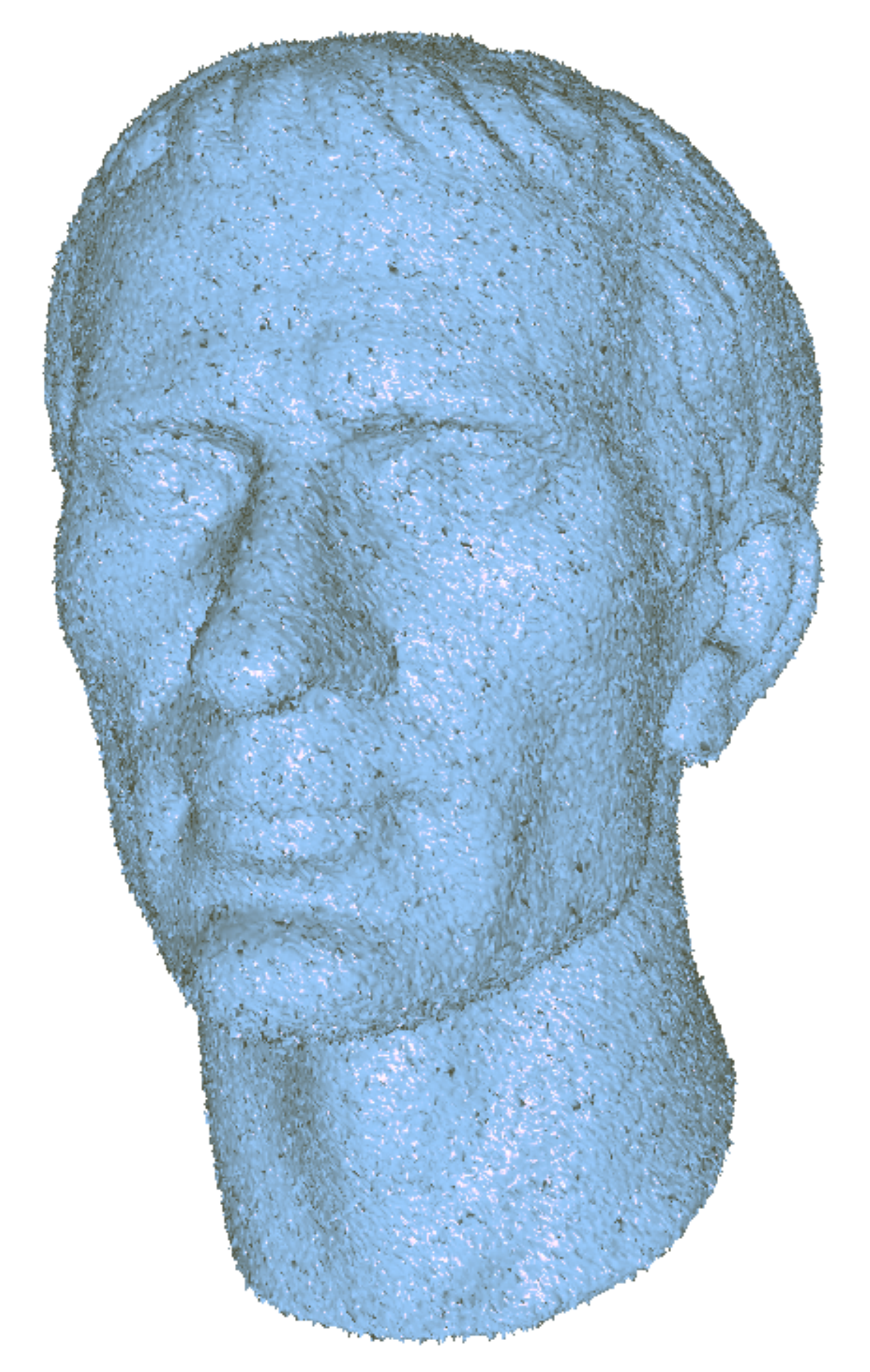}\hspace{0.9cm}
\includegraphics[width=.19\linewidth]{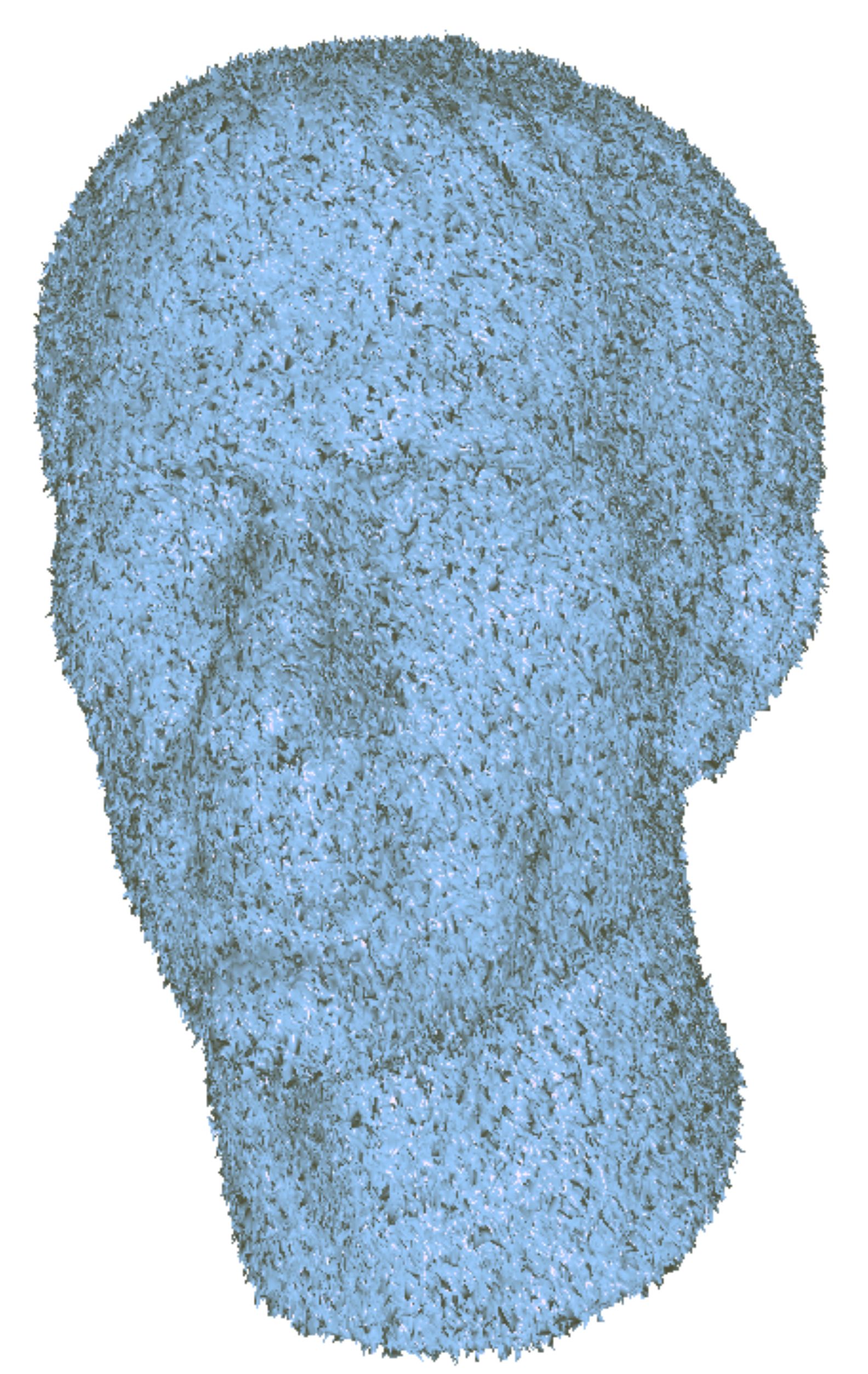}\hspace{0.9cm}
\includegraphics[width=.187\linewidth]{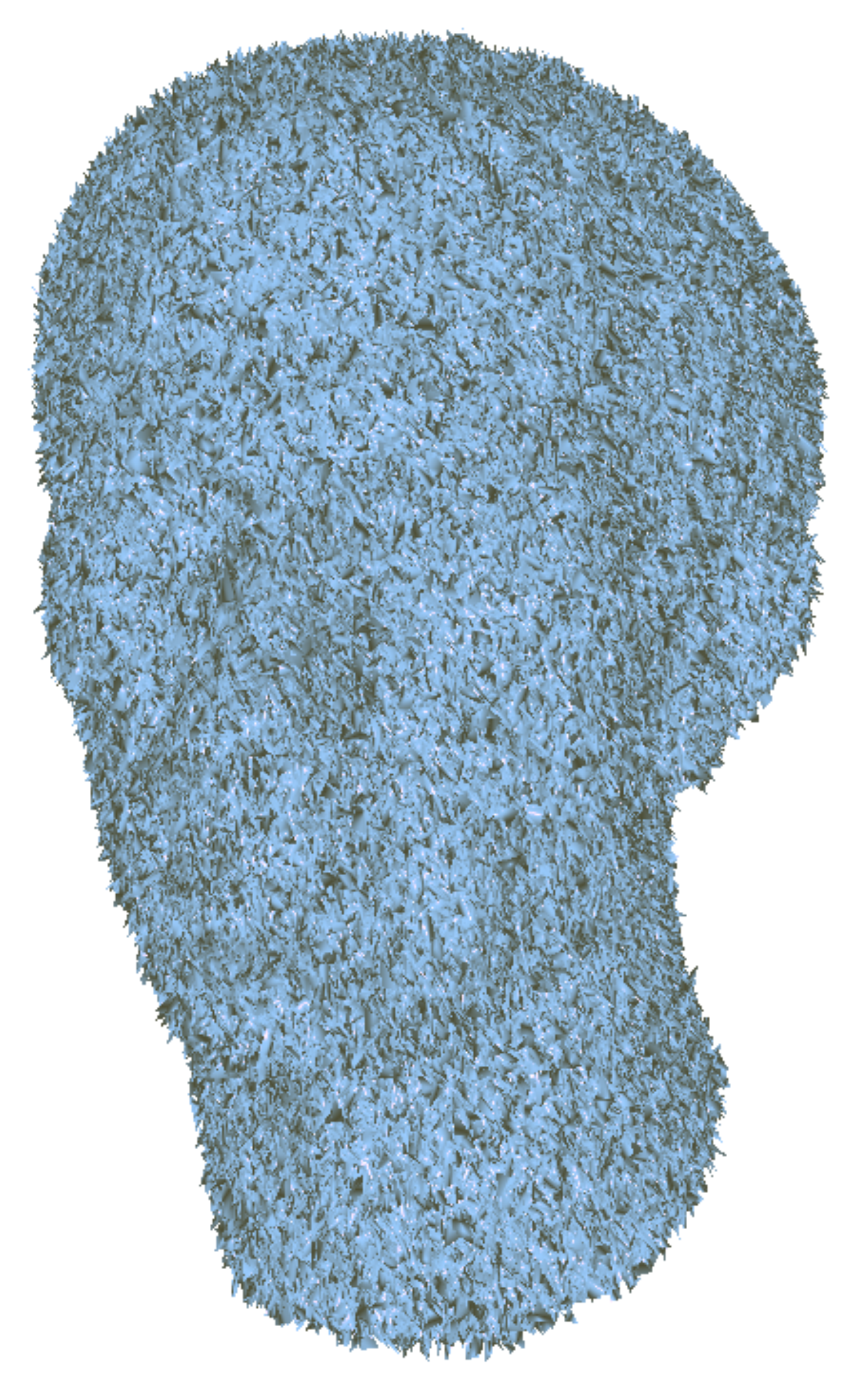}\\
\includegraphics[width=.19\linewidth]{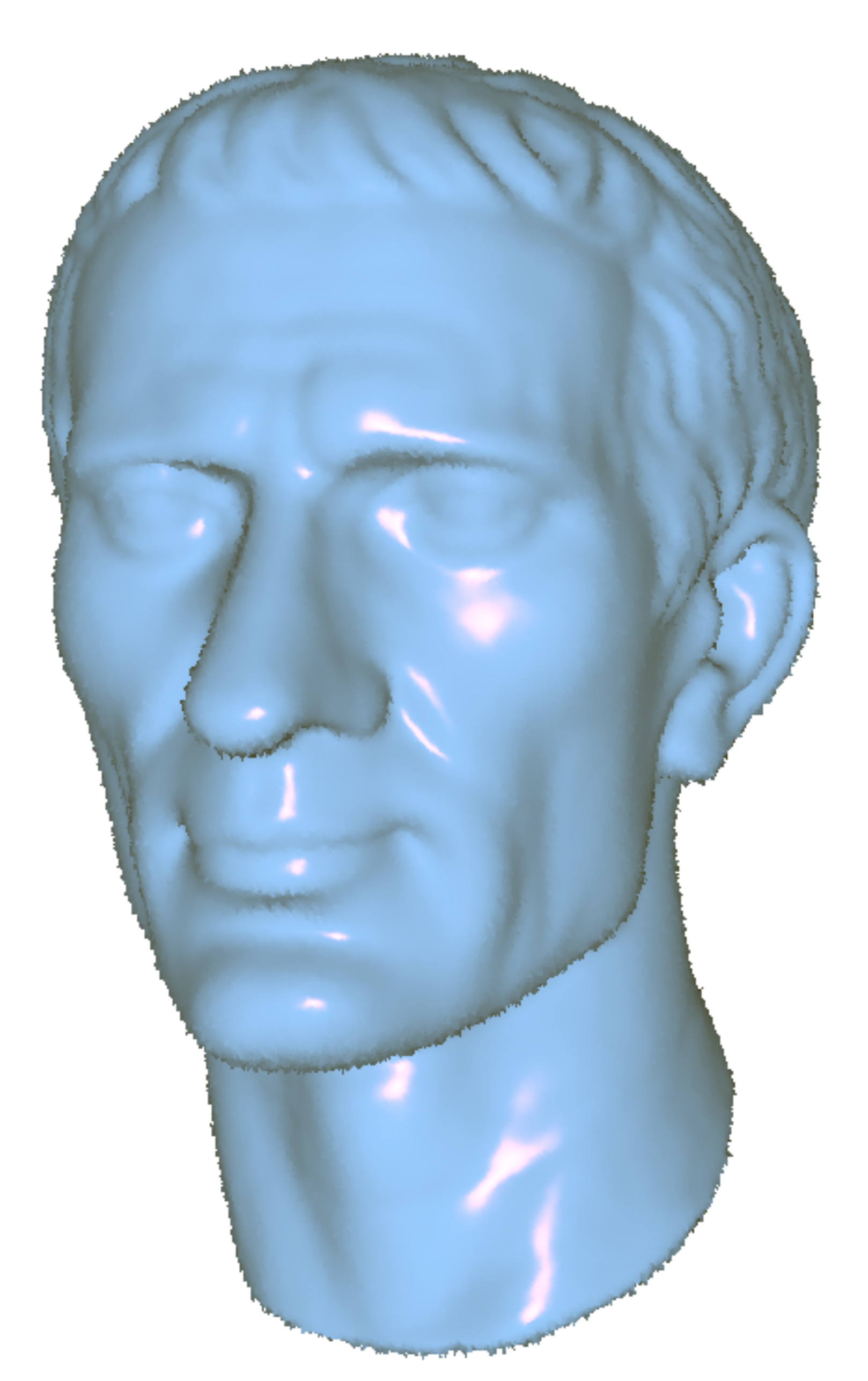}\hspace{0.9cm}
\includegraphics[width=.192\linewidth]{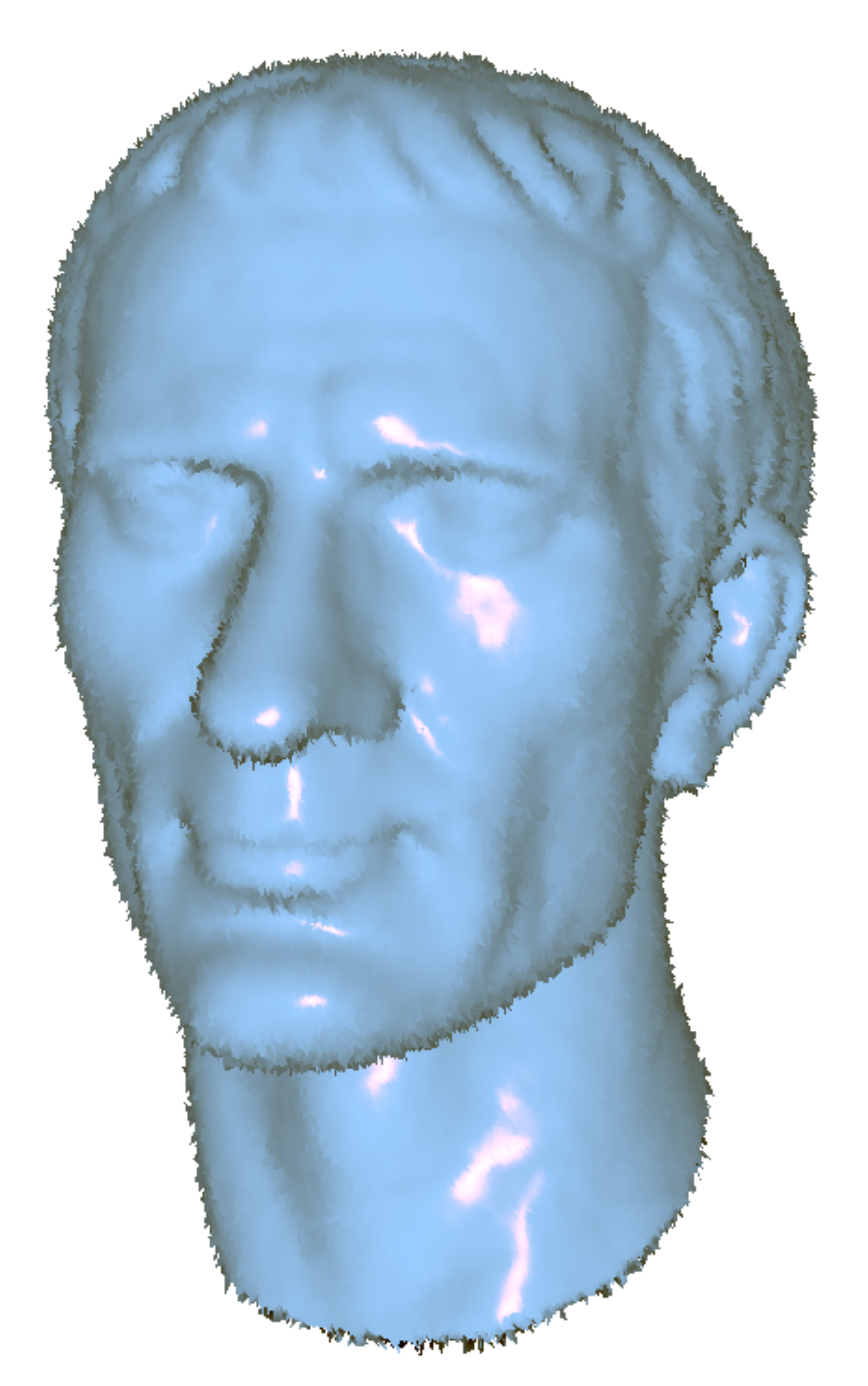}\hspace{0.9cm}
\includegraphics[width=.191\linewidth]{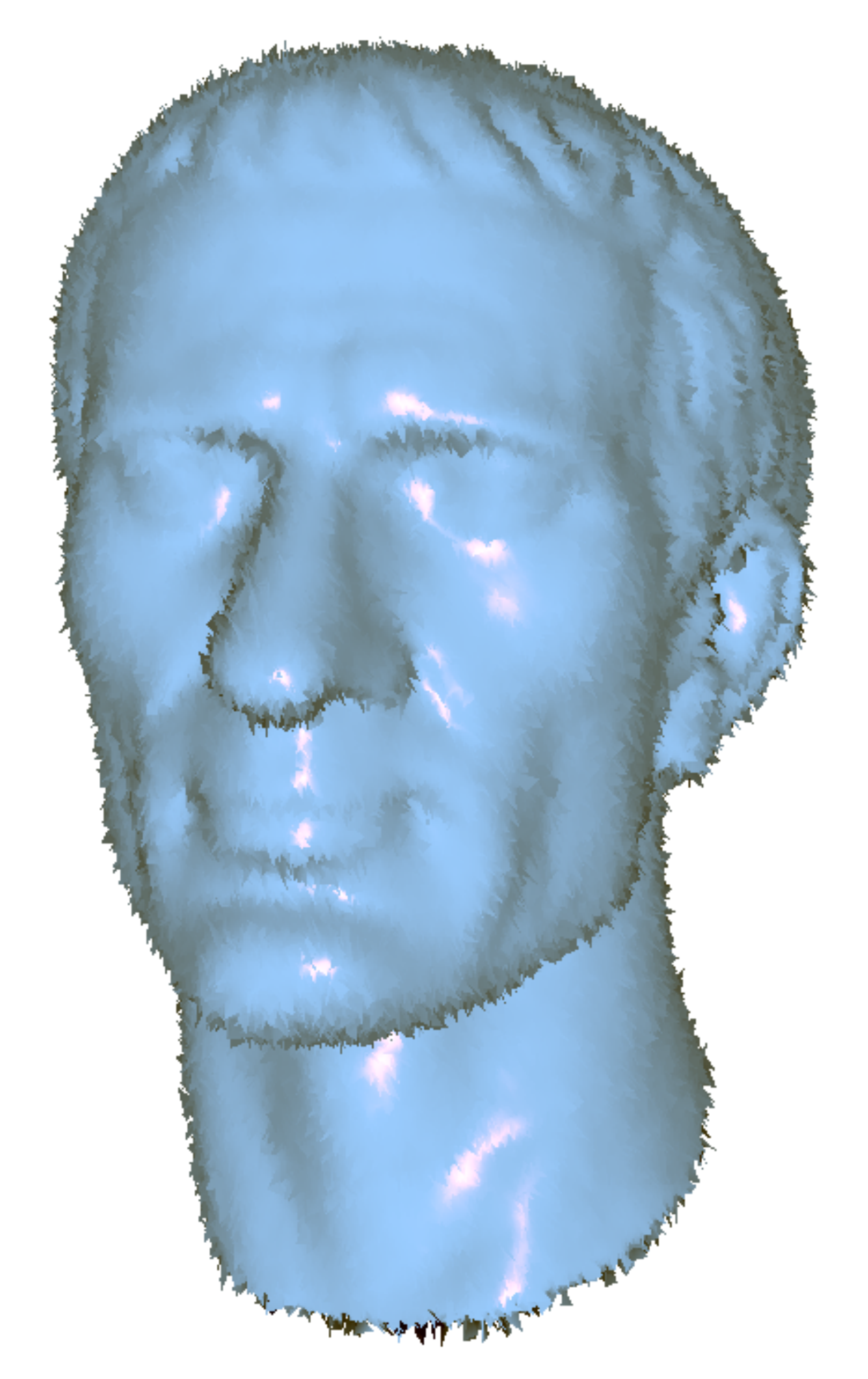}
\caption{Rendering of ``caesar'' data using triangle normal (top row)
  and estimated $d^{\mathrm{w}}_{P,k}$-VCM normal (bottom row) with
  Phong shading (parameters $R=0.04\diamnoise$, $r=0.04\diamnoise$, $k=30$, where $\diamnoise$ is the diameter of the original shape). From left to
  right, the Hausdorff noise is $0.02\diamnoise$, $0.04\diamnoise$ and $0.06\diamnoise$.}
\label{ceasarN}
\end{figure}



\begin{figure}[!h]  
\centering
\includegraphics[width=.2\linewidth]{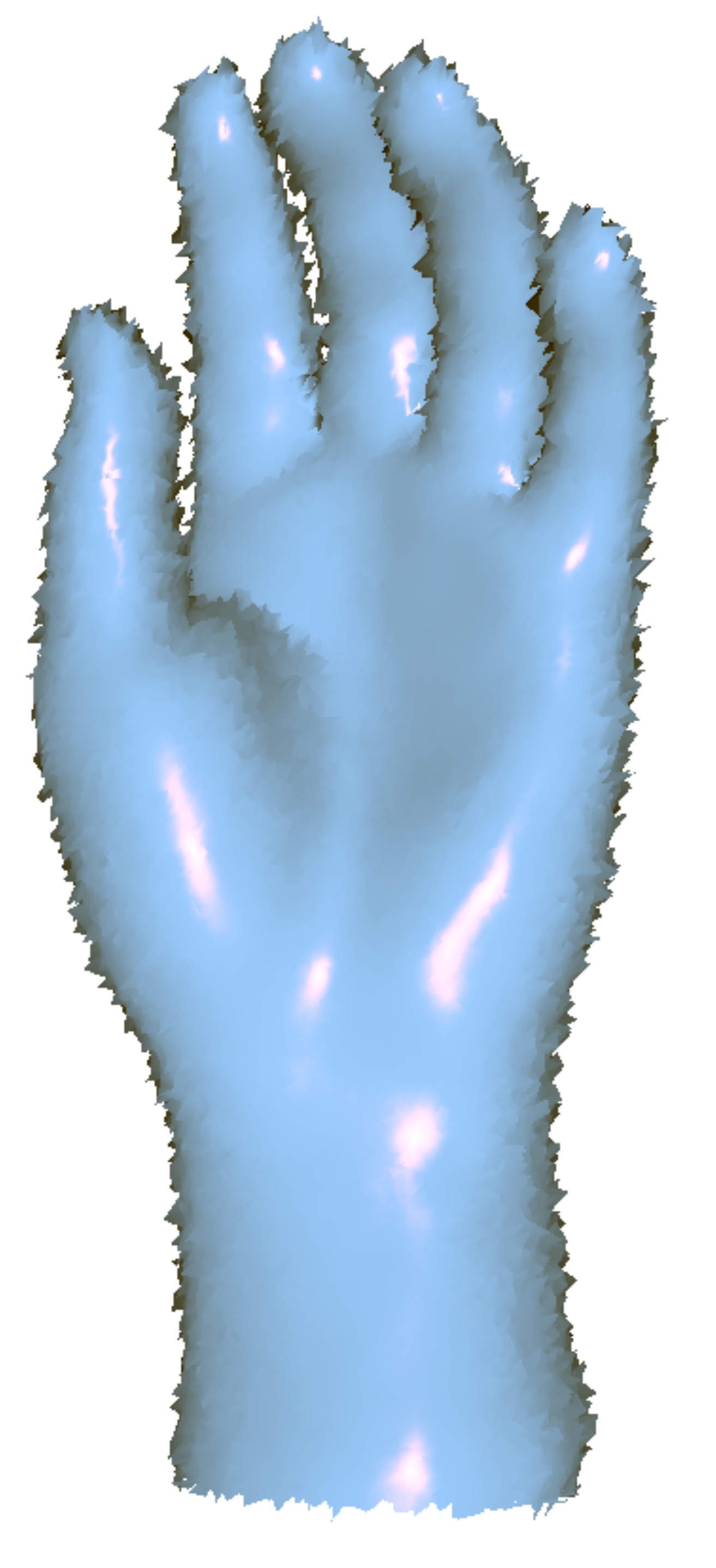}\hspace{0.9cm}
\includegraphics[width=.195\linewidth]{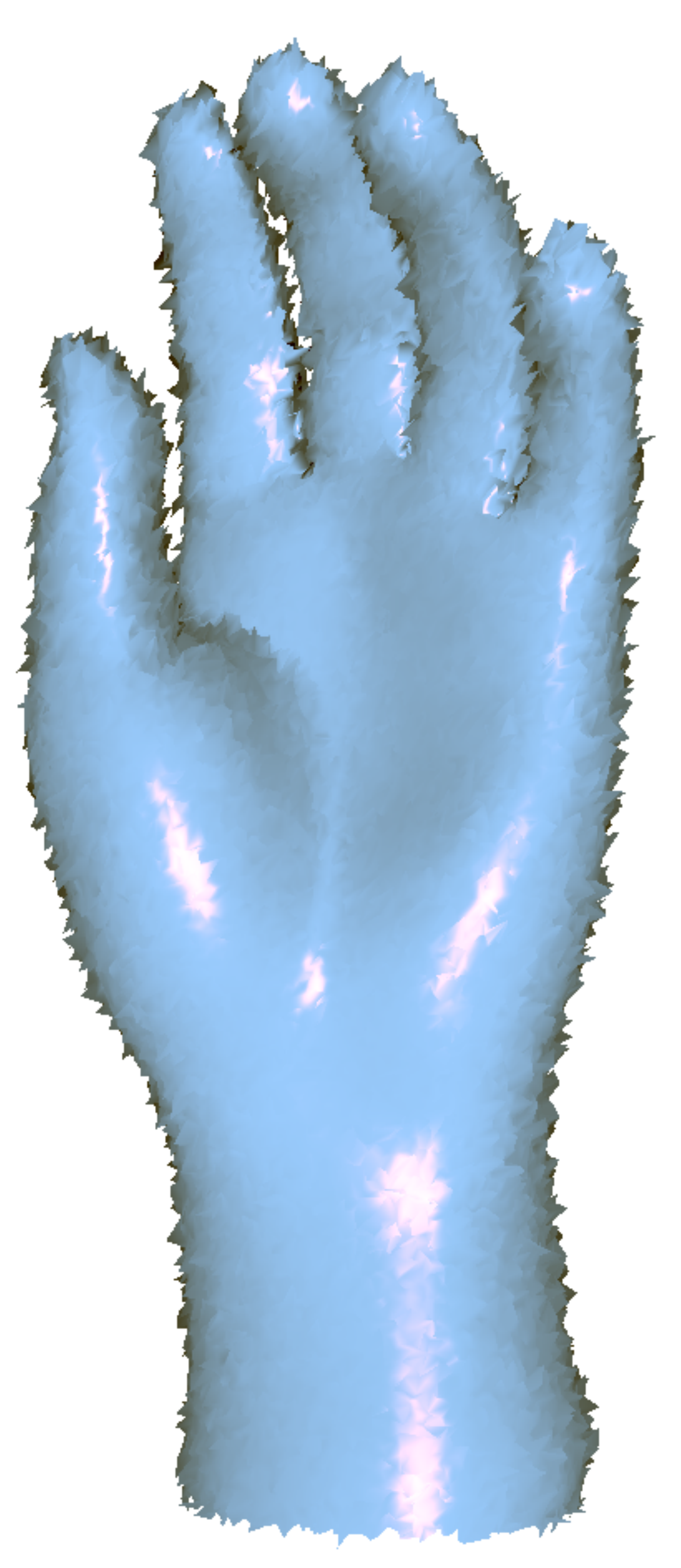}
\caption{Rendering comparison of ``hand'' data: (left) using our
  $d^{\mathrm{w}}_{P,k}$-VCM normal estimation with parameters
  $R=0.04\diamnoise$, $r=0.02\diamnoise$, $k=30$ ($\diamnoise$ is the diameter of the original shape), and (right) the Randomized Hough
  Transform method \cite{sgp-2012} with standard parameters: neighbors
  number $= 500$ and the others equal to default values. Triangles are
  displayed with Phong shading and ``hand'' is perturbated by a
  Hausdorff noise equal to $0.04\diamnoise$.}
\label{handN}
\end{figure}

We test the quality of our normal estimator on standard shapes, by
rendering meshes according to their estimated normals. In Figure
\ref{ceasarN}, we notice that the rendering done with normals computed
with our method is much better than the rendering done with normals
induced by the geometry of the underlying mesh, at the same time
robust to noise while keeping intact significant features. In Figure
\ref{handN}, we compare this rendering with the rendering obtained
when the normal is calculated by the Randomized Hough Transform
method \cite{sgp-2012}, which is a statistic method known to be
 resilient to noise and outliers.  Our method achieves a much
smoother rendering.

\subsubsection{Estimation of curvatures and features detection} 
The covariance matrix also carries curvature information along other
eigendirections \cite{vcm}. We denote by $\lambda_0 \geq \lambda_1
\geq \lambda_2$ the three eigenvalues of
$\mathcal{V}_{\delta_{P},R}(\chi_p^r)$ at a point $p$. Up to a
multiplicative constant, $\lambda_1$ and $\lambda_2$ correspond to the
absolute value of the minimal and maximal curvatures respectively
\cite{vcm}. We call the corresponding eigenvectors respectively
minimal and maximal principal directions.

We compare our method with the Jet Fitting method \cite{Pouget} in
Figure \ref{bimba1}. Experiments have been done for a large choice of
different parameters, and we find better principal directions with our method. In the presence of many outliers, we plot in Figure
\ref{bimba2} the minimal principal direction estimation of
$d^{\mathrm{w}}_{P,k}$-VCM projected on the initial mesh.

\begin{figure}[!h] 
\centering
\begin{tabular}{cc}
\includegraphics[width=.4172\linewidth]{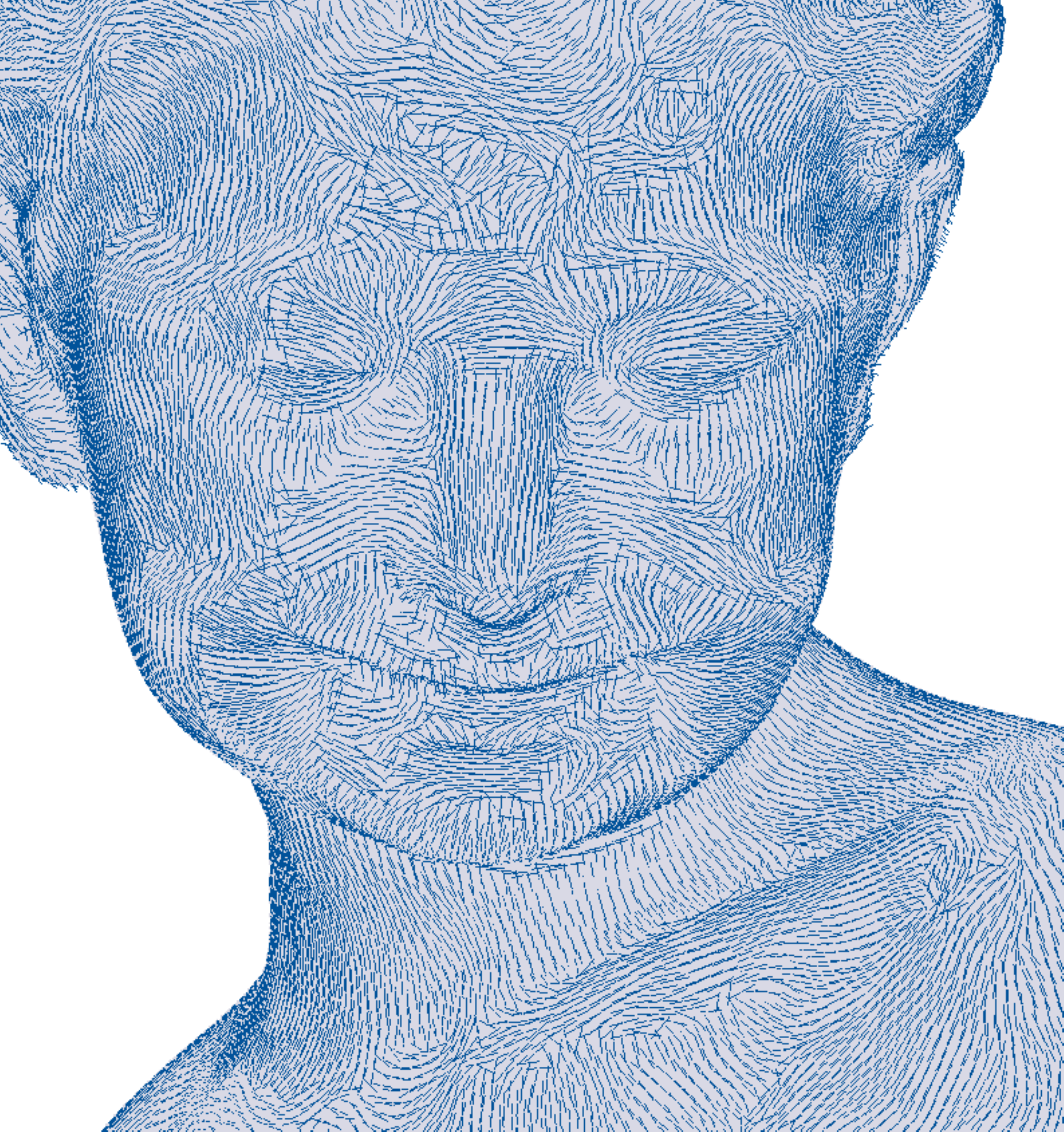}&
\includegraphics[width=.4\linewidth]{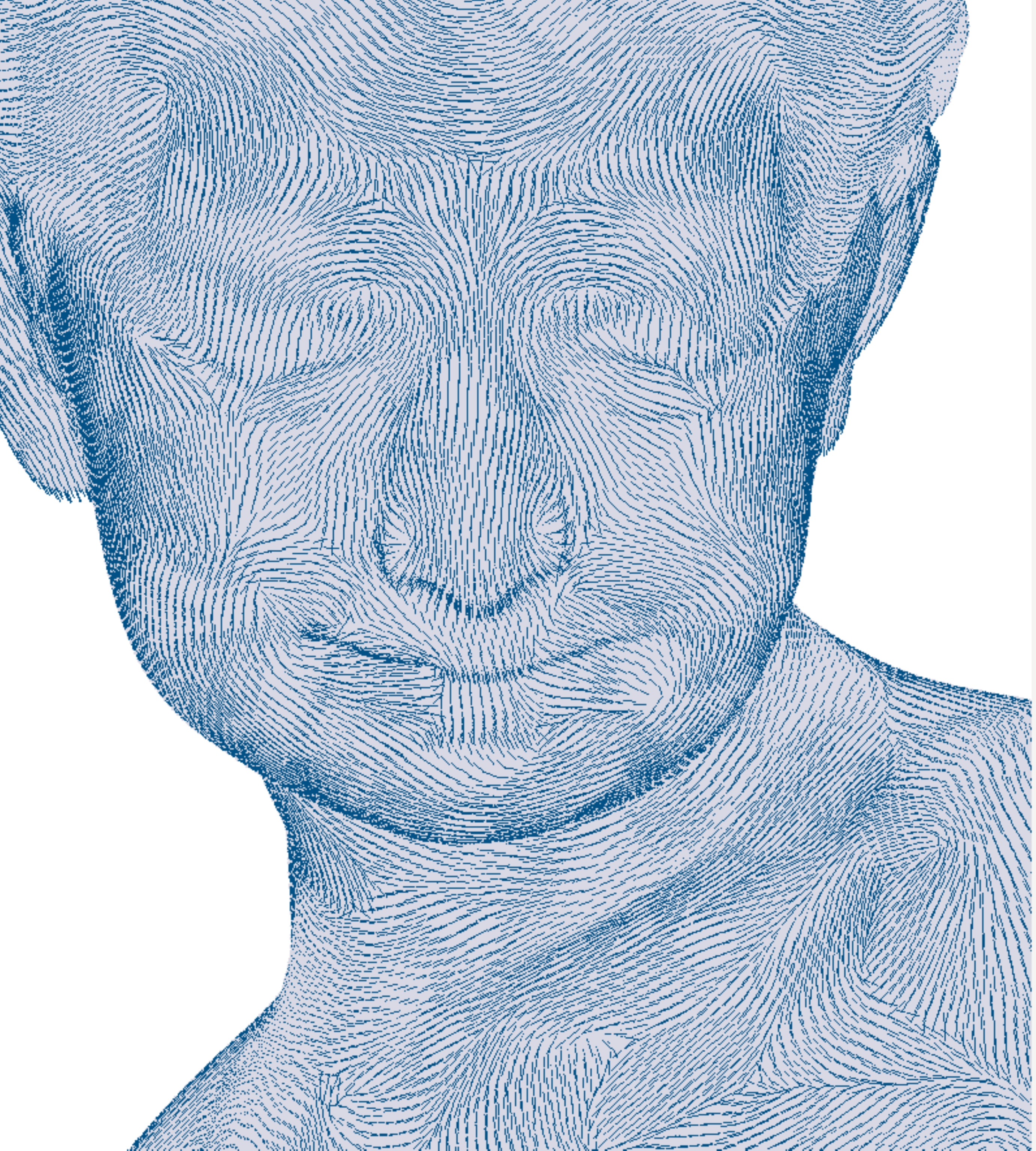}\\
\end{tabular}
\caption{Principal direction estimation on Bimba datas (no noise). We put a line-segment aligned at every point to the minimal principal direction. Left: Jet-fitting method (with a polynomial of degree $4\times 4$ and $k=100$ neighbors for each point). Right: $d^{\mathrm{w}}_{P,k}$-VCM method (parameters $R=0.04\diamnoise$, $r=0.08\diamnoise$, $k=30$, where $\diamnoise$ is the diameter of the original shape.)}
\label{bimba1}
\end{figure}

\begin{figure}[!h] 
\centering
\begin{tabular}{cc}
\includegraphics[width=.4\linewidth]{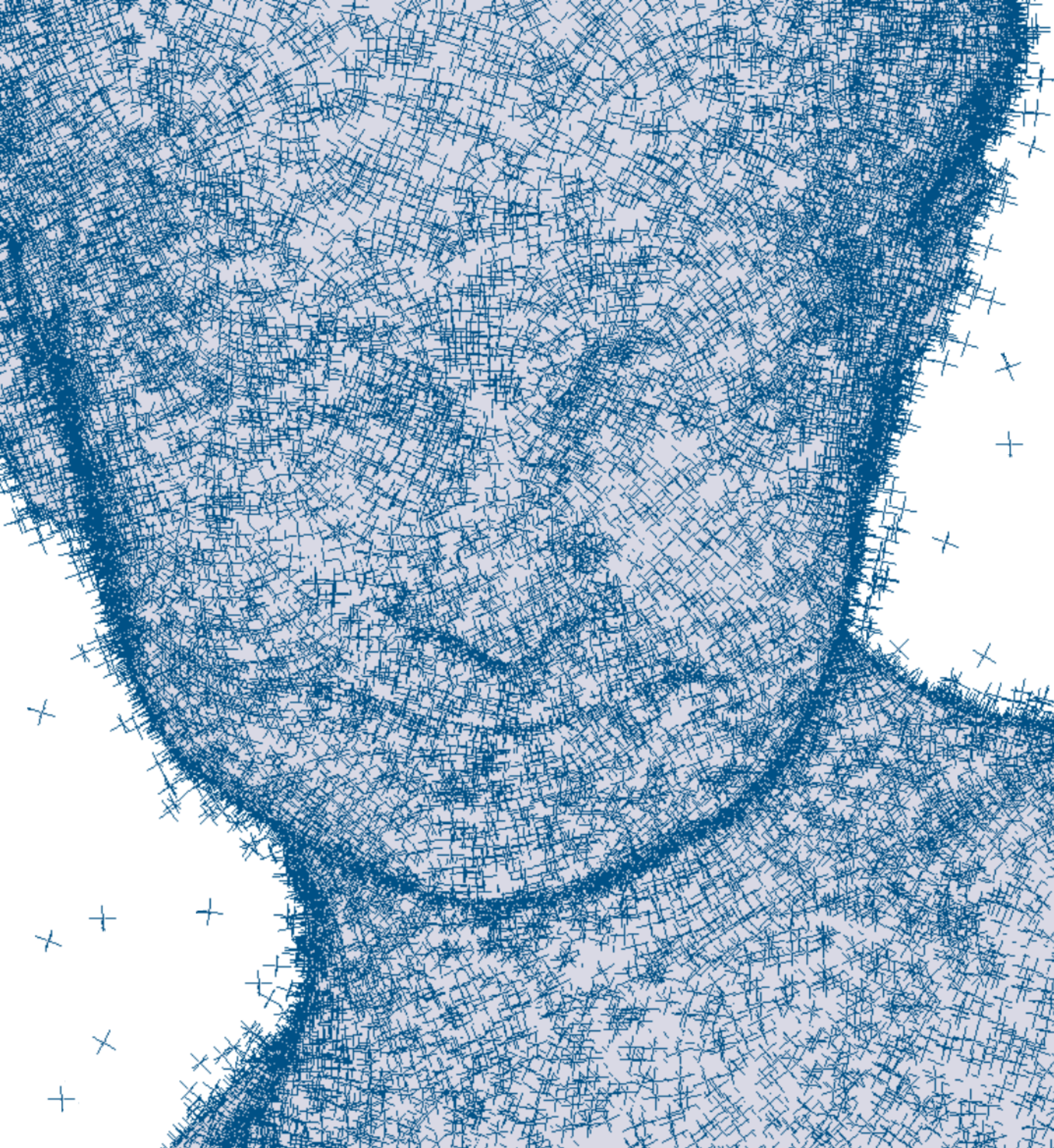}&
\includegraphics[width=.405\linewidth]{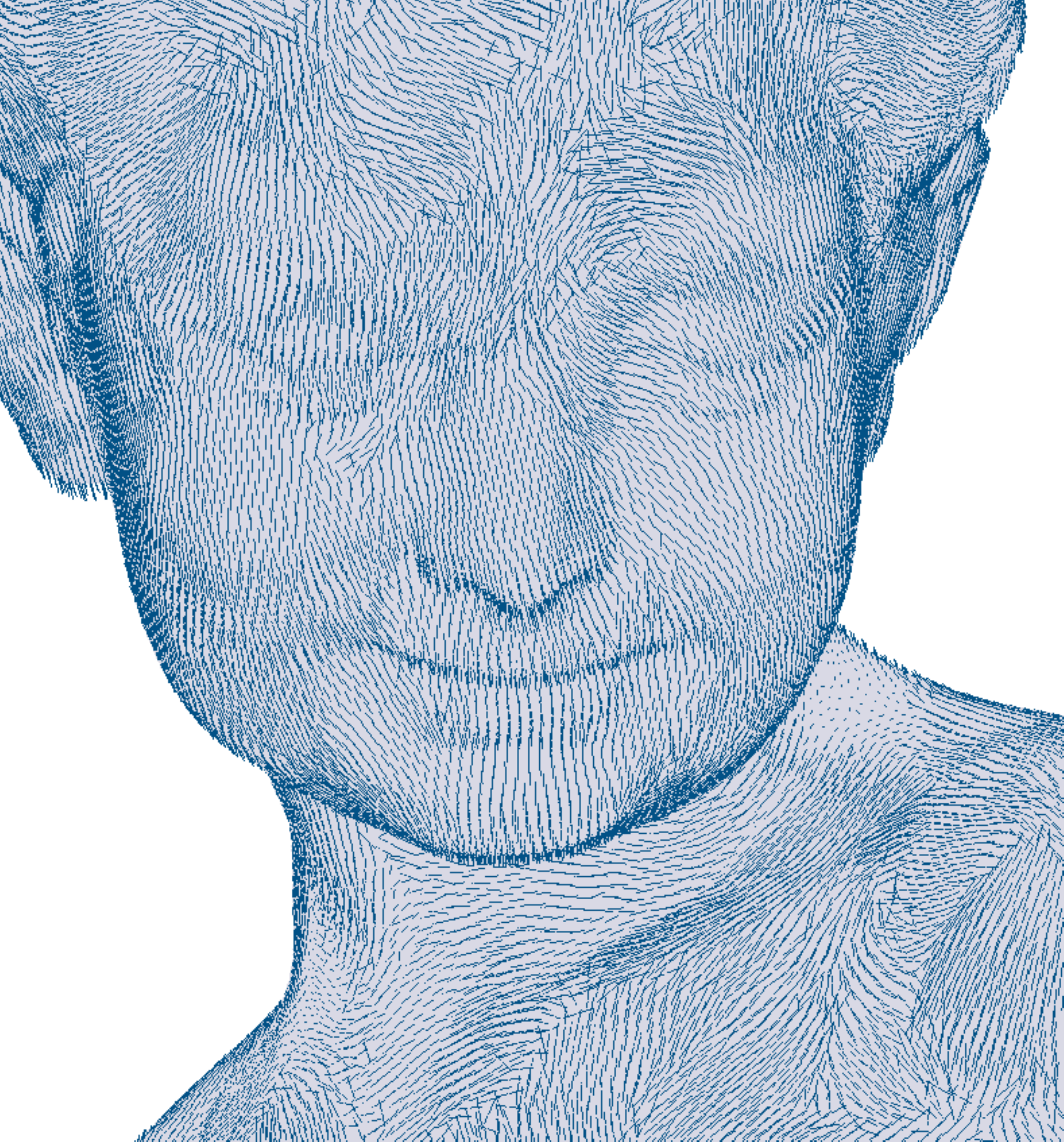}\\
\end{tabular}
\caption{Principal direction estimation on Bimba datas of diameter $\diamnoise$. Left: input data with outliers ($80\%$ of the points are moved at a distance at most $0.02\diamnoise$, $10\%$ are moved at a distance between $0.02\diamnoise$ and $0.1\diamnoise$ and $10\%$ are outliers taken randomly in the bounding box). Right: for every vertex, we project the minimal principal direction estimation of the $d^{\mathrm{w}}_{P,k}$-VCM on the initial mesh.}

\label{bimba2}
\end{figure}

\begin{figure}[!h]
\begin{tabular}{cc}
\includegraphics[width=.19\linewidth]{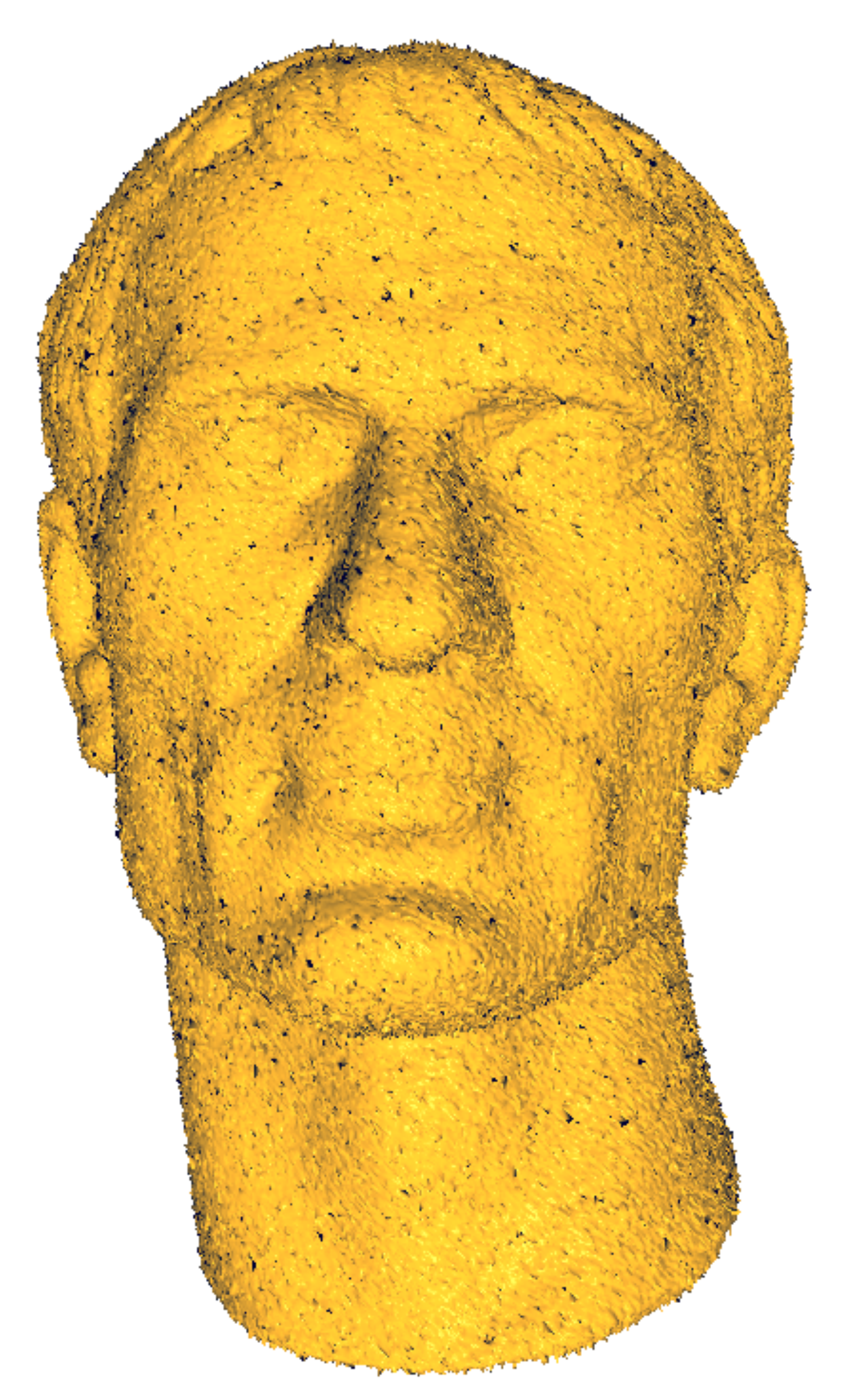}
\includegraphics[width=.202\linewidth]{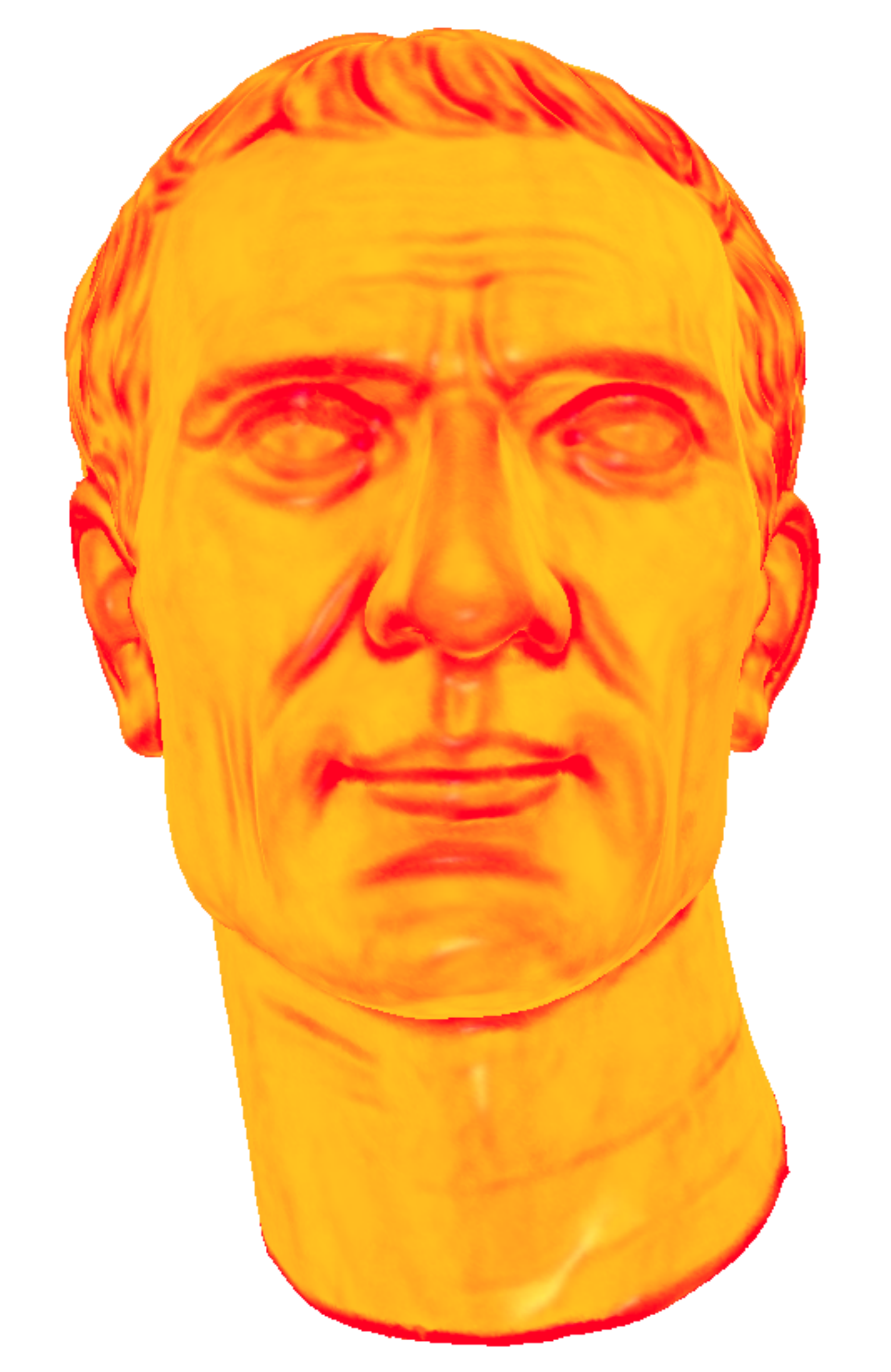} \hspace{0.3cm}&\hspace{0.3cm}
\includegraphics[width=.191\linewidth]{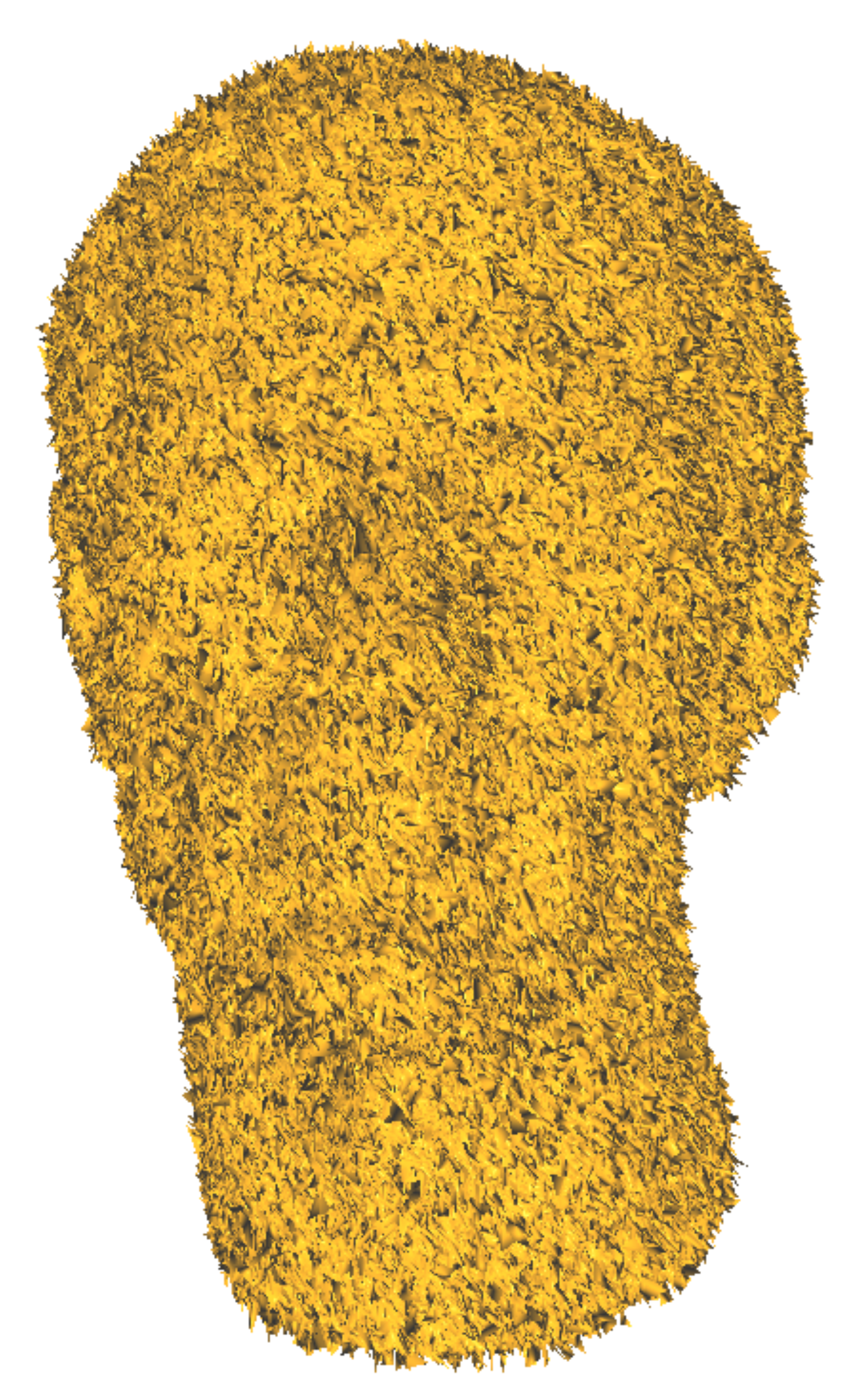}
\includegraphics[width=.19\linewidth]{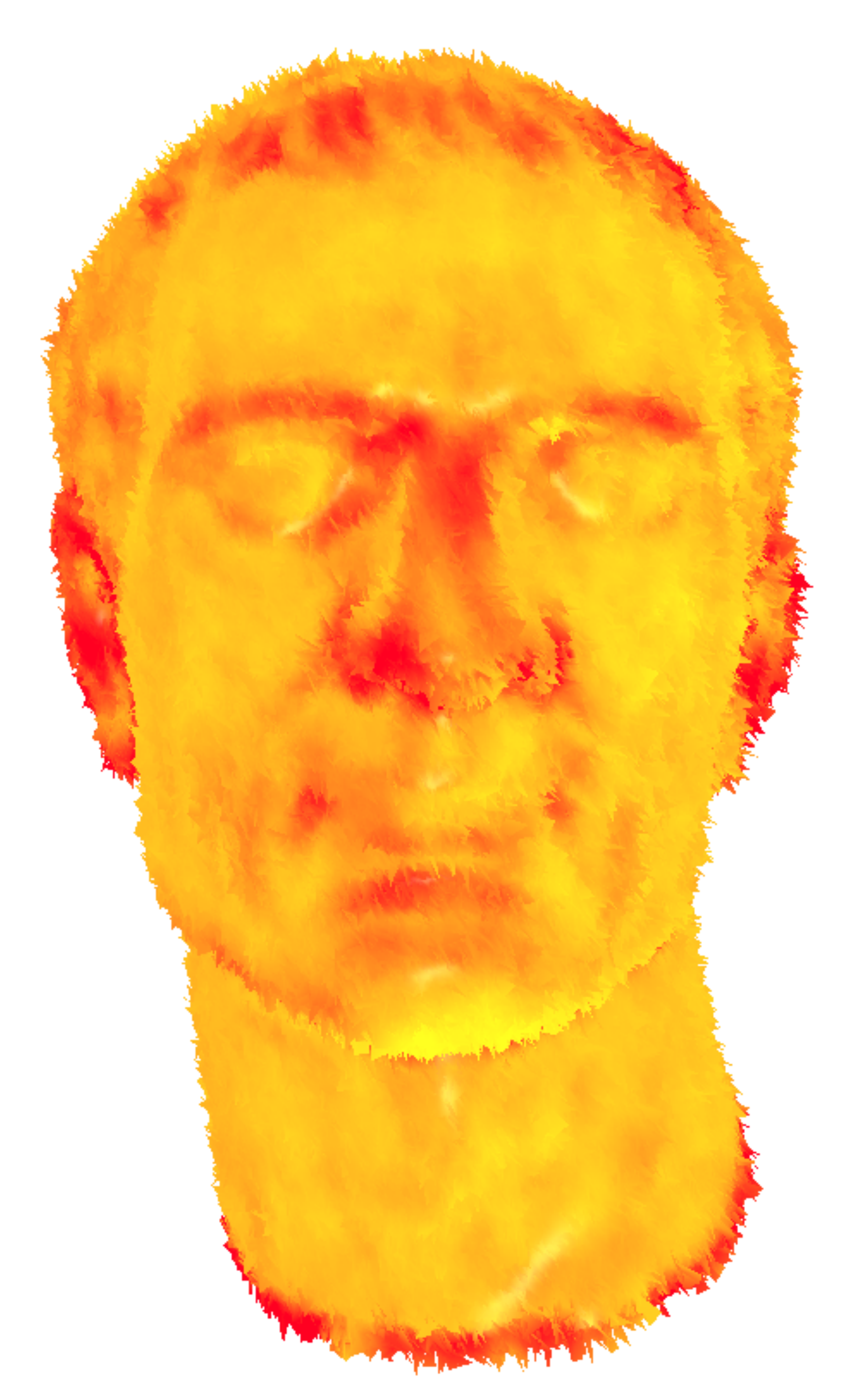}\\
\end{tabular}
\caption{Mean absolute curvature on ``ceasar'' data (of diameter $\diamnoise$) estimated with $d^{\mathrm{w}}_{P,k}$-VCM (yellow is low curvature while red is high curvature). Parameters are set as $R=0.04\diamnoise$, $r=0.04\diamnoise$, $k=30$. From left to right: input data with Hausdorff noise $=0.02\diamnoise$; corresponding curvature estimation; input data with Hausdorff noise $=0.06\diamnoise$; corresponding curvature estimation.}
\label{ceasar-curv}
\end{figure}

\begin{figure}[!h]
\centering
\includegraphics[width=.15\linewidth]{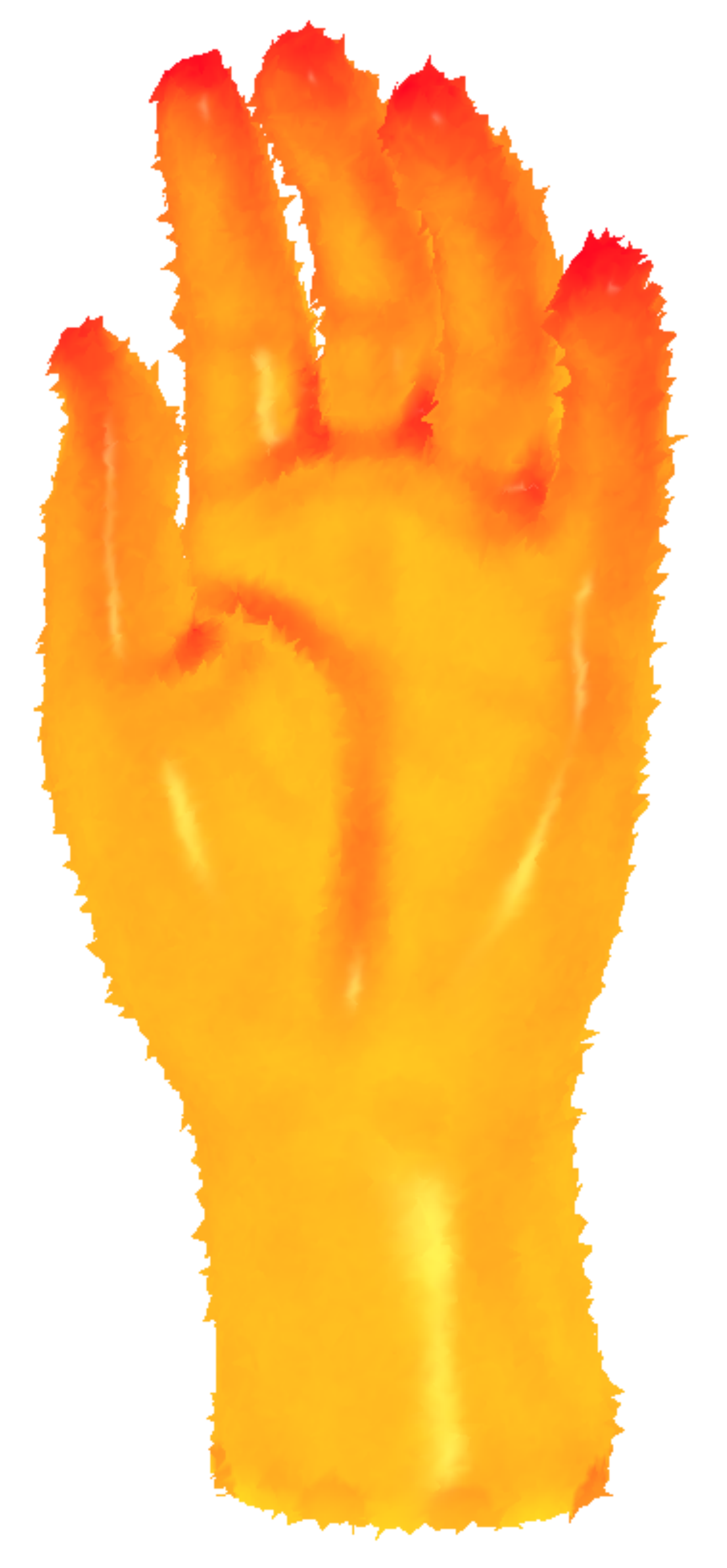}\hspace{0.9cm}
\includegraphics[width=.15\linewidth]{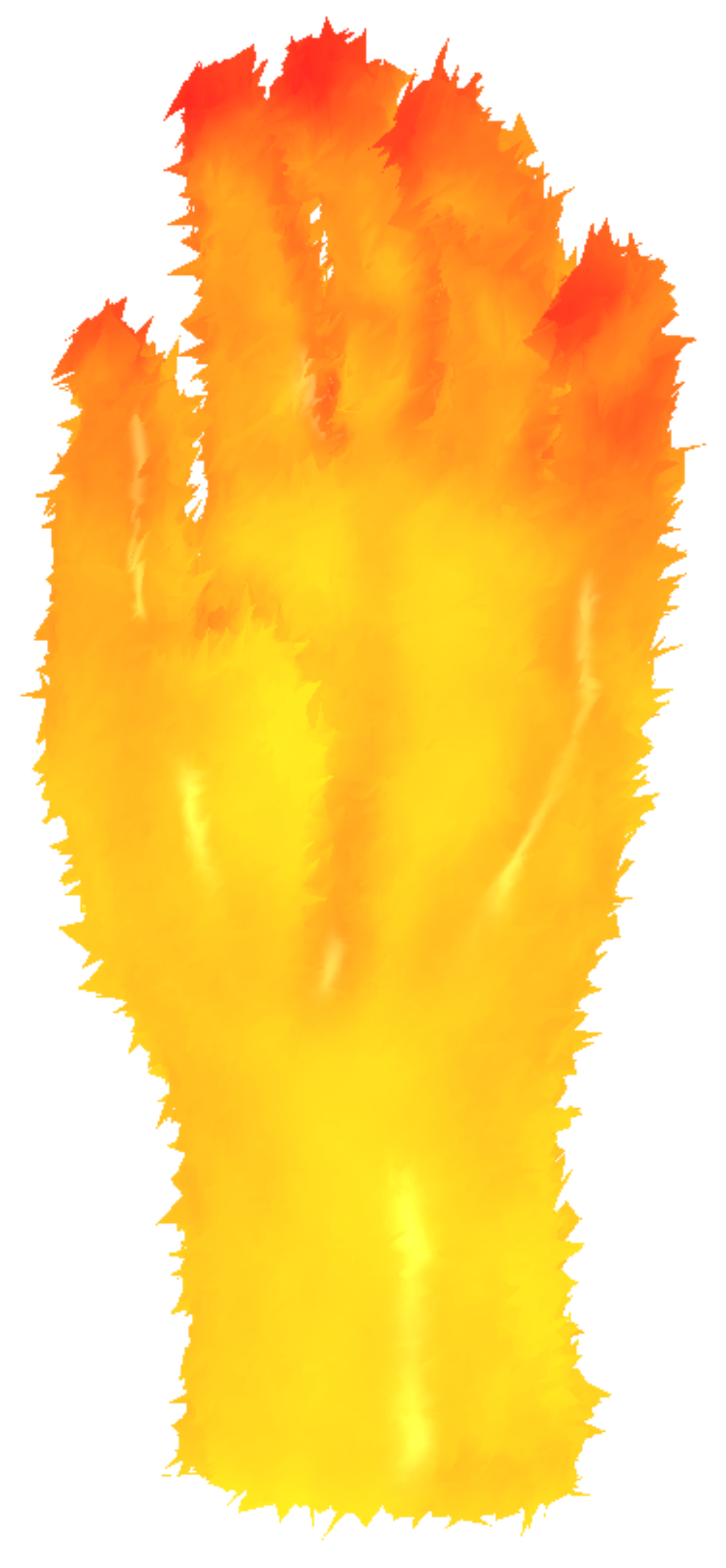}\hspace{0.9cm}
\includegraphics[width=.15\linewidth]{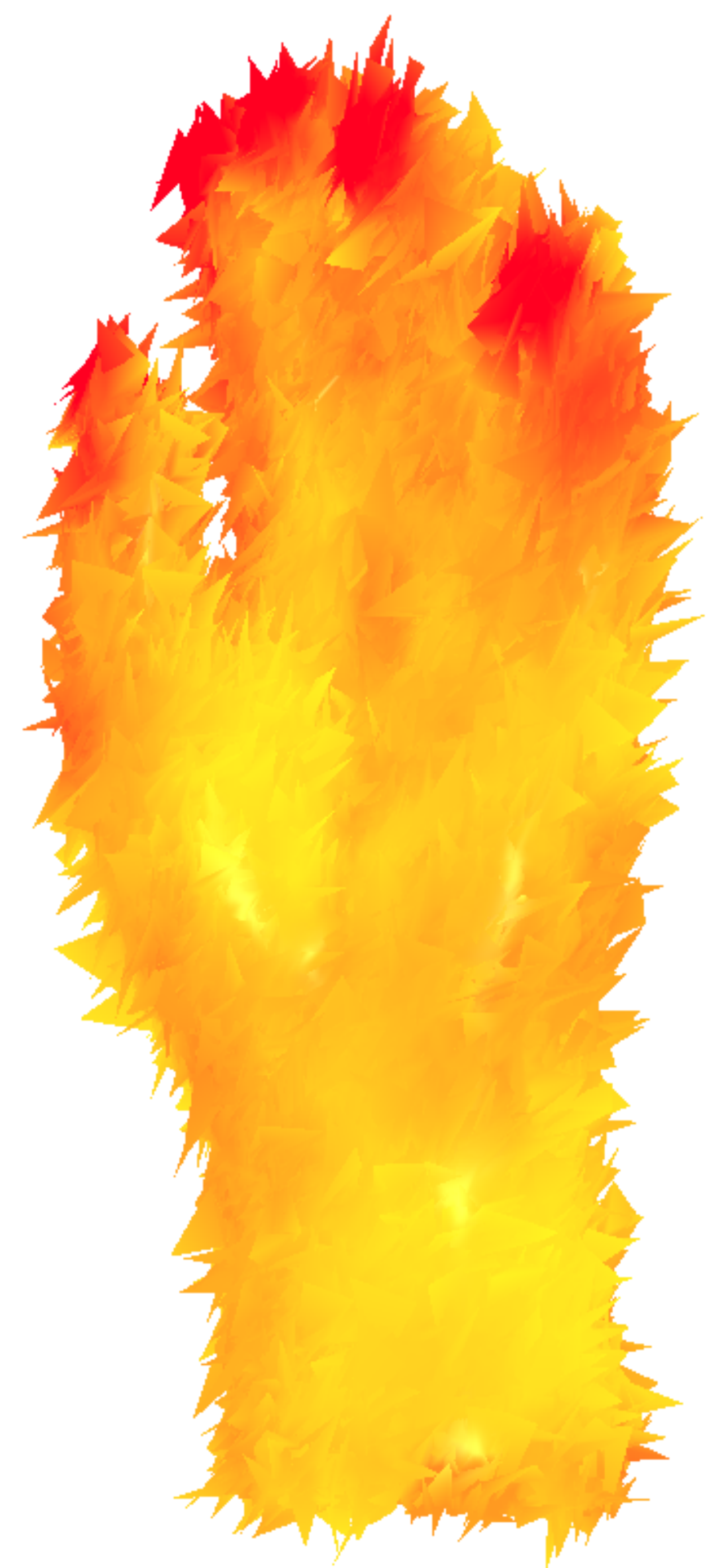}
\caption{Mean absolute curvature on ``hand'' data estimated with $d^{\mathrm{w}}_{P,k}$-VCM (yellow is low
  curvature while red is high curvature). Parameters are set as
  $R=0.04\diamnoise$, $r=0.02\diamnoise$, $k=30$. From left to right, the Hausdorff noise
  is $0.04\diamnoise$, $0.08\diamnoise$ and $0.16\diamnoise$, where $\diamnoise$ is the diameter of the original shape.}
\label{hand-curv}
\end{figure}

To illustrate the performance of curvature estimation by
$d^{\mathrm{w}}_{P,k}$-VCM, the estimated mean absolute curvature
$\lambda_1 + \lambda_2 $ is displayed for the standard ``caesar'' and
``hand'' models (Figures~\ref{ceasar-curv} and \ref{hand-curv}). On
Figure \ref{ceasar-curv}, we also render the triangulation before any
processing to better illustrate the amplitude of noise.


\begin{figure}
\begin{tabular}{ccc}

\includegraphics[width=.28\linewidth]{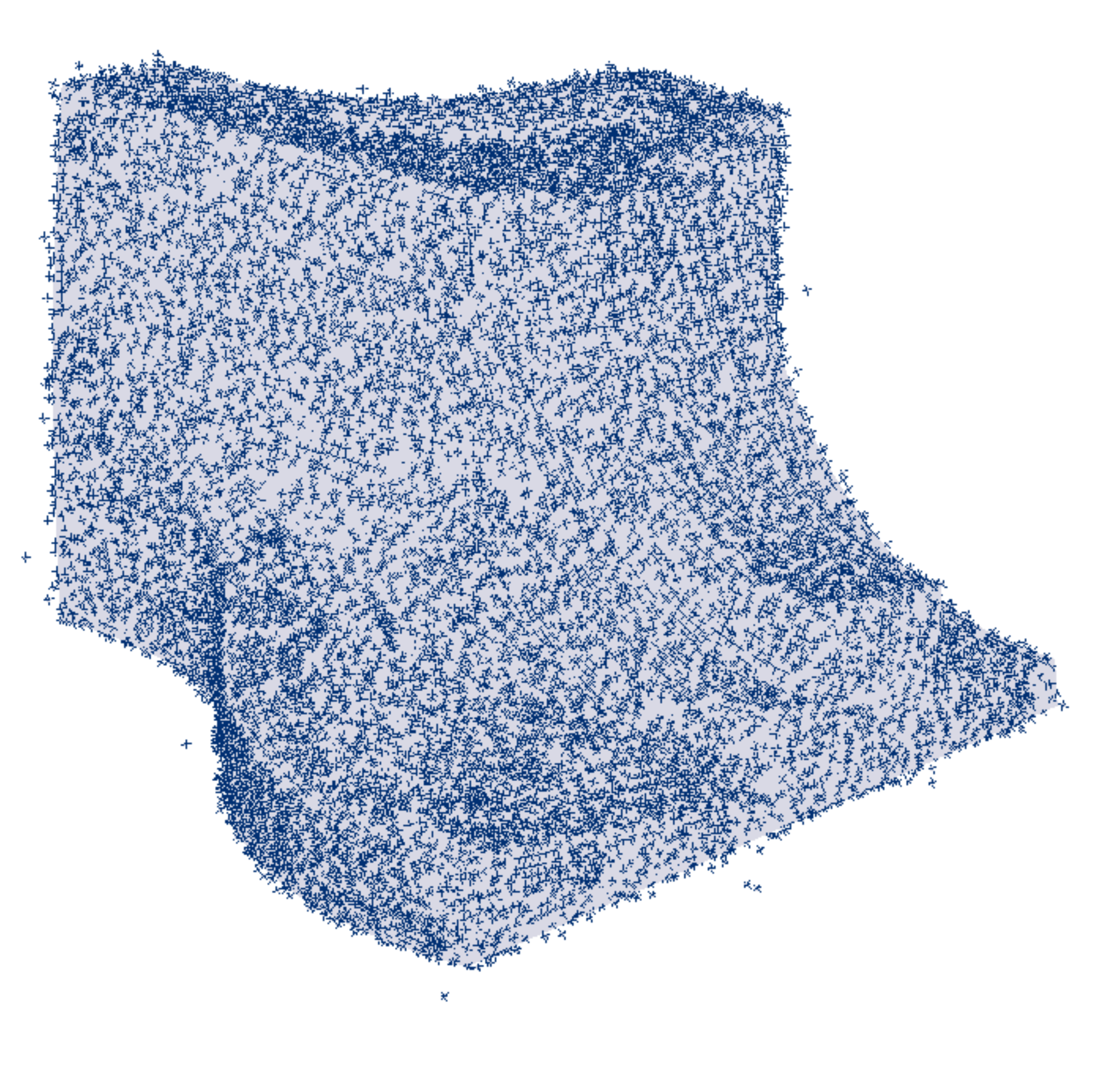}&
\includegraphics[width=.28\linewidth]{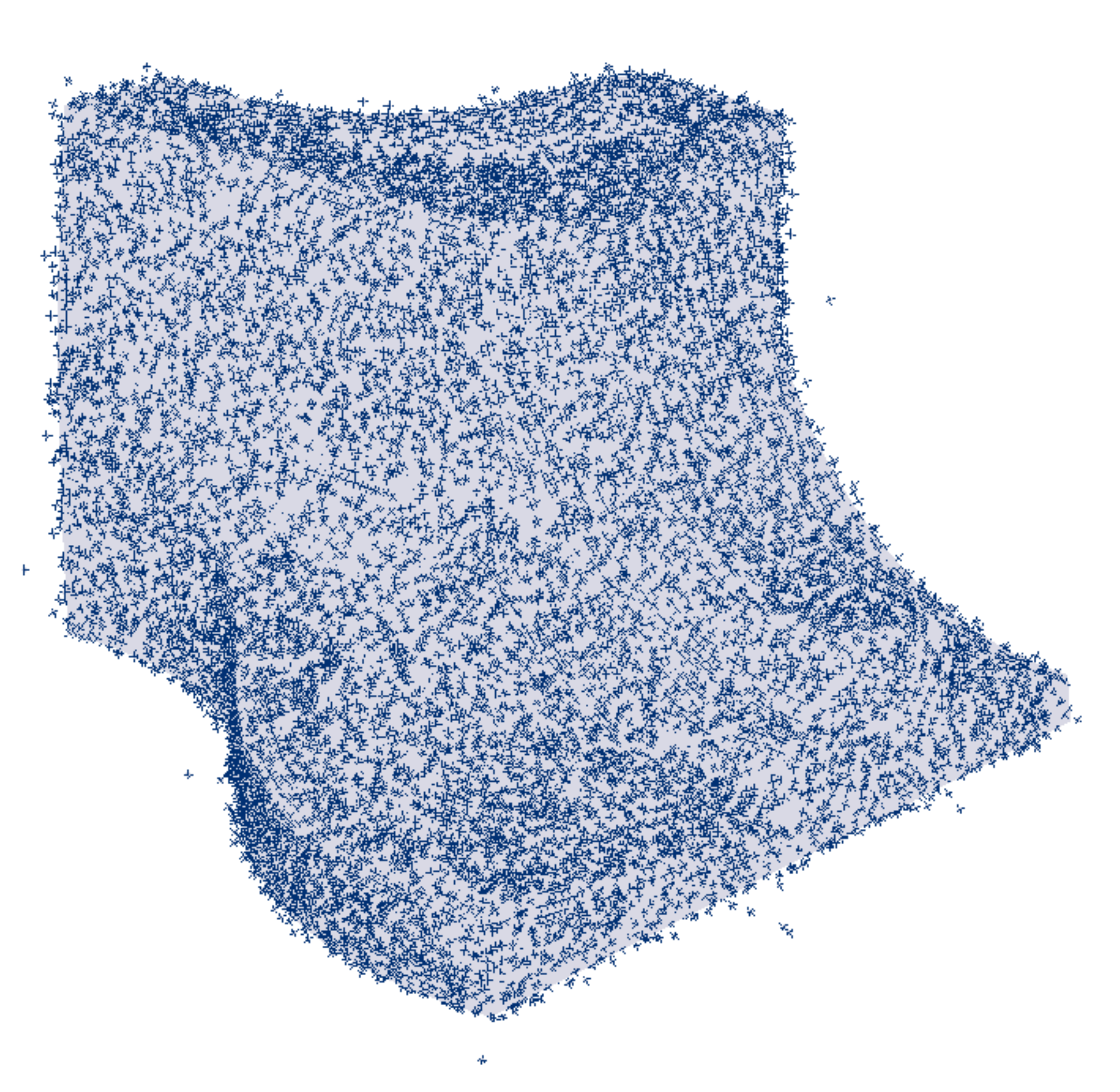}&
\includegraphics[width=.28\linewidth]{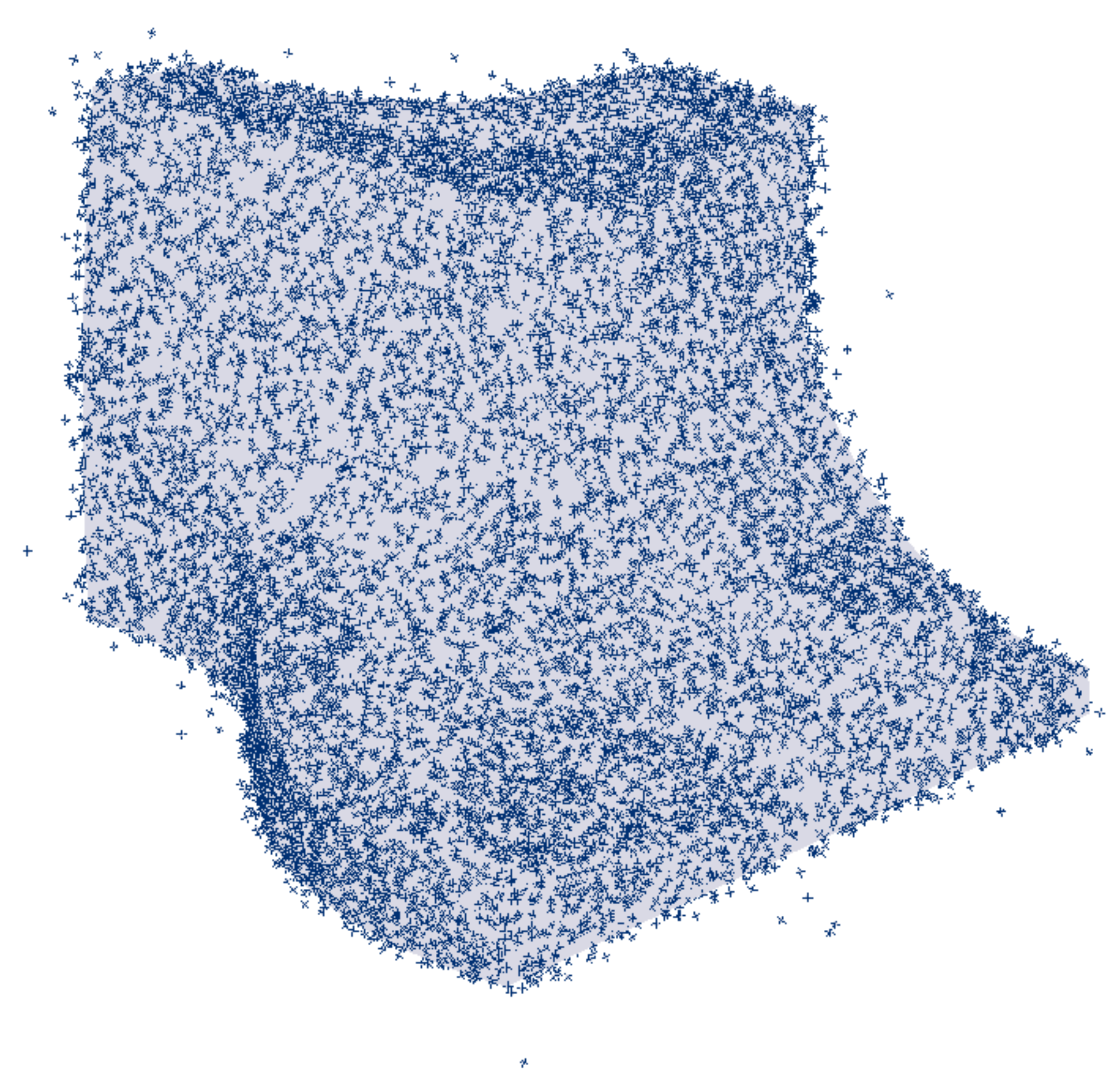}\\
Input 1 &
Input 2 &
 Input 3 \\
 ~&
~&
~\\
\includegraphics[width=.28\linewidth]{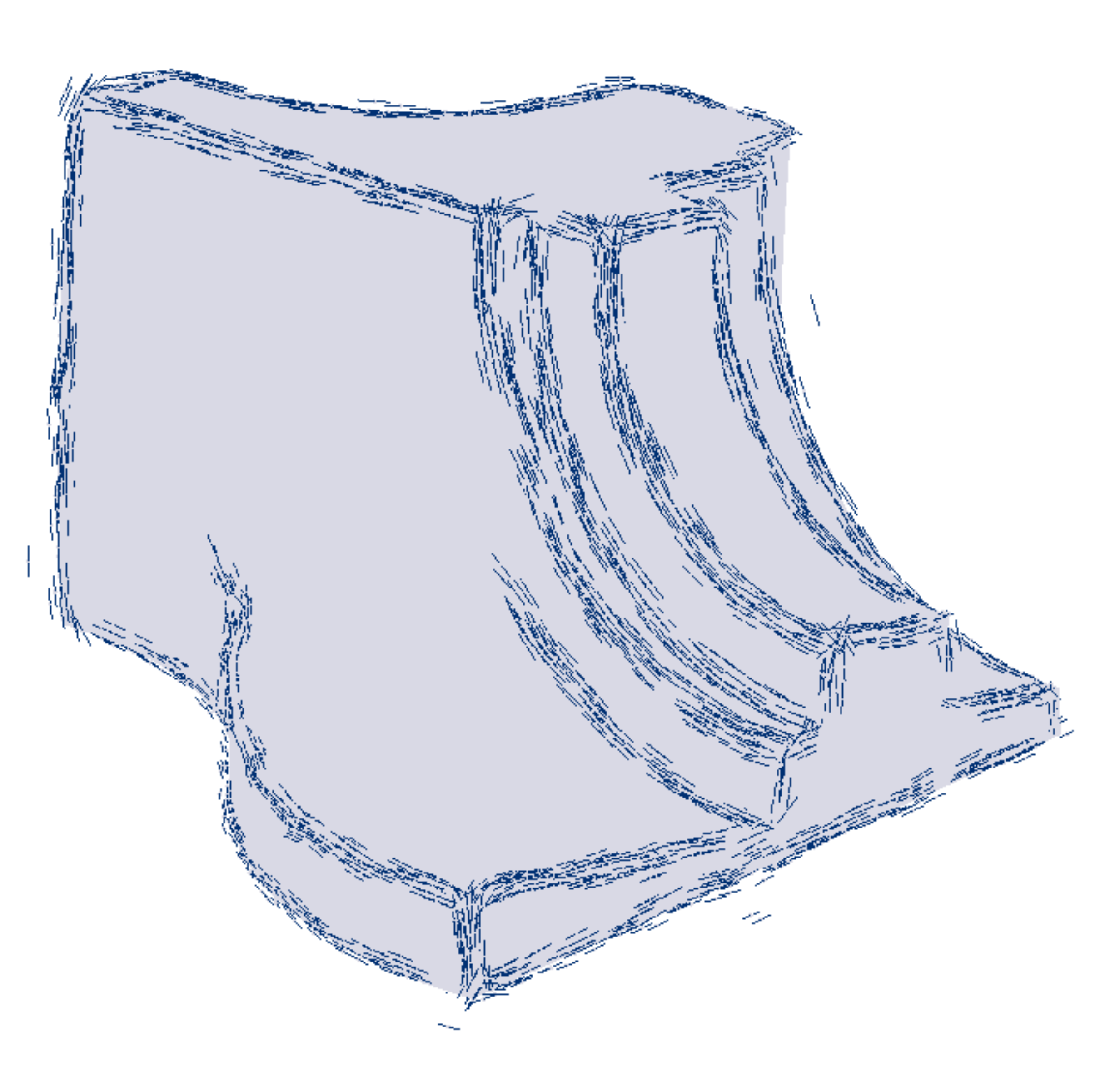}&
\includegraphics[width=.28\linewidth]{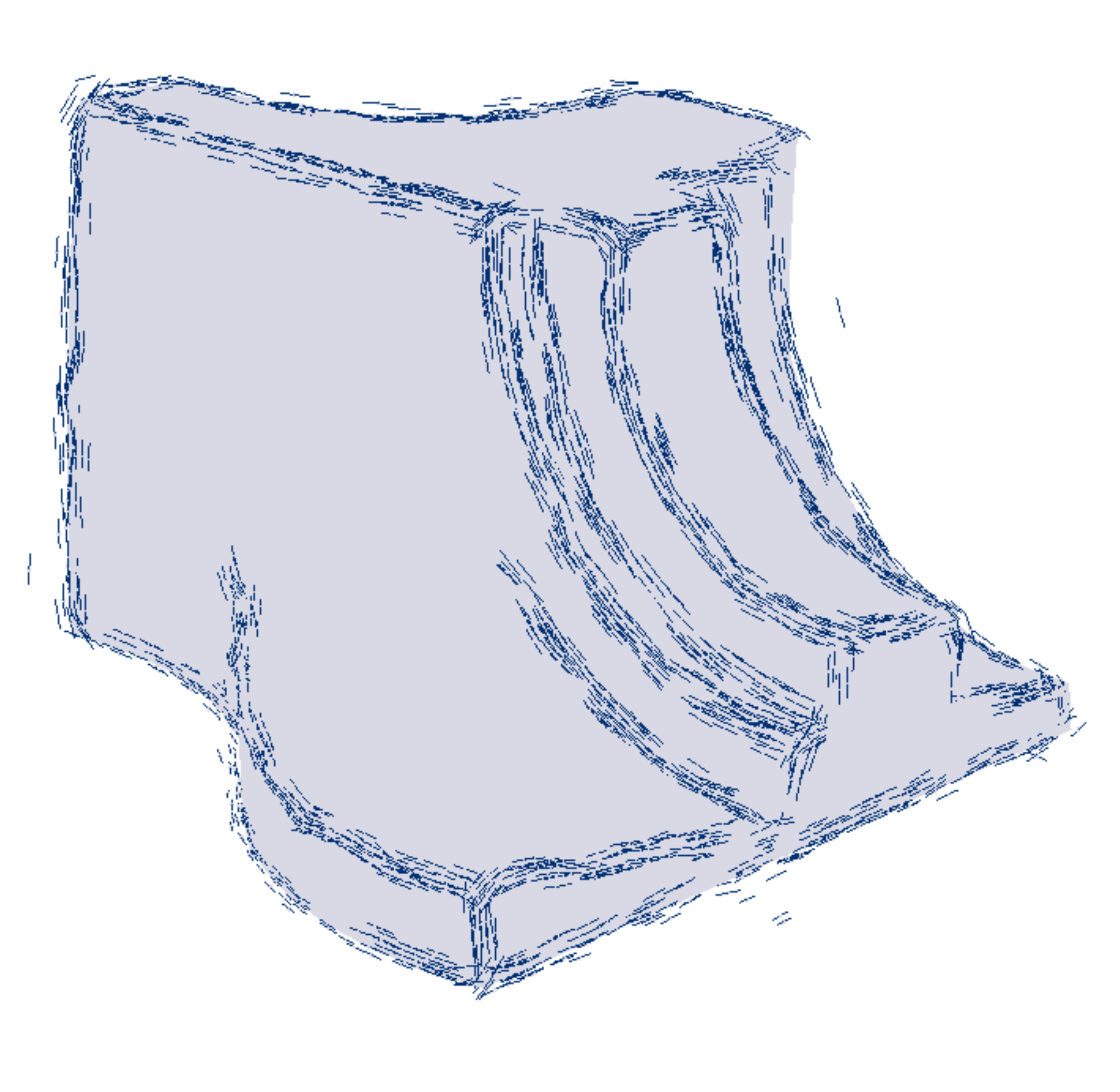}&
\includegraphics[width=.28\linewidth]{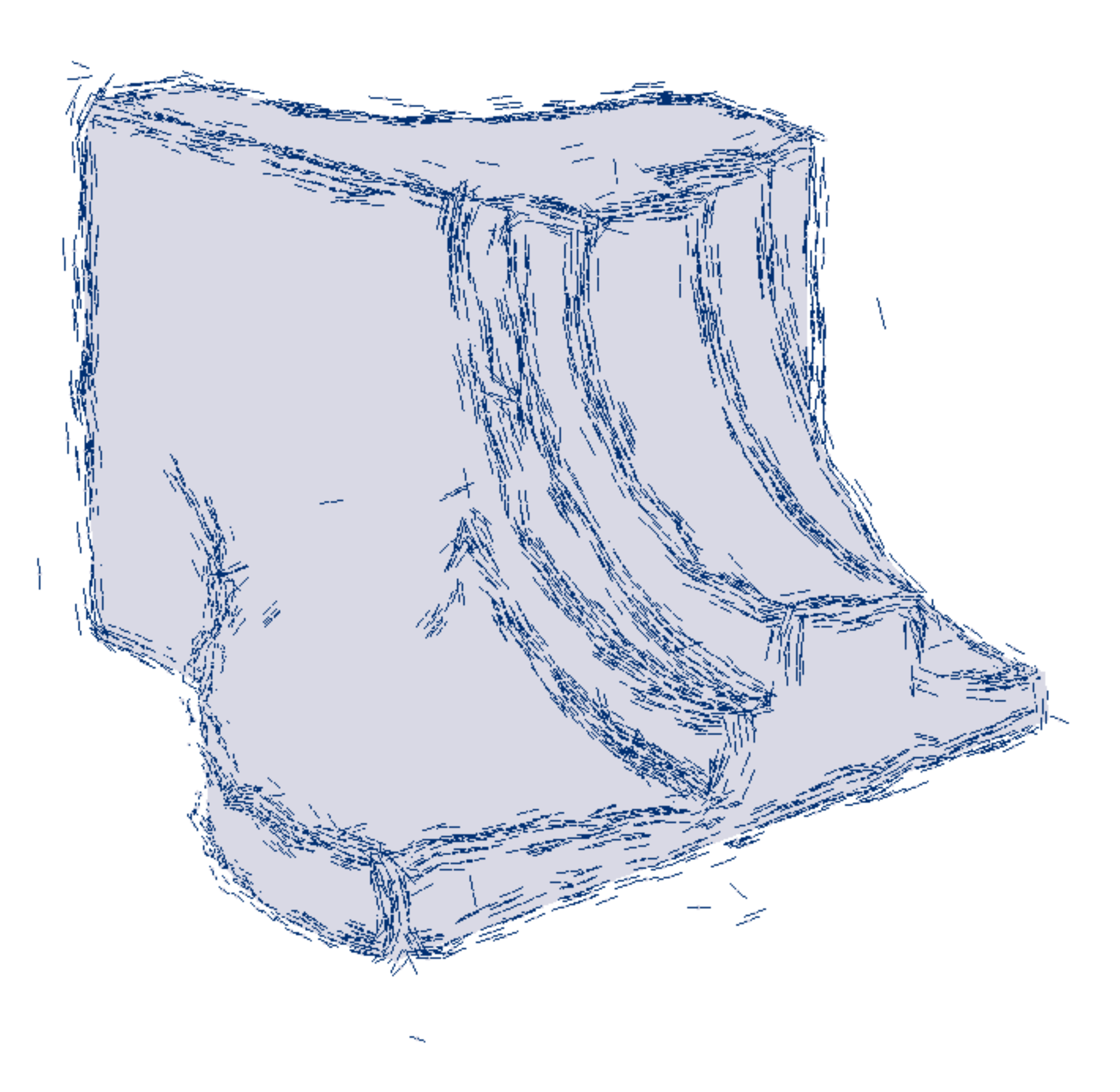}\\
$d^{\mathrm{w}}_{P,k}$-VCM &
$d^{\mathrm{w}}_{P,k}$-VCM &
$d^{\mathrm{w}}_{P,k}$-VCM \\
 ~&
~&
~\\
\includegraphics[width=.28\linewidth]{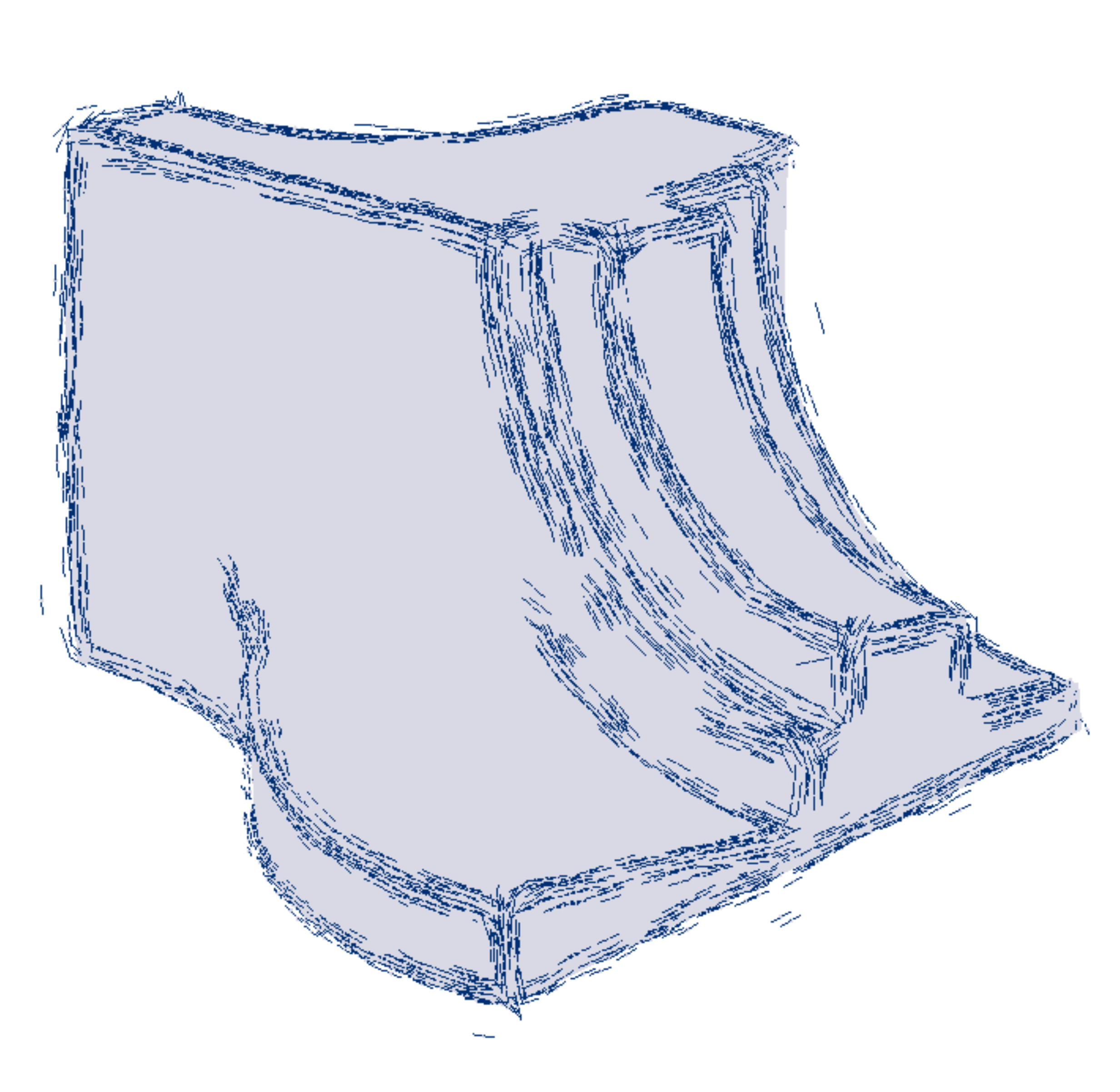}&
\includegraphics[width=.28\linewidth]{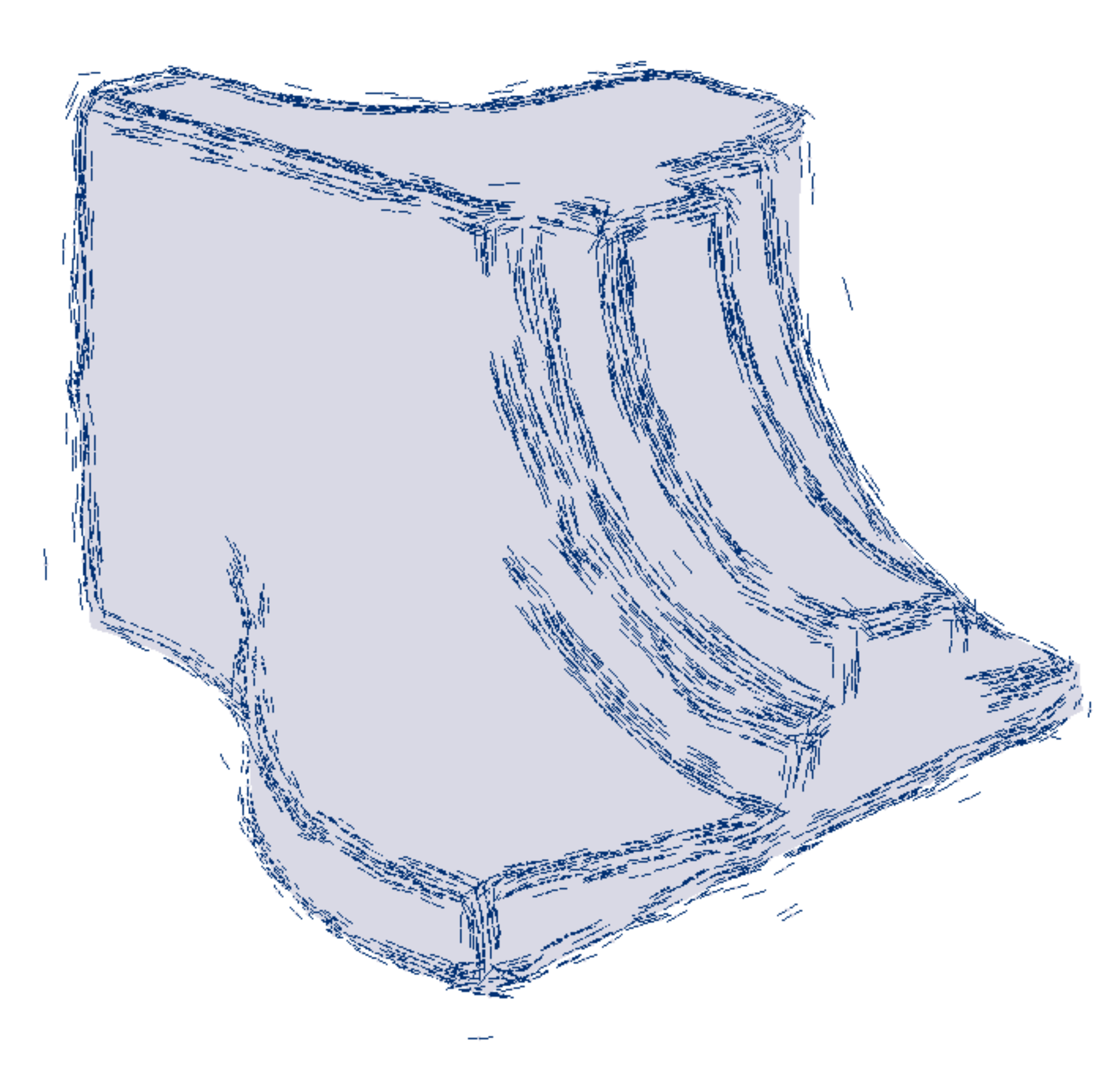}&
\includegraphics[width=.28\linewidth]{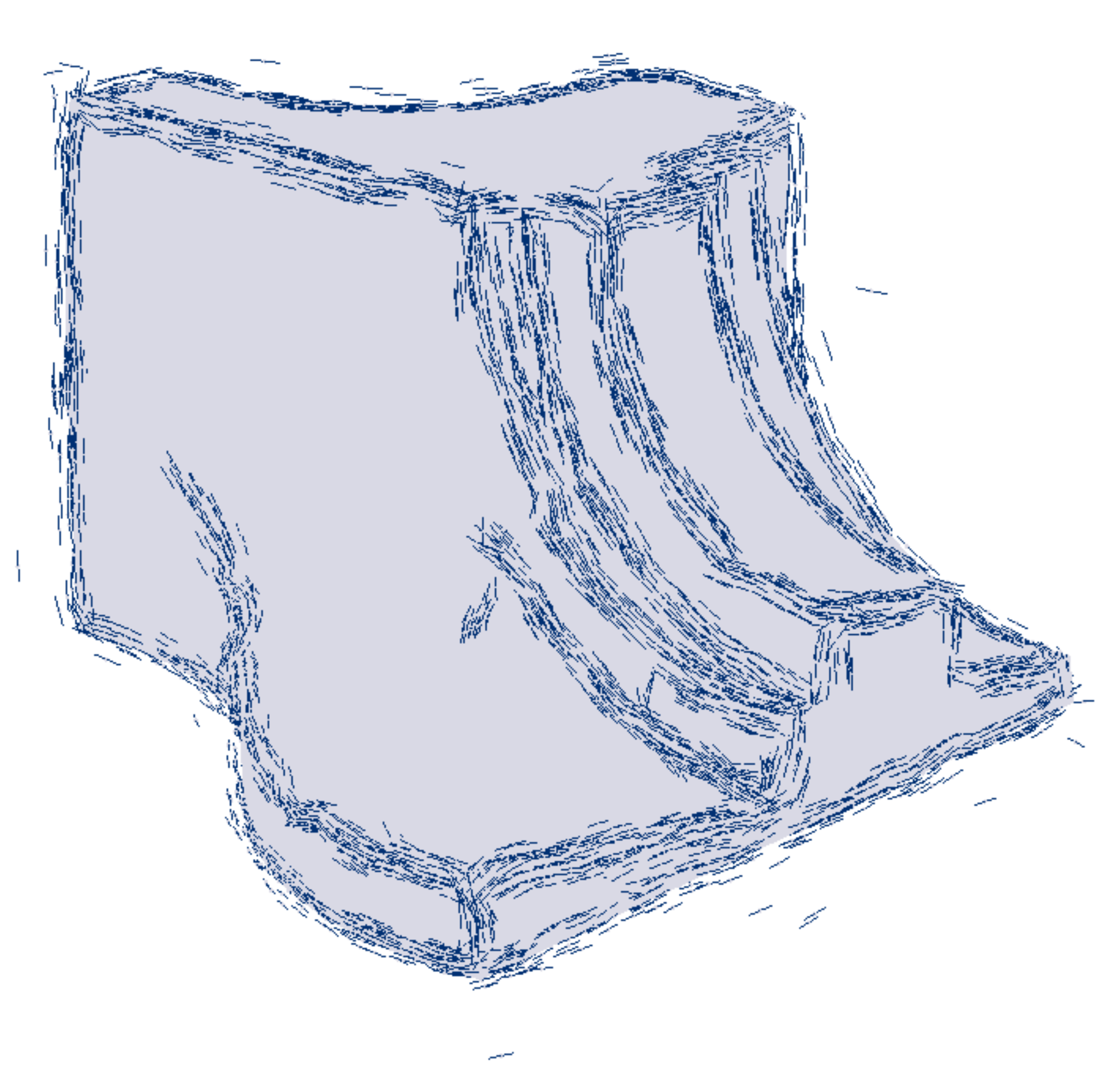}\\
$d^{\mathrm{m}}_{P,k}$-VCM &
$d^{\mathrm{m}}_{P,k}$-VCM &
$d^{\mathrm{m}}_{P,k}$-VCM \\
 ~&
~&
~\\
\includegraphics[width=.28\linewidth]{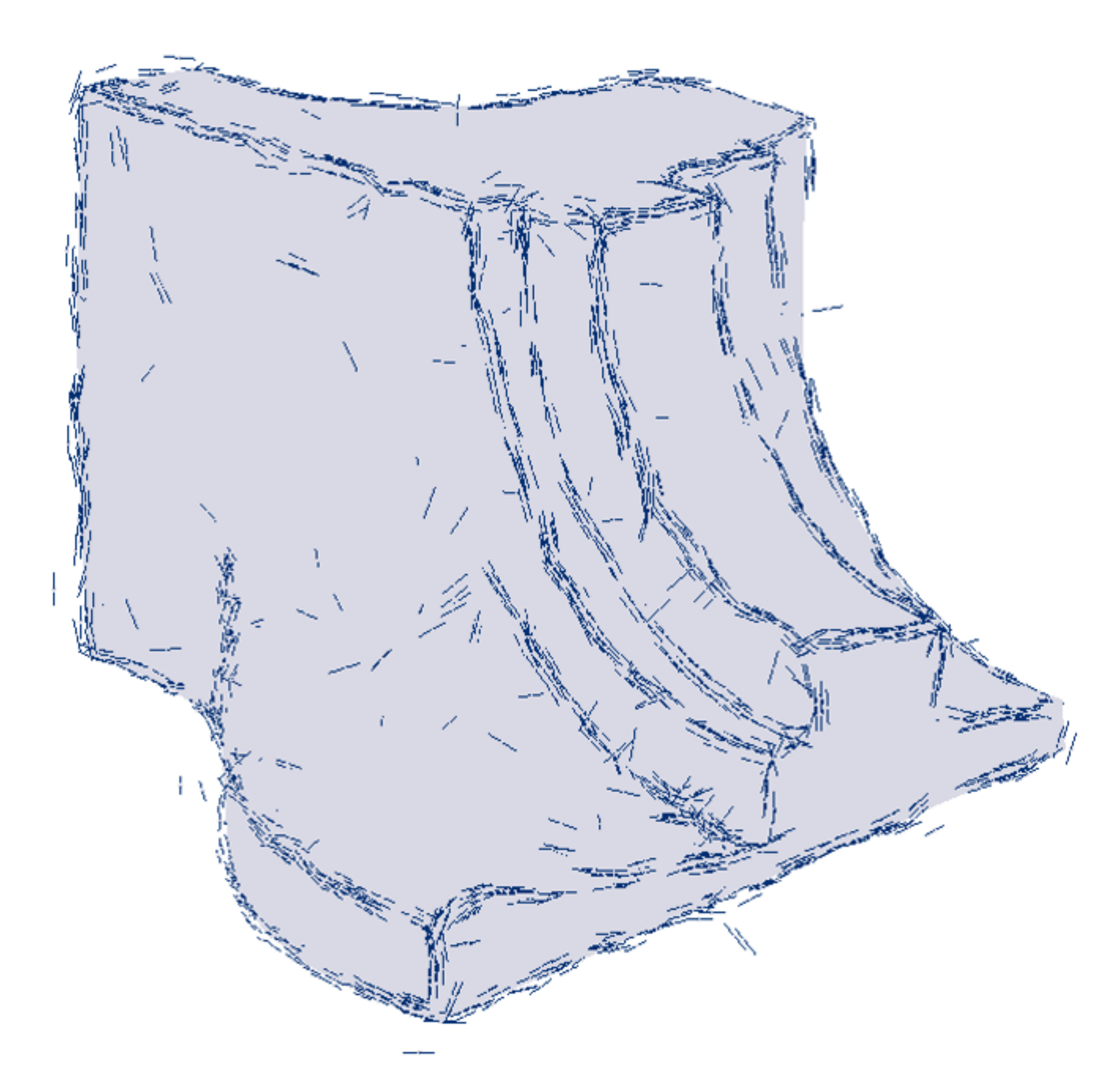}&
\includegraphics[width=.28\linewidth]{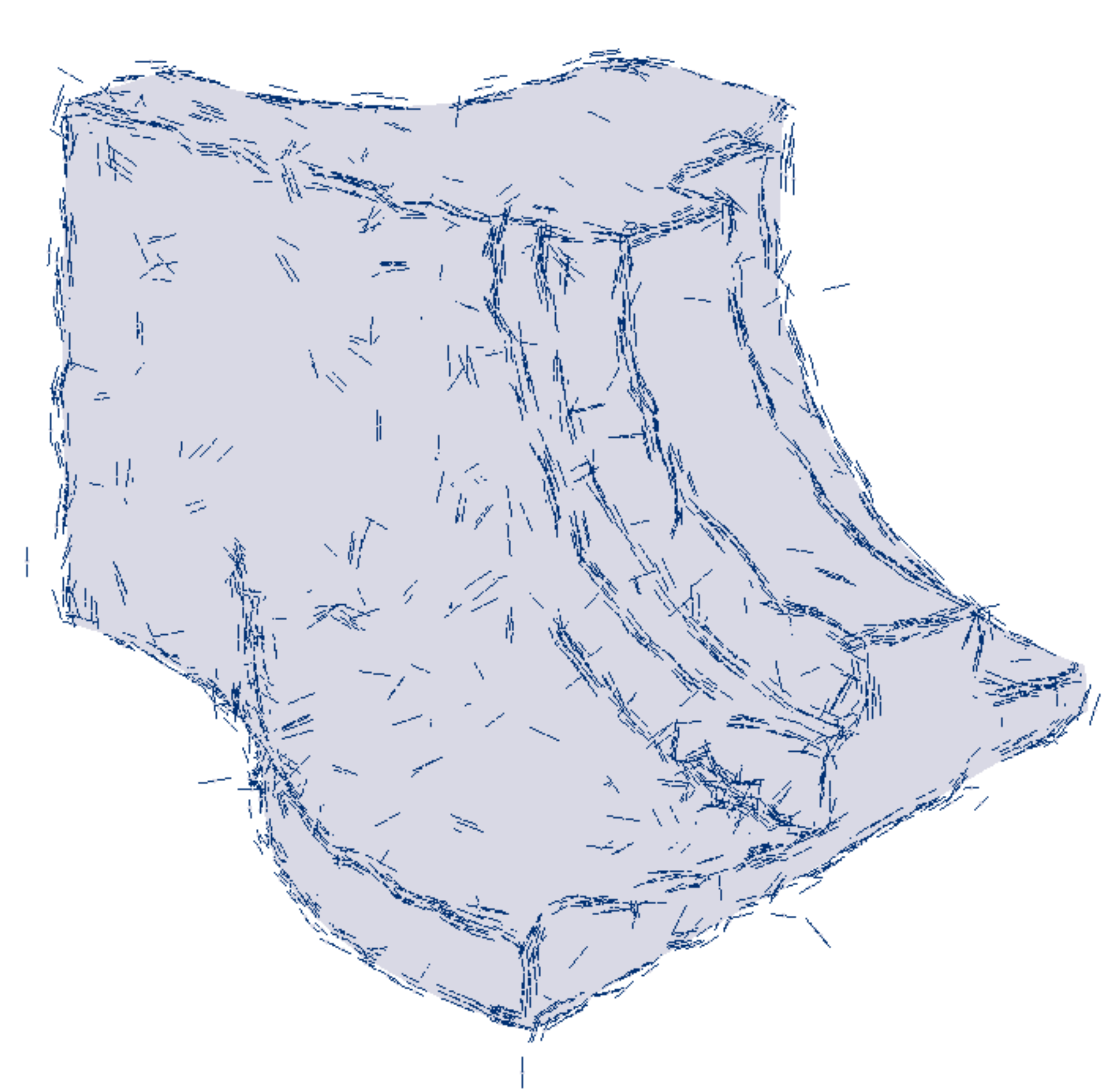}&
\includegraphics[width=.28\linewidth]{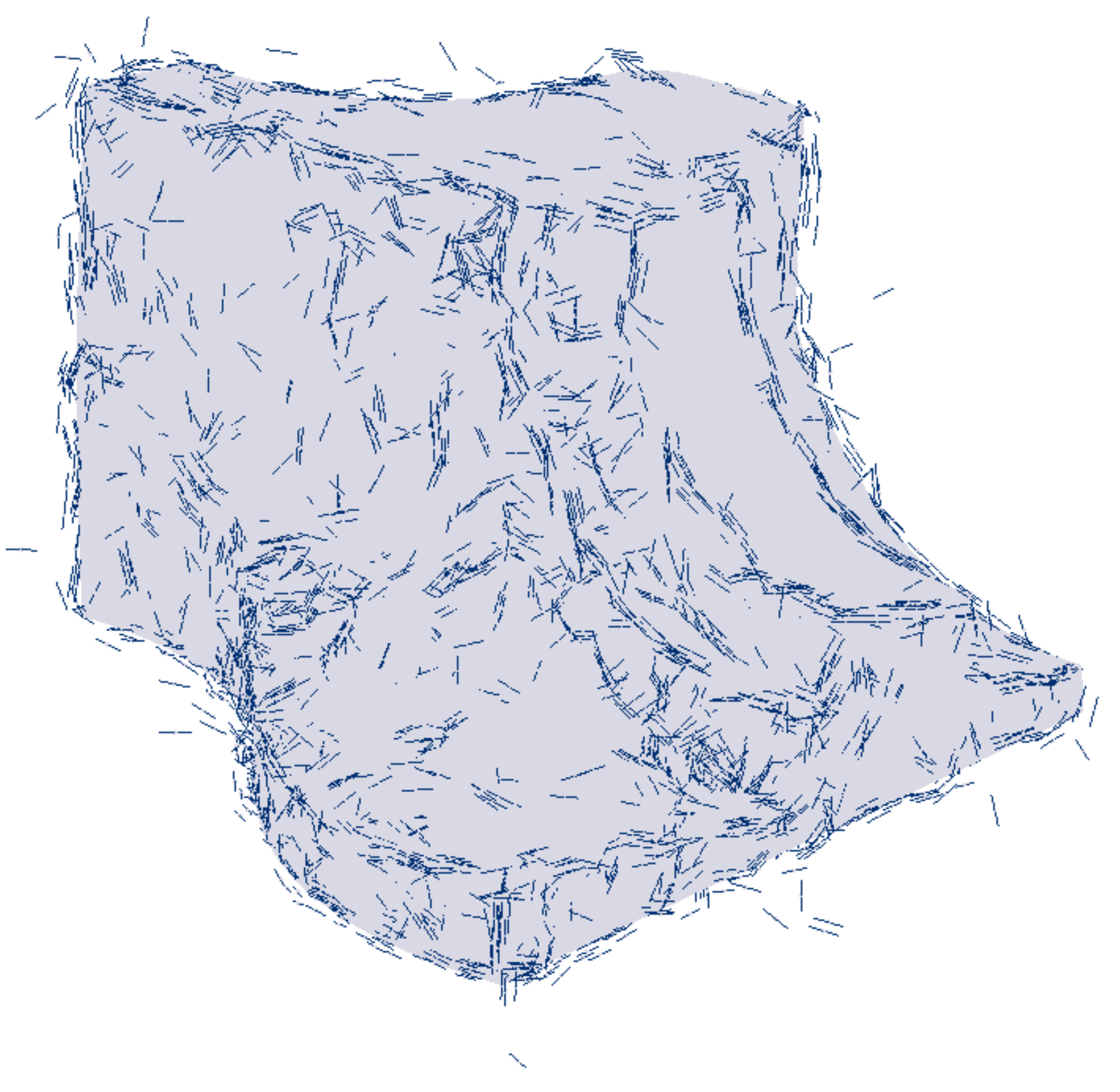}\\
VCM &
VCM &
VCM \\

\end{tabular}
\caption{Edge detection on a noisy fandisk of diameter $\diamnoise$: $98\%$ of the points are moved at a distance at most $0.02\diamnoise$ (Input 1), $0.035\diamnoise$ (Input 2) or $0.05\diamnoise$ (Input 3), $2\%$ are moved at a distance at most $0.1\diamnoise$ (Input 1), $0.175\diamnoise$ (Input 2) or $0.25\diamnoise$ (Input 3). All the results are computed with the parameters $R=0.06\diamnoise$ and $r=0.02\diamnoise$. The $d^{\mathrm{w}}_{P,k}$-VCM and the $d^{\mathrm{m}}_{P,k}$-VCM are computed with $k=30$.}
\label{fan}
\end{figure}

We also detect sharp features, using the same criteria as in \cite{vcm}: we say that a point $p$ belongs to a sharp feature if 
$\lambda_1 / (\lambda_0+\lambda_1+\lambda_2) \geq T$, for a given threshold $T$. 
Figure~\ref{fan} shows results on fandisk data with increasing noise (Hausdorff+outliers). It is clear that both witnessed-k-distance VCM and median-k-distance VCM are much less sensitive to outliers than VCM. More precisely our methods classify an outlier as a feature only if it is close to a sharp feature, while the VCM fails. This behavior is reinforced when the noise level increases.

\begin{figure}[!h]
\begin{tabular}{cccc}
\includegraphics[width=.24\linewidth]{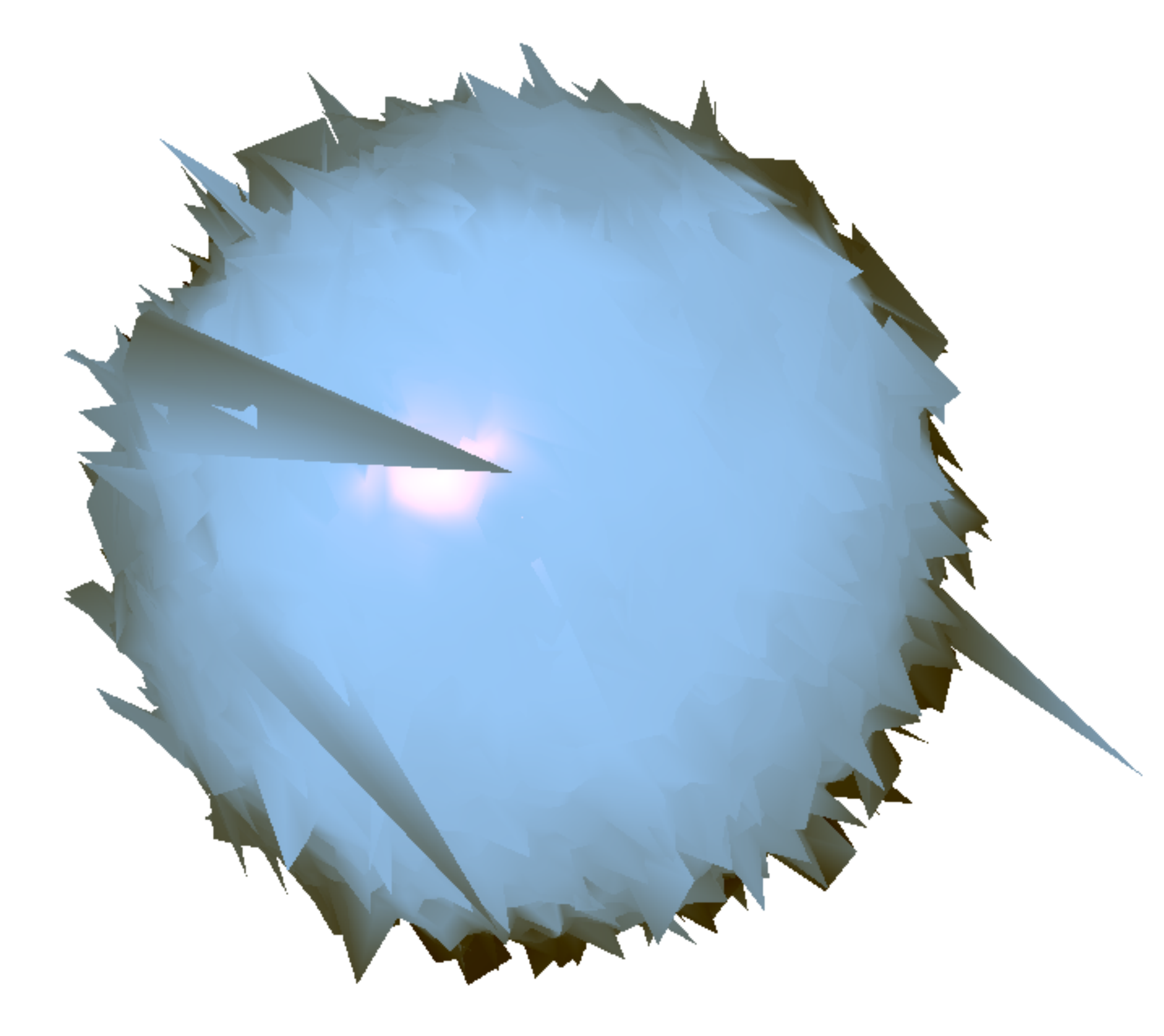}&
\includegraphics[width=.23\linewidth]{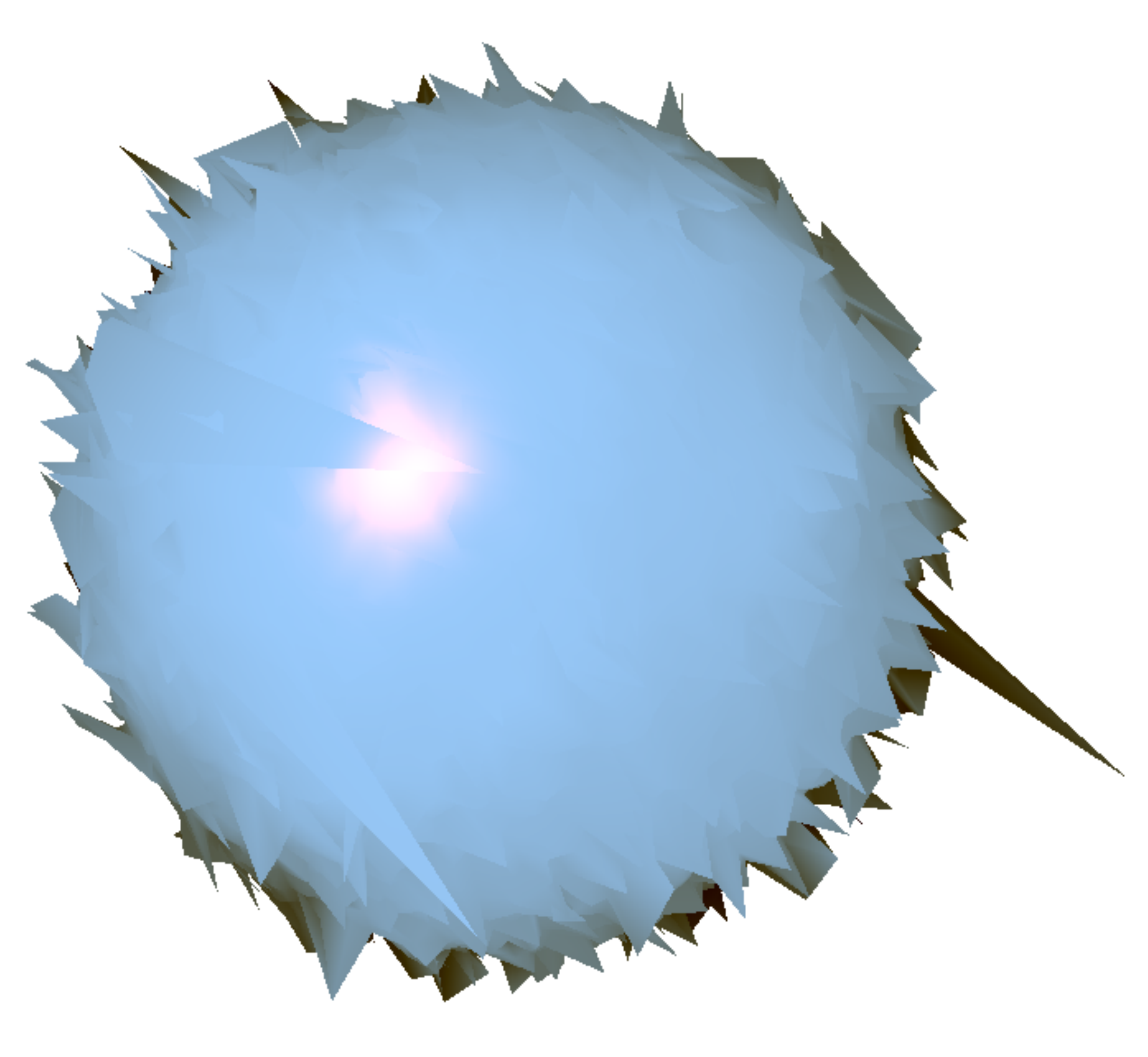}&
\includegraphics[width=.16\linewidth]{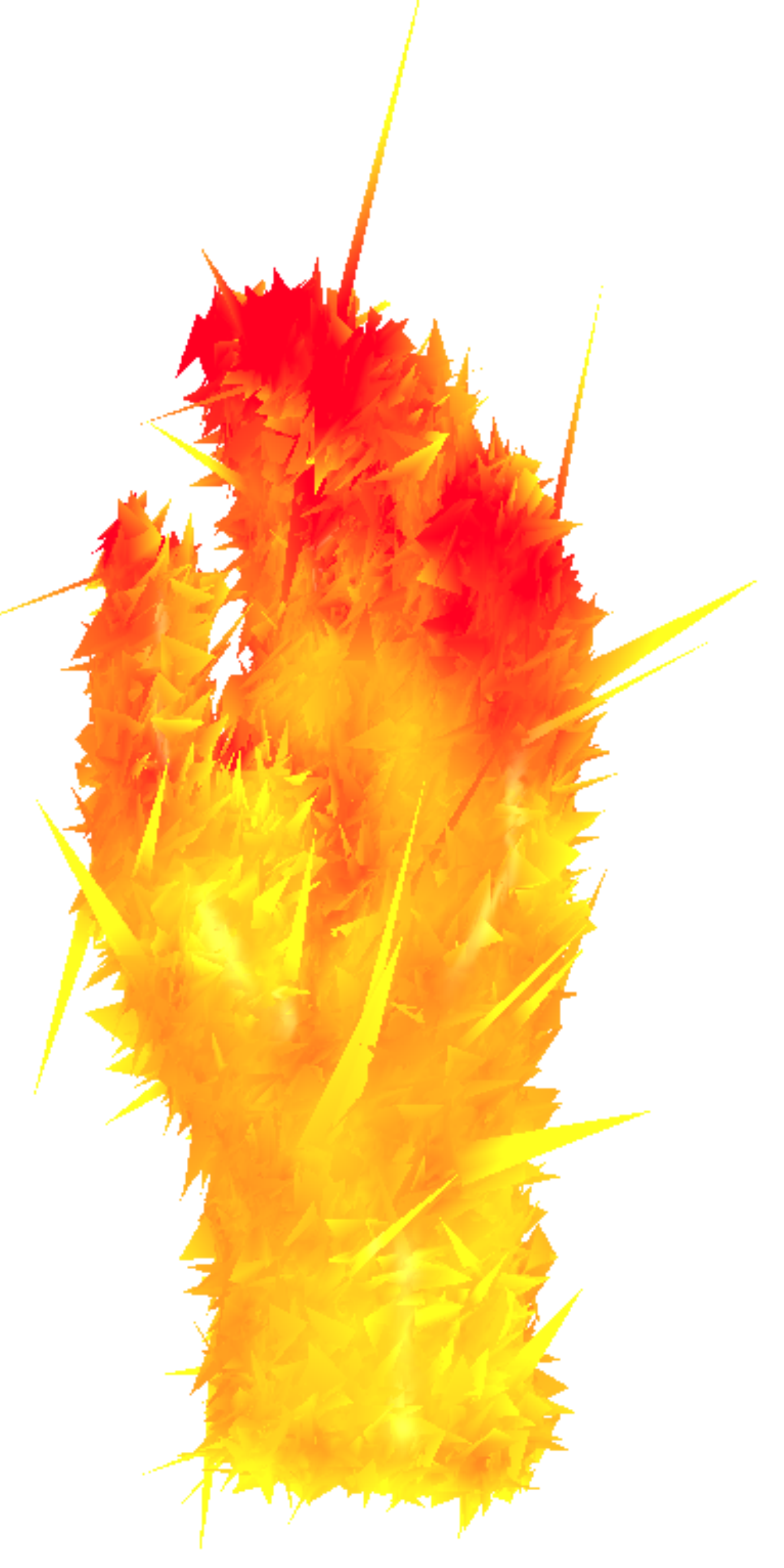}&
\includegraphics[width=.17\linewidth]{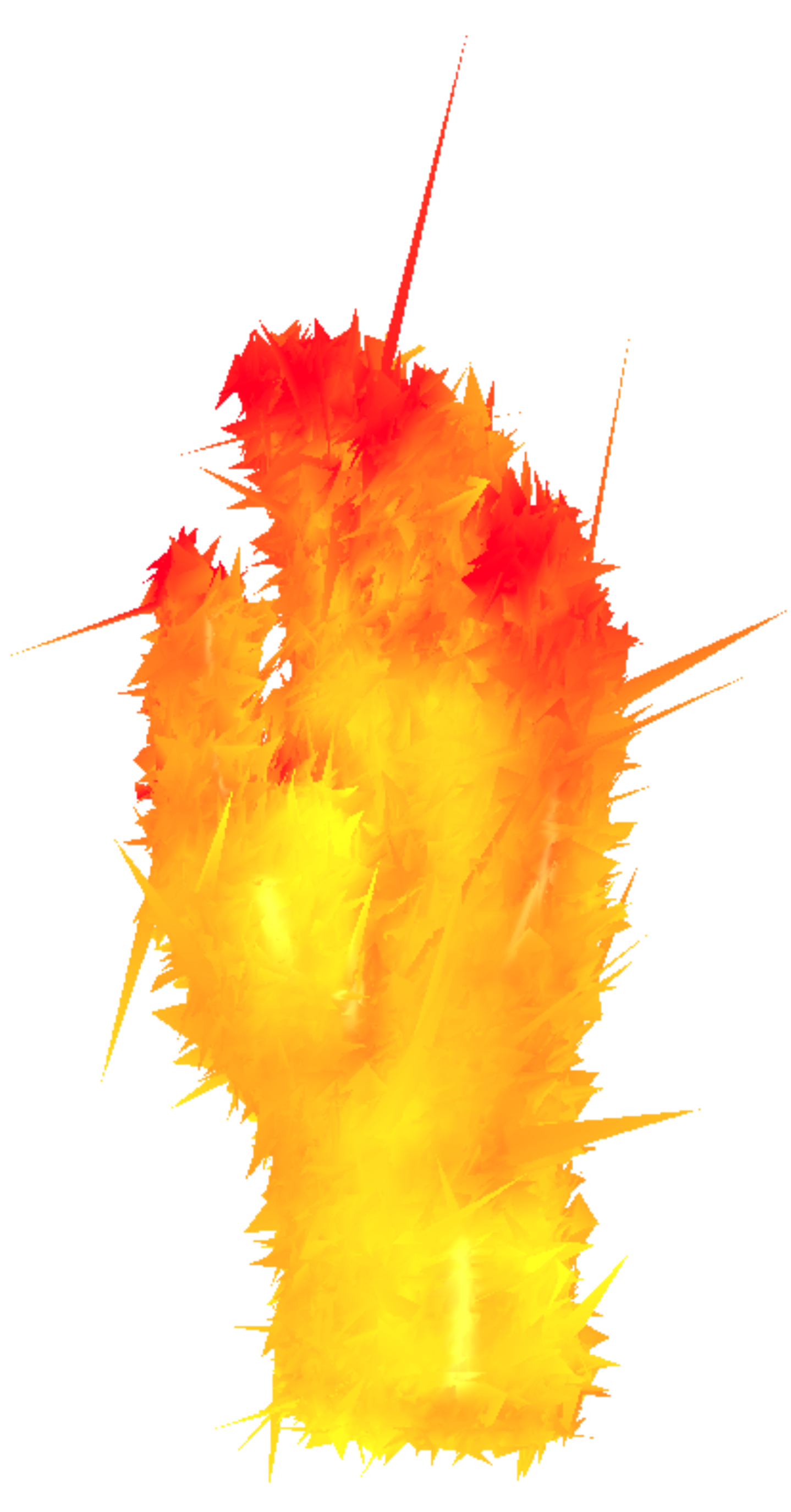}\\
$d^{\mathrm{w}}_{P,k}$-VCM & $d^{\mathrm{m}}_{P,k}$-VCM & $d^{\mathrm{w}}_{P,k}$-VCM & $d^{\mathrm{m}}_{P,k}$-VCM
\end{tabular}
\caption{Comparison between $d^{\mathrm{w}}_{P,k}$-VCM and $d^{\mathrm{m}}_{P,k}$-VCM (for the same parameters $R=0.04\diamnoise$, $r=0.02\diamnoise$, $k=30$ where $\diamnoise$ is the diameter of the original shape). Input datas contain outliers ($99.9\%$ of the points are moved at a distance at most $0.1\diamnoise$ and $0.1\%$ are moved at a distance between $0.1\diamnoise$ and $1\diamnoise$). Rendering illustrates the normal estimation quality on the ellipsoid data. The mean absolute curvature is displayed on the ``hand'' model.}
\label{out-estimation-comparison}
\end{figure}

\begin{figure}[!h]
\begin{tabular}{cc}
\includegraphics[width=.32\linewidth]{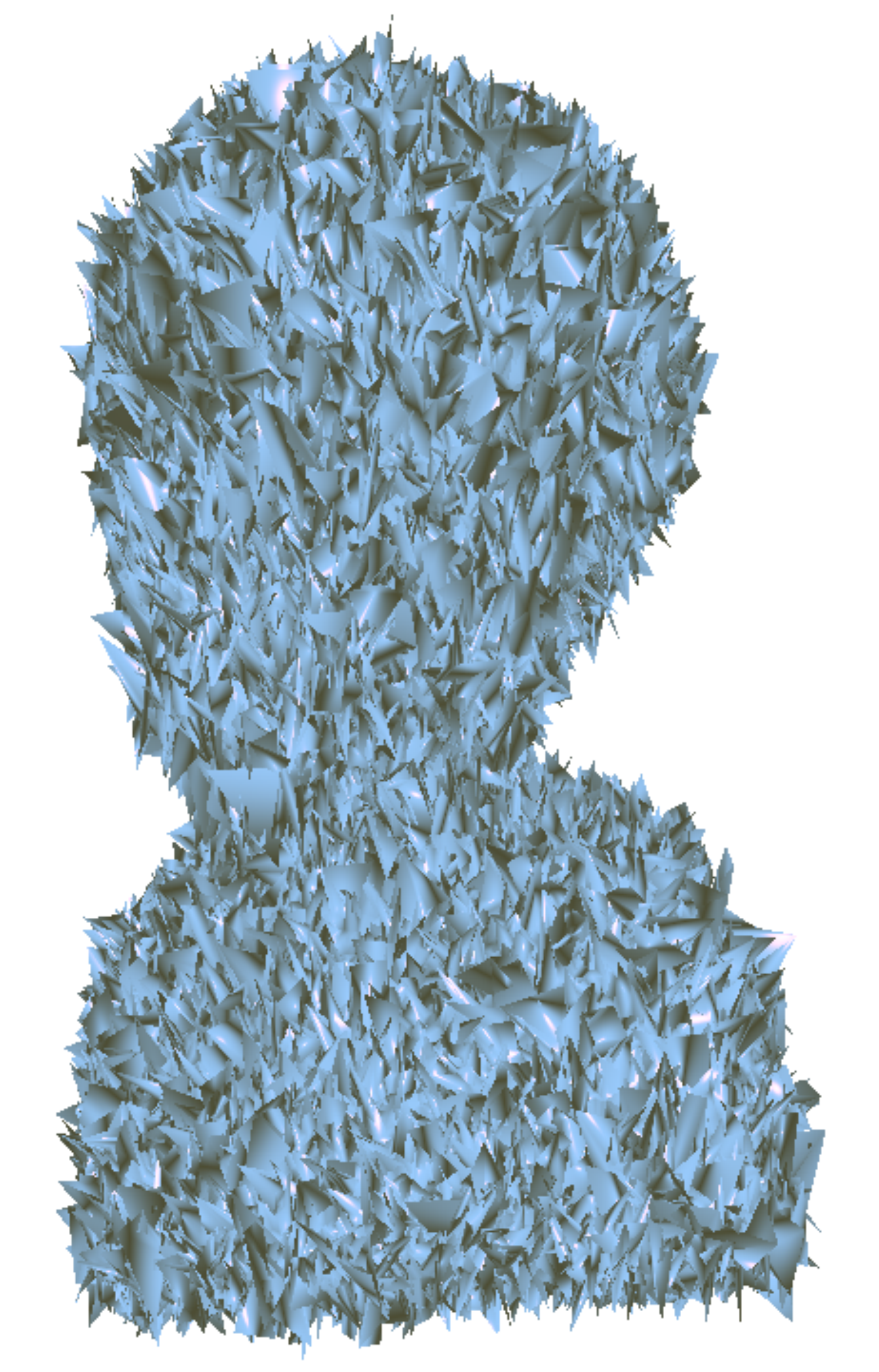}\hspace{0.2cm}&
\includegraphics[width=.6\linewidth]{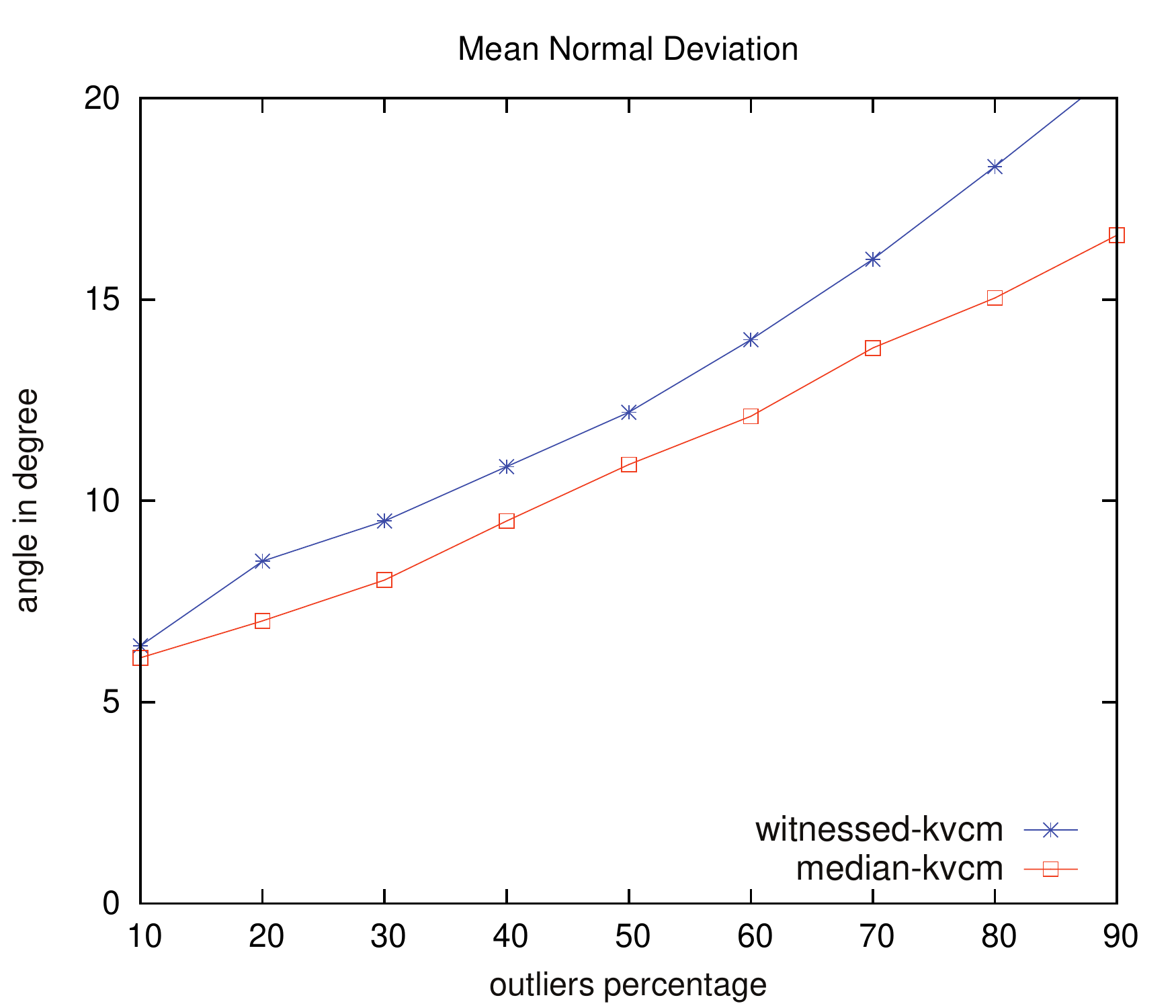}\\
\end{tabular}
\caption{Comparison between $d^{\mathrm{w}}_{P,k}$-VCM and $d^{\mathrm{m}}_{P,k}$-VCM (for the same parameters $R=0.1\diamnoise$, $r=0.1\diamnoise$, $k=30$, where $\diamnoise$ is the diameter of the original shape). Hausdorff noise is $0.2\diamnoise$. The percentage of ``outliers" corresponds to the percentage of points that have been moved at a distance in-between $0.05\diamnoise$ and $0.1\diamnoise$. Left: input mesh with $50\%$ of outliers. Right: evolution of the angle error when the proportion of outliers increases.} 
\label{normal-comparison}
\end{figure}

\subsection{Comparison with the $d^{\mathrm{m}}_{P,k}$-VCM}
Figures \ref{out-estimation-comparison} and \ref{normal-comparison}
indicate empirically that the $d^{\mathrm{m}}_{P,k}$-VCM gives a
slightly better estimation of the normal vector and absolute mean
curvature than the $d^{\mathrm{m}}_{P,k}$-VCM in the presence of
outliers. 
It appears in Figure~\ref{fan} that the median-k-distance and the witnessed-k-distance give very similar results. One may also remark that at very few points, the median-k-distance gives slightly better results.


\bibliographystyle{amsplain}
\bibliography{biblio}

\end{document}